\documentclass[11pt,a4paper]{article}

\usepackage[a4paper,hmargin=2.8cm,vmargin=3cm]{geometry}
\usepackage[utf8x]{inputenc}
\usepackage{amsmath,amsthm,amsfonts,amssymb}
\usepackage{authblk}
\usepackage[bookmarksnumbered=true]{hyperref}
\usepackage{dsfont}
\usepackage{graphicx}
\usepackage{tikz-cd}
\usepackage[nottoc]{tocbibind}
\usetikzlibrary{positioning,calc,fadings,decorations.pathreplacing,shadings}

\newcommand\pgfmathsinandcos[3]{%
  \pgfmathsetmacro#1{sin(#3)}%
  \pgfmathsetmacro#2{cos(#3)}%
}
\newcommand\LongitudePlane[3][current plane]{%
  \pgfmathsinandcos\sinEl\cosEl{#2} 
  \pgfmathsinandcos\sint\cost{#3} 
  \tikzset{#1/.style={cm={\cost,\sint*\sinEl,0,\cosEl,(0,0)}}}
}

\newcommand\LatitudePlane[3][current plane]{%
  \pgfmathsinandcos\sinEl\cosEl{#2} 
  \pgfmathsinandcos\sint\cost{#3} 
  \pgfmathsetmacro\yshift{\RadiusSphere*\cosEl*\sint}
  \tikzset{#1/.style={cm={\cost,0,0,\cost*\sinEl,(0,\yshift)}}} %
}

\newcommand\DrawLongitudeArc[4][black]{
  \LongitudePlane{\angEl}{#2}
  \tikzset{current plane/.prefix style={scale=1}}
  \pgfmathsetmacro\angVis{atan(sin(#2)*cos(\angEl)/sin(\angEl))} %
  \pgfmathsetmacro\angA{mod(max(\angVis,#3),360)} %
  \pgfmathsetmacro\angB{mod(min(\angVis+180,#4),360} %
  \draw[current plane,#1]  (\angA:\RadiusSphere) arc (\angA:\angB:\RadiusSphere);
}%
\newcommand\DrawLatitudeCircle[2][1]{
  \LatitudePlane{\angEl}{#2}
  \tikzset{current plane/.prefix style={scale=1}}
  \pgfmathsetmacro\sinVis{sin(#2)/cos(#2)*sin(\angEl)/cos(\angEl)}
  \pgfmathsetmacro\angVis{asin(min(1,max(\sinVis,-1)))}
  \draw[current plane] (\angVis:\RadiusSphere) arc (\angVis:-\angVis-180:\RadiusSphere);
}

\newcommand\DrawLatitudeArc[4][black]{
  \LatitudePlane{\angEl}{#2}
  \tikzset{current plane/.prefix style={scale=1}}
  \pgfmathsetmacro\sinVis{sin(#2)/cos(#2)*sin(\angEl)/cos(\angEl)}
  \pgfmathsetmacro\angVis{asin(min(1,max(\sinVis,-1)))}
  \pgfmathsetmacro\angA{max(min(\angVis,#3),-\angVis-180)} %
  \pgfmathsetmacro\angB{min(\angVis,#4)} %

  \draw[current plane,#1] (\angA:\RadiusSphere) arc (\angA:\angB:\RadiusSphere);
}

\newtheorem{thm}{Theorem}[section]

\newtheorem{prp}[thm]{Proposition}
\newtheorem{lem}[thm]{Lemma}

\theoremstyle{remark}

\newtheorem*{rems}{Remarks}

\renewcommand{\theequation}{\thesection.\arabic{equation}}

\numberwithin{equation}{section}

\newcommand{\cercle}[4]{
\node[circle,inner sep=0,minimum size={2*#2}](a) at (#1) {};
\draw (a.#3) arc (#3:{#3+#4}:#2);}

\newcommand{\boldcercle}[4]{
\node[circle,inner sep=0,minimum size={2*#2}](a) at (#1) {};
\draw[very thick] (a.#3) arc (#3:{#3+#4}:#2);}

\DeclareMathOperator{\lcm}{lcm}

\newcommand{\dbtilde}[1]{\tilde{\raisebox{0pt}[0.85\height]{$\tilde{#1}$}}}

\newcommand{\fock}{\mathcal{F}}		
\newcommand{\di}{{\textnormal{d}}}		
\newcommand{\rfrak}{\mathfrak{r}}

\newcommand{\Tbb}{\mathbb{T}}
\newcommand{\W}{W}
\newcommand{\Wt}{\tilde W}
\newcommand{\D}{D}

\newcommand{\Ecal}{\mathcal{E}}
\newcommand{\Hbb}{\mathbb{H}}

\newcommand{\Ncal}{\mathcal{N}}		
\newcommand{\Vcal}{\mathcal{V}}		
\newcommand{\Hcal}{\mathcal{H}}		
\newcommand{\Ical}{\mathcal{I}}
\newcommand{\Ikp}{\Ical_{k}^{+}}
\newcommand{\Ikm}{\Ical_{k}^{-}}
\newcommand{\Ik}{\Ical_{k}}
\newcommand{\Il}{\Ical_{l}}

\newcommand{\ik}{I}

\newcommand{\lebesgue}{\mu}

\newcommand{\Sbb}{\mathbb{S}}

\newcommand{\Ocal}{\mathcal{O}}		
\newcommand{\hc}{\textnormal{h.c.}}		

\newcommand{\cc}[1]{\overline{#1}}	
\newcommand{\Rbb}{\mathbb{R}}		
\newcommand{\Cbb}{\mathbb{C}}		
\newcommand{\Nbb}{\mathbb{N}}		
\newcommand{\Zbb}{\mathbb{Z}}

\newcommand{\Zbbn}{\Zbb^3 \setminus \{0\}}
\renewcommand{\Re}{\operatorname{Re}} 	
\newcommand{\id}{\mathbb{I}}
\newcommand{\norm}[1]{\lVert#1\rVert}	
\newcommand{\tr}{\operatorname{tr}}
\newcommand{\HS}{_{\textnormal{HS}}}
\newcommand{\TR}{_{\textnormal{tr}}}
\newcommand{\OP}{_{\textnormal{op}}}
\newcommand{\sgn}{\operatorname{sgn}}
\newcommand{\tagg}[1]{ \stepcounter{equation} \tag{\theequation} \label{eq:#1} } 

\newcommand{\Efrak}{\mathfrak{E}}
\newcommand{\efrak}{\mathfrak{e}}
\newcommand{\nfrak}{\mathfrak{n}}

\newcommand{\north}{\Gamma^{\textnormal{nor}}}

\newcommand{\diag}{\operatorname{diag}}
\newcommand{\diam}{\operatorname{diam}}
\newcommand{\supp}{\operatorname{supp}}

\newcommand{\patch}[3]{{\substack{#1\in B_\textnormal{F}^c\cap B_#3\\#2 \in B_\textnormal{F}\cap B_#3}}}

\title{Optimal Upper Bound for the Correlation Energy\\of a Fermi Gas in the Mean-Field Regime}

\author[1,*]{Niels Benedikter}
\author[2]{Phan Th\`anh Nam}
\author[3]{Marcello Porta} 
\author[4]{Benjamin Schlein}
\author[1]{Robert Seiringer}
\affil[1]{IST Austria, Am Campus 1, 3400 Klosterneuburg, Austria}
\affil[2]{LMU Munich, Department of Mathematics, Theresienstra{\ss}e 39, 80333 M\"unchen, Germany}
\affil[3]{University of T\"ubingen, Department of Mathematics, Auf der Morgenstelle 10, 72076 T\"ubingen, Germany}
\affil[4]{Institute of Mathematics, University of Zurich, Winterthurerstrasse 190, 8057 Zurich, Switzerland}
\affil[*]{corresponding author, \href{mailto:niels.benedikter@ist.ac.at}{niels.benedikter@ist.ac.at}}

\begin{document}

\maketitle

\begin{abstract} While Hartree--Fock theory is well established as a fundamental approximation for interacting fermions, it has been unclear how to describe corrections to it due to many-body correlations. In this paper we start from the Hartree--Fock state given by plane waves and introduce collective particle--hole pair excitations. These pairs can be approximately described by a bosonic quadratic Hamiltonian. We use Bogoliubov theory to construct a trial state yielding a rigorous Gell-Mann--Brueckner--type upper bound to the ground state energy. Our result justifies the random-phase approximation in the mean-field scaling regime, for repulsive, regular interaction potentials.
 \end{abstract}

\tableofcontents

\section{Introduction}
While Hartree--Fock theory describes some aspects of interacting fermionic systems very well, it utterly fails at others. The best known example is that Hartree--Fock theory predicts a vanishing density of states at the Fermi momentum, which is incompatible with measurements of the conductivity and specific heat in metals \cite{GV05}. It is therefore important to develop a rigorous understanding of many-body corrections to Hartree--Fock theory. The simplest theory of many-body correlations is the \emph{random-phase approximation (RPA)}.

In this paper we show that the RPA is mathematically rigorous, insofar as the RPA correlation energy provides an upper bound on the ground state energy of interacting fermions in the mean-field scaling regime. Our approach also sheds some light on the emergence of bosonic collective modes in the Fermi gas, described by an effective quadratic Hamiltonian.

\medskip

We consider a system of $N\gg 1$ fermionic particles with mass $m >0$ in the torus $\Tbb^3 = \Rbb^3/(2\pi\Zbb^3)$, interacting via a two-body potential $V$, in the mean-field scaling regime. Setting
\[\hbar = N^{-1/3}\,,\]
the Hamiltonian is defined as
\[
H_N := -\frac{\hbar^2}{2m} \sum_{i=1}^N \Delta_{x_i} + \frac{1}{N}\sum_{1\leq i<j \leq N} V(x_i - x_j)\;,
\]
and acts on the Hilbert space $L^2_a\left((\Tbb^3)^N\right)$ consisting of square-integrable functions that are anti-symmetric under permutations of the $N$ arguments. For simplicity we consider only the spinless case\footnote{For the analogous model of fermions with spin we can repeat our construction of an upper bound for the correlation energy treating the spin states as independent. In general of course spin gives rise to many intricate phenomena such as formation of spin density waves, in fact already on the level of Hartree--Fock theory \cite{GL18}.}. The choice of $\hbar = N^{-1/3}$ and coupling constant $1/N$ defines the fermionic mean-field regime: it guarantees that both kinetic and potential energies are of order $N$, as $N \to \infty$ (see \cite{BPS16} for a detailed introduction).

The ground state energy of the system is defined as
\begin{equation}
\label{eq:gse-HN}E_N := \inf_{\substack{\psi \in L^2_a\left((\Tbb^3)^N\right)\\\norm{\psi} = 1}} \langle \psi, H_N \psi \rangle\;.\end{equation}
In Hartree--Fock theory, one restricts the attention to Slater determinants
\[
\psi_\text{Slater} (x_1, \dots , x_N) = \frac{1}{\sqrt{N!}} \sum_{\sigma \in S_N} \sgn(\sigma) f_1 (x_{\sigma(1)}) f_2 (x_{\sigma(2)}) \dots f_N (x_{\sigma(N)})
\]
with  $\{f_j\}_{j=1}^N$ an orthonormal set in $L^2 (\Rbb^3)$. Slater determinants are an example of quasi-free states: all reduced density matrices can be expressed in terms of the one-particle reduced density matrix $\omega := N\tr_{2,\ldots, N} |\psi\rangle \langle \psi |$. For a Slater determinant, one has $\omega = \sum_{j=1}^N \lvert f_j \rangle \langle f_j\rvert$. In particular, the energy of a Slater determinant is given by the Hartree--Fock energy functional, depending only on $\omega$:
\[
\begin{split}
& \mathcal{E}_\text{HF}(\omega) := \langle \psi_\text{Slater}, H_N \psi_\text{Slater} \rangle \\ & = \tr \, \left(\frac{- \hbar^2}{2m} \Delta  \omega \right)+ \frac{1}{2N} \int \di x \di y V(x-y) \omega (x,x) \omega (y,y) 
- \frac{1}{2N} \int \di x \di y V(x-y) \lvert \omega (x,y) \rvert^2\,.\end{split}
\]
(The first two summands are typically of order $N$ and called the kinetic and direct term, respectively; the third summand is typically of order $1$ and called the exchange term.)
Thus, minimizing $\mathcal{E}_\text{HF}(\omega)$ over all orthogonal projections $\omega$ with $\tr \, \omega = N$ gives an upper bound to the ground state energy $E_N$. Actually, it turns out that Hartree--Fock theory provides more than an upper bound for the ground state energy: the method developed in  \cite{Bac92, Bac93, GS94} for the jellium model can also be applied to show that in the present mean-field scaling the Hartree--Fock minimum agrees with the many-body ground state energy up to an error of size $o(1)$ for $N \to \infty$. Moreover, by projection of the time-dependent Schr\"odinger equation onto the manifold of quasi-free states  one obtains the time-dependent Hartree--Fock equation \cite{BSS18}, which was proven to effectively approximate the many-body evolution of mean-field fermionic systems \cite{BJP+16, BPS14a, BPS14, BPS14b, PRSS17, Saf18}. 

For $N$ \emph{non-interacting} particles on the torus, the ground state is given by the Slater determinant constructed from plane waves
\begin{equation}\label{eq:plane-wave} f_k (x) = (2\pi)^{-3/2} e^{i k \cdot x}, \quad k \in \Zbb^3,
\end{equation}
where the momenta $k_{1}, \ldots, k_{N}\in \Zbb^{3}$ are chosen to minimize the kinetic energy in a way compatible with the Pauli principle; i.\,e., by filling the Fermi ball, up to the Fermi momentum $k_\textnormal{F}$. The energy $E_\textnormal{F} := k_\textnormal{F}^2/(2m)$ is called the Fermi energy, and the sphere $k_F \Sbb^2$ of radius $k_F$ is called the Fermi surface. (We assume that $N$ is chosen so that this state is unique, no modes in the Fermi ball being left empty.) We shall denote by $\omega_{\textnormal{pw}}$ the reduced one-particle density matrix of this state, 
\[
\omega_{\textnormal{pw}} = \sum_{i=1}^{N} |f_{k_{i}} \rangle \langle f_{k_{i}}|\;.
\]
It turns out that this simple state is a stationary state of the Hartree--Fock energy functional even with interactions, and in our setting provides a good approximation to the minimum of the Hartree--Fock functional. The focus of the present paper is to quantify the effect of correlations in the true many-body ground state: in particular, we shall be interested in the \emph{correlation energy}, defined as the difference of the ground state energy and the Hartree--Fock energy of the plane wave state\footnote{This is the definition used by Gell-Mann and Brueckner \cite{GB57}. Some authors define the correlation energy with respect to the minimum of the Hartree--Fock functional instead. For the present translation invariant setting, it was recently proved \cite{GHL18} that the energy of the plane wave state and the minimal Hartree--Fock energy differ only by an exponentially small amount as $N\to\infty$. However, for systems that are not translation invariant, the ground state of non-interacting fermions will not even be a stationary point of the interacting Hartree--Fock functional. In this case it is important to take the true Hartree--Fock minimizer as reference point.}, $E_{N} - \mathcal{E}_\textnormal{HF}(\omega_\textnormal{pw})$.

The quest of calculating the correlation energy has been a driving force in the early development of theoretical condensed matter physics. Let us discuss the case of the jellium model: that is, fermions interacting via Coulomb repulsion, exposed to a neutralizing background charge on the torus, in the large volume limit. Let us consider the ground state energy per volume of the system, in the high density regime. As noticed already by Wigner \cite{Wig34} and Heisenberg \cite{Hei47}, the computation of the correlation energy is an intricate matter because perturbation theory with respect to the Coulomb potential becomes more and more infrared divergent at higher orders. It was however quickly understood that these divergences are an artefact of perturbation theory \cite{Mac50}; a partial resummation of the perturbative expansion allows to capture the effect of screening, that ultimately trades the infrared divergence for a $\rho \log \rho$ contribution ($\rho$ being the density) to the ground state energy. 

In their seminal work \cite{Pin53, BP53} Bohm and Pines related the screening of the Coulomb potential to an auxiliary bosonic mode called the plasmon, and coined the name ``random-phase approximation''; see also \cite{Gas61} for a reformulation of their result using Jastrow--type states. Gell-Mann and Brueckner showed that the RPA can be seen as a systematic resummation of the most divergent diagrams of perturbation theory \cite{GB57}, which has become the most popular point of view for physicists. Another interpretation of the RPA was given by Sawada et al.\ in \cite{Saw57,SBFB57} as an effective theory of approximately bosonic particle--hole pairs. A systematic mapping of particle--hole pairs to bosonic operators was introduced by Usui in \cite{Usu60} but does not lead to a quadratic Hamiltonian. (In Usui's approach there are parallels to bosonization in the Heisenberg model \cite{Dys56,Dys56a,HP40,CGS15}, which also gives rise to interesting problems in the calculation of higher order corrections to the free energy \cite{Ben17}.) Sawada's approach has been systematically related to perturbation theory in \cite{AP75}.  Sawada's effective Hamiltonian has proved useful for further investigations into diamagnetism and the Meissner effect \cite{Wen57}. While Sawada's concept of bosonic pairs is very elegant, it remained unclear which parameter makes the error of the bosonic approximation small. This was clarified many years later, highlighting the role of collective excitations delocalized over many particle--hole pairs \cite{CF94,CF95,FG97,FGM95,FGM95a,Hal94,HKMS94,HM93,KC96,KHS95,KS96,Lut79,Ng00}; the main idea being that collective excitations of pairs of fermions do not experience the Pauli exclusion principle if they involve many fermionic modes of which only few are occupied. 

Concerning rigorous works for the jellium model, the only available result for the correlation energy is the work of Graf and Solovej \cite{GS94}, which provided an upper and lower bound proportional to $\rho^{4/3 - \delta}$ for some $\delta > 0$. This bound has been obtained using correlation inequalities for the many-body interaction together with semiclassical methods. Unfortunately, this is still far from the expected $\rho \log \rho$ behavior: to improve on \cite{GS94}, new ideas are needed.

In the context of interacting fermions in the mean-field regime, the first rigorous result on the correlation energy has been recently obtained in \cite{HPR18}, for small interaction potentials, via upper and lower bounds matching at leading order. One has: 
\begin{equation}\label{eq:HPR18}
\lim_{N \to \infty} \frac{E_N-\Ecal_\textnormal{HF}(\omega_\textnormal{pw})}{\hbar} = - m \pi (1-\log(2)) \sum_{k \in \Zbb^3} \lvert k\rvert \hat{V}(k)^2 \big(1 + \Ocal(\hat V(k)) \big)\;.
\end{equation}
The strategy of \cite{HPR18} is based on a rigorous formulation of second order perturbation theory following \cite{HS02, Hai03, CH04, HHS05}, combined with methods developed in the context of many-body quantum dynamics \cite{BPS14a, BPS14b, BJP+16, PRSS17}. For larger interaction potentials however, this method is limited to a lower bound of the right order in $\hbar$ and $N$ but not capturing the precise value.
 
Here we shall provide a rigorous upper bound on the correlation energy, without any smallness assumption on the size of the potential. It improves on the upper bound of \cite{HPR18}, to which it reduces in the limit of small interactions. The method of the proof is inspired by a mapping of the particle--hole excitations around the Fermi surface to emergent bosonic degrees of freedom: this allows to estimate the correlation energy in terms of the ground state energy of a \emph{quadratic, bosonic Hamiltonian}. The expression we obtain, if formally extrapolated to the infinite volume limit, agrees with the Gell-Mann--Brueckner formula for the jellium model. 

Our method can be seen as a rigorous version of the \emph{Haldane--Luther bosonization} for interacting Fermi gases, a nonperturbative technique widely used in condensed matter physics; see \cite{Kop97} for a review. To our knowledge, this is the first time that this method is formulated in a mathematically rigorous setting. We believe that this method, possibly combined with \cite{HPR18}, will be crucial to rigorously understand the correlation energy for a large class of high density Fermi gases, including the jellium model.

\medskip

Correlation corrections to the ground state energy of interacting Bose gases have been studied to a much larger extent. Upper and lower bounds have been proven for the mean-field scaling regime in \cite{Sei11,GS13,DN14,LNSS15,Piz15,Piz15a}, for the jellium model in \cite{LS01,LS04,Sol06}, for the Gross--Pitaevskii scaling regime in \cite{BBCS18}, and in an intermediate scaling regime in \cite{GS09,BBCS17,BS19}. The Lee--Huang--Yang formula for the low-density limit has been proven as an upper bound in \cite{ESY08} for small potential and in \cite{YY09} for general potential, and only very recently as a lower bound \cite{FS19}.

\section{Main Result}
In this section we present our main result, Theorem \ref{thm:main}. Our theorem provides an upper bound for the ground state energy, which is consistent with the Gell-Mann--Brueckner formula for the correlation energy. 

Notice that for the interaction potential we normalize the Fourier transform such that $\hat{V}(k) = (2\pi)^{-3} \int \di x\, e^{-ik\cdot x} V(x)$, whereas for wave functions we choose it unitary in $L^2$.
\begin{thm}[Upper Bound for the Ground State Energy]\label{thm:main}
Let $\hat{V}: \Zbb^3 \to \mathbb{R}$ be non-negative and compactly supported.
Let $k_F > 0$ be the Fermi momentum and $N := \lvert\{ k \in \Zbb^3 : \lvert k\rvert \leq k_F\}\rvert$ the number of particles; recall that $\hbar = N^{-1/3}$. Let 
$\omega_\textnormal{pw} := \sum_{k \in \Zbb^3 : \lvert k\rvert \leq k_F} \lvert f_k \rangle \langle f_k \rvert$ be the projection on the filled Fermi ball. Then, asymptotically for  $k_F \to \infty$, the ground state energy \eqref{eq:gse-HN} satisfies the upper bound
\begin{align} E_N &\leq  \mathcal{E}_\textnormal{HF}(\omega_\textnormal{pw}) \nonumber\\
& \quad + \frac{\hbar \kappa_0}{2m} \sum_{k \in \Zbb^3} \lvert k\rvert \left[ \frac{1}{\pi} \int_0^\infty \log \left(1+4\pi\hat{V}(k) m\kappa_0 \left( 1 - \lambda \arctan \frac{1}{\lambda} \right)\right) \di \lambda - \hat{V}(k)m\kappa_0\pi \right]\nonumber\\
& \quad + \hbar\, \Ocal(N^{-1/27})\;,\label{eq:mainresult}\end{align}
where $\kappa_0 = (\frac{3}{4\pi})^{1/3}$.
\end{thm}
\begin{rems}
\ 
\begin{enumerate}
\item We conjecture that there is actually equality in \eqref{eq:mainresult}; i.\,e., a corresponding lower bound, possibly with different error exponent, should hold.
\item Recall that the Hartree--Fock energy $\mathcal{E}_\textnormal{HF}(\omega_\textnormal{pw})$ consists of kinetic energy (order $N$), direct interaction energy (order $N$), and exchange interaction energy (order $1$). Our many-body correction is of order $\hbar = N^{-1/3}$. As expected, it is negative, so that it improves over $\mathcal{E}_\textnormal{HF}(\omega_\textnormal{pw})$.

\item Notice that already with regular interaction potential the correlation correction at order $\hbar$ involves arbitrarily high powers of the interaction potential.

\item If we formally extrapolate our formula to the jellium model it agrees with the correlation energy first obtained by Gell-Mann and Brueckner \cite[Equation~(19)]{GB57} as a power series; see also \cite[Equation~(37)]{SBFB57} for the first appearance of the explicit expression. Gell-Mann and Brueckner also obtain a contribution from a second order exchange-type term denoted ${\epsilon_b}^{(2)}$; for us, in mean-field scaling and with compactly supported $\hat{V}$, this term is only of order $\hbar^2$. However, since our trial state captures the second order direct-type term correctly and can be expanded in powers of $\hat{V}$, we expect that it would also capture the second order exchange-type term in models where it has a bigger contribution.

\item For small interaction potentials $\hat V$, we can expand
\[
\begin{split}
  & \frac{1}{\pi} \int_0^\infty \log \left(1+4\pi\hat{V}(k) m\kappa_0 \left( 1 - \lambda \arctan \frac{1}{\lambda} \right)\right) \di \lambda - \hat{V}(k)m\kappa_0\pi  \\
  &\hspace{5cm}  = - \frac{8 \pi^2}{3} \hat{V}(k)^2 m^2 \kappa^2_0  \left(1-\log(2)\right)  + \mathcal{O}\left(\hat{V}(k)^3\right).
\end{split}
\]
Therefore
\[
\frac{E_N - \Ecal_\textnormal{HF}(\omega_\textnormal{pw})}{\hbar} \leq - m \pi (1-\log(2)) \sum_{k \in \Zbb^3} \lvert k\rvert \hat{V}(k)^2 (1 + \Ocal(\hat V(k))) + \Ocal(N^{-1/27}).
\]
This is consistent with \cite{HPR18}, see \eqref{eq:HPR18} (notice that \cite{HPR18} considered the Fermi gas in $[0,1]^3$ instead of $[0,2\pi]^3$). Whereas \cite{HPR18} uses rigorous second-order perturbation theory, here we use a non-perturbative bosonization method which directly yields a resummation of the dominant contributions of the perturbation series both of the ground state and the ground state energy to all orders in the potential.
\item The assumption of $\hat{V}$ being compactly supported is mainly used to control the number of particle--hole pairs that may be lost near the boundaries of patches (see Section \ref{sec:continuumapprox}) and to avoid interaction between different patches across the separating corridors (see Figure \ref{fig:nontouching}). A sufficiently fast power law decay of $\hat{V}(k)$ for large $k$ should also be sufficient to control such error terms but to keep the presentation readable we do not follow up on this question here.
\end{enumerate}
\end{rems}
In the remaining part of the paper we prove Theorem \ref{thm:main}. Our proof is based on a reorganization of the particle--hole excitations around the Fermi surface in terms of approximately bosonic collective degrees of freedom, which we will introduce in the next section.
Notice that $1/m$ can be factored out from the Hamiltonian, replacing the potential $V$ by $mV$, so we consider only $m=1$ and the dependence on $m$ is easily restored at the end.

\section{Collective Particle-Hole Pairs} 
In this section we represent the correlation energy in terms of particle--hole excitations around the Fermi surface. These excitations will be described by quadratic fermionic operators on the Fock space, that behave as almost bosonic operators. The advantage of this rewriting is that the correlation energy can thus be related to the ground state energy of a \emph{quadratic} almost-bosonic Hamiltonian.

\subsection{The Correlation Hamiltonian}
Here we shall introduce a Fock space representation of the model. We shall follow the notations of \cite[Chapter 6]{BPS16}, to which we refer for more details. Let $\mathcal{F}:= \mathcal{F}(L^{2}(\Tbb^{3}))$ be the fermionic Fock space built on the single-particle space $L^{2}(\Tbb^{3})$. Let us denote by $\mathcal{H}_{N}$ the second quantization of $H_{N}$. We have
\[\mathcal{H}_N = \frac{\hbar^2}{2} \int \di x \nabla_x a^*_x \nabla_x a_x + \frac{1}{2N} \int \di x \di y\, V(x-y) a^*_x a^*_y a_y a_x\,,\]
where $a^{*}_{x}$, $a_{x}$ are the creation and annihilation operators (more precisely, operator-valued distributions), creating or annihilating a fermionic particle at $x \in \Tbb^3$. They satisfy the usual canonical anticommutation relations (CAR)
\begin{equation}\label{eq:car}\{a_x,a_y\} = 0 = \{a^*_x,a^*_y\}, \quad \{a_x,a^*_y\} = \delta(x-y)\,.\end{equation}
Given a function $f \in L^{2}(\Tbb^{3})$ we also define $a(f) := \int \di x\, a_{x} \overline{f(x)}$ and $a^{*}(f) = \left(a(f)\right)^{*}$. 

\medskip

Let us define the Fermi ball
\[
B_{\textnormal{F}} := \{ k\in \Zbb^{3} : |k|\leq k_{\textnormal{F}} \}\;,
\]
where $k_{\textnormal{F}}$ is the Fermi momentum. Let $N$ be the number of points in the Fermi ball, $N := \lvert B_{\textnormal{F}} \rvert$. Then, by Gauss' classical counting argument, 
\begin{equation}\label{eq:gauss}
\begin{split}
k_{\textnormal{F}} = \kappa N^{1/3}\;,\qquad \kappa = \kappa(N) & =  (3/4\pi)^{1/3} + \Ocal(N^{-1/3})\\
& =: \kappa_0 + \Ocal(N^{-1/3}) \;.
\end{split}
\end{equation}
We also introduce the complement of the Fermi ball,
\[
B_{\textnormal{F}}^{c} = \Zbb^{3} \setminus B_{\textnormal{F}}.
\]

The filled Fermi ball is obtained by considering the Slater determinant $\psi_\textnormal{pw}$ built from the plane waves $f_{k_{i}}(x) = (2\pi)^{-3/2} e^{ik_{i}\cdot x}$, associated to the points $k_{i} \in B_{\textnormal{F}}$, $i = 1,\ldots, N$. Let $\omega_{\textnormal{pw}}$ be the reduced one-particle density matrix associated to such states, $\omega_{\textnormal{pw}} = \sum_{i=1}^{N} \lvert f_{k_{i}} \rangle \langle f_{k_{i}} \rvert$.  With the plane waves $f_k$ defined in \eqref{eq:plane-wave}, we define the unitary\footnote{It is an amusing exercise to check that $R_{\omega_\textnormal{pw}}$ is invertible; in fact $R_{\omega_\textnormal{pw}} = R^{-1}_{\omega_\textnormal{pw}}$. Furthermore, $R_{\omega_\textnormal{pw}}$ is clearly isometric, and thus unitary.} particle--hole transformation $R_{\omega_\textnormal{pw}}: \fock \to \fock$ by setting 
\[R_{\omega_\textnormal{pw}} a (f_k) R^*_{{\omega_\textnormal{pw}}} := \left\{ \begin{array}{cc} a(f_k) & \text{for }k \in B_\textnormal{F}^c\\ a^*(\cc{f_k}) & \text{for }k \in B_\textnormal{F}\end{array} \right. \qquad \text{and} \qquad R_{\omega_\textnormal{pw}} \Omega := \psi_\textnormal{pw}\, . \]
Here we introduced the vacuum vector $\Omega = (1,0,0,\ldots) \in \fock$. Particle-hole transformations are a particular kind of fermionic Bogoliubov transformation. 
In fact, formally writing $a_x = a(\delta(\cdot-x))$ and $\delta(y-x) = \sum_{k \in \Zbb^3} f_k(y) \cc{f_k(x)}$ one can rewrite the previous relation in position space,
\begin{equation}\label{eq:phtrafo}R_{\omega_\textnormal{pw}} a_x R^*_{\omega_\textnormal{pw}} = a(u_x) + a^*(\cc{v}_x)\,, \qquad R_{\omega_\textnormal{pw}} a^*_x R^*_{\omega_\textnormal{pw}} = a^*(u_x) + a(\cc{v}_x)\,,\end{equation}
where $u= \id - \omega_\textnormal{pw}$, $v = \sum_{k \in B_\textnormal{F}} \lvert \cc{f_k}\rangle \langle f_k \rvert$ and where we also introduced the short-hand notation $v_x(\cdot) = v(\cdot,x) = \sum_{k \in B_\textnormal{F}} \cc{f_k}(\cdot) \cc{f_k}(x)$ and $u_x(\cdot) = u(\cdot,x) = \delta(\cdot-x)-\sum_{k \in B_\textnormal{F}} {f_k}(\cdot) \cc{f_k}(x)$.

The state $R_{\omega_\textnormal{pw}} \Omega$ plays the role of the new vacuum for the model, on which the new fermionic operators $R_{\omega_\textnormal{pw}} a (f_k) R^*_{{\omega_\textnormal{pw}}}$ act. We call momenta in $B_\textnormal{F}$ \emph{hole modes}, and momenta in $B^{c}_\textnormal{F}$ \emph{particle modes}. We will use the notation $a^*_k := a^*(f_k)$. If we want to emphasize that the index is outside the Fermi ball we write $a^*_p$, $p \in B_\textnormal{F}^c$ (``p'' like ``particle'') and say that $a^*_p$ creates a particle. Similarly we use $a^*_h$, $h \in B_\textnormal{F}$ (``h'' like ``hole'') and say that $a^*_h$ creates a hole in the Fermi ball. We call $\Ncal_\textnormal{p} := \sum_{p \in B_\textnormal{F}^c} a^*_p a_p$ the number-of-particles operator and $\Ncal_\textnormal{h} := \sum_{h \in B_\textnormal{F}} a^*_h a_h$ the number-of-holes operator. If we do not want to distinguish between particles and holes we use the word ``fermion'', for example calling $\Ncal = \Ncal_\textnormal{p} + \Ncal_\textnormal{h}$ the number-of-fermions operator. 

\medskip

Let us consider the conjugated Hamiltonian $R^*_{{\omega_\textnormal{pw}}} \mathcal{H}_N R_{\omega_\textnormal{pw}}$. Using \eqref{eq:phtrafo}, and rewriting the result into a sum of normal-ordered contributions one gets (see \cite[Chapter 6]{BPS16} for a similar computation in the context of many-body quantum dynamics):
\begin{align}
R^*_{{\omega_\textnormal{pw}}} \mathcal{H}_N R_{\omega_\textnormal{pw}} &= \mathcal{E}_\textnormal{HF}(\omega_\textnormal{pw}) + \di\Gamma(uhu-\cc{v}\cc{h}v) + Q_N \label{eq:cLN}
\end{align}
with $d\Gamma(A)$ the second quantization\footnote{The second quantization of the one-particle operator $A$ is defined on the $n$-particle sector of $\mathcal{F}$ as $\di\Gamma (A) := \sum_{j=1}^n A_j$, where $A_j := \id^{\otimes j-1} \otimes A \otimes \id^{\otimes n-j}$ acts non-trivially only on the $j$-th particle. If $A$ has an integral kernel $A(x,y)$, its second quantization can be written as $\di\Gamma (A) = \int A(x,y) a_x^* a_y \di x \di y$.} of a one-particle operator $A$.
The operator $h$ is the one-particle Hartree--Fock Hamiltonian, given by 
\begin{equation}\label{eq:h-def} h = -\frac{\hbar^2 \Delta}{2} + (2\pi)^3 \hat{V} (0) + X \,\end{equation} 
where $X$ is the exchange operator, defined by its integral kernel $X(x,y) = -N^{-1} V(x-y) \omega_\textnormal{pw}(x,y)$. As for the operator $Q_N$ on the r.\,h.\,s.\ of \eqref{eq:cLN}, it contains all contributions that are quartic in creation and annihilation operators. It is given by 
\[\begin{split}Q_N & = \frac{1}{2N} \int_{\Tbb^3 \times \Tbb^3} \di x\di y\, V(x-y) \bigg(\mathcal{E}_1(x,y)  + 2 a^*(u_x) a^*(\cc{v}_x) a(\cc{v}_y) a(u_y)\\
& \hspace{4cm} + \Big[ a^*(u_x) a^*(\cc{v}_x) a^*(u_y) a^*(\cc{v}_y)  + \mathcal{E}_2(x,y)  + \hc\Big]
\bigg)\end{split}\]
where
\begin{equation}\label{eq:cE1} \mathcal{E}_1(x,y) = a^*(u_x)a^*(u_y) a(u_y)a(u_x)- 2 a^*(u_x) a^*(\cc{v}_y) a(\cc{v}_y) a(u_x) + a^*(\cc{v}_y)a^*(\cc{v}_x) a(\cc{v}_x) a(\cc{v}_y)\end{equation}
and
\begin{equation}\label{eq:cE2} \mathcal{E}_2(x,y) = - 2a^*(u_x) a^*(u_y)a^*(\cc{v}_x) a(u_y) + 2 a^*(u_x) a^*(\cc{v}_y) a^*(\cc{v}_x) a(\cc{v}_y).\end{equation}
As we shall see, both $\Ecal_{1}$ and $\Ecal_{2}$ will provide subleading corrections to the correlation energy, as $N\to \infty$. The operator $R^*_{{\omega_\textnormal{pw}}} \mathcal{H}_N R_{\omega_\textnormal{pw}} - \mathcal{E}_\textnormal{HF}(\omega_\textnormal{pw})$ is called the \emph{correlation Hamiltonian},
\begin{equation}
\label{eq:Hcorr}
\mathcal{H}_{\textnormal{corr}} := \di\Gamma(uhu-\cc{v}\cc{h}v) + Q_N\;.
\end{equation}

\medskip

Let $\psi\in \mathcal{F}$ be a normalized $N$-particle state in the fermionic Fock space, that is $\psi = (0, 0, \ldots, 0, \psi^{(N)}, 0, \ldots)$. By the variational principle, we have
\[
E_{N} \leq \langle \psi, \mathcal{H}_{N} \psi\rangle = \mathcal{E}_\textnormal{HF}(\omega_\textnormal{pw}) + \langle \xi , \mathcal{H}_{\textnormal{corr}} \xi \rangle\;,
\]
where $\xi = R_{\omega_\textnormal{pw}}^* \psi$. The last step follows from the identity \eqref{eq:cLN}.

\medskip

We are going to construct an $N$-particle state $\psi_\textnormal{trial} = R_{\omega_\textnormal{pw}} \xi$ such that $\langle \xi, \mathcal{H}_\textnormal{corr} \xi\rangle$ is given by the Gell-Mann--Brueckner formula
\[
\frac{\hbar \kappa_0}{2} \sum_{k \in \Zbb^3} \lvert k\rvert \left[ \frac{1}{\pi} \int_0^\infty \log \left(1+4\pi\hat{V}(k) \kappa_0 \left( 1 - \lambda \arctan \frac{1}{\lambda} \right)\right) \di \lambda - \hat{V}(k)\kappa_0\pi \right]\,,
\]
up to errors that are of smaller order as $N \to \infty$. To construct this state, we shall represent $\mathcal{H}_{\textnormal{corr}}$ in terms of suitable almost-bosonic operators, obtained by combining fermionic particle--hole excitations. As we shall see, the resulting expression will be quadratic in terms of these new operators; the state $\xi$ will be chosen to minimize the bosonic energy.

\subsection{Particle-Hole Excitations}
We start by rewriting the quartic contribution to the correlation Hamiltonian as
\begin{equation}
\begin{split}\label{eq:QN0-def} Q_N & = Q_N^{\textnormal{B}} + \frac{1}{2N} \int_{\Tbb^3 \times \Tbb^3} \di x\di y V(x-y) \big(\mathcal{E}_1(x,y) + \left[ \mathcal{E}_2(x,y) +\hc\right] \big)\;,\\
Q_{N}^{\textnormal{B}} &=   \frac{1}{2N} \sum_{k \in \Zbb^3} \hat{V}(k) \int_{\Tbb^3 \times \Tbb^3} \di x \di y \bigg( 2 a^*(u_x)e^{ikx} a^*(\cc{v}_x) a(\cc{v}_y) e^{-iky} a(u_y)\\
&  \hspace{4.5cm}+ \left[a^*(u_x)e^{ikx} a^*(\cc{v}_x) a^*(u_y) e^{-iky} a^*(\cc{v}_y)  + \hc\right]\bigg).
\end{split}
\end{equation}
The main contribution to $Q_N$ is $Q_{N}^{\textnormal{B}}$, which, as we shall see, can be represented as a quadratic operator in terms of collective particle--hole pair operators. These operators behave approximately like bosonic creation and annihilation operators.

\medskip

Let us define the (unnormalized) particle--hole operator as
\[\tilde b^*_k := \int_{\Tbb^3} \di x\, a^*(u_x) e^{ikx} a^*(\cc{v}_x)\,.
\]
Notice that $\tilde{b}^*_0 = 0$ since $u\cc{v} =0$.
Writing this operator in momentum representation,
\begin{equation}
\label{eq:bosonicpair}\tilde b^*_k = \sum_{\substack{p\in B_\textnormal{F}^c\\h \in B_\textnormal{F}}} a^*_p a^*_h \delta_{p-h,k}\,,\end{equation}
we can think of it as creating a particle--hole pair of momentum $k$, \emph{delocalized} over all the Fermi surface. In terms of these operators
\[
Q_N^{\textnormal{B}} = \frac{1}{2N} \sum_{k\in \Zbbn} \hat{V}(k) \left( 2 \tilde b^*_k \tilde b_k + \tilde b^*_k \tilde b^*_{-k} + \tilde b_{-k} \tilde b_k \right).
\]
Recall that $\hat{V}$ has compact support by assumption, so there exists
\[
R > 0 \text{ such that } \hat{V} (k) = 0 \text{ for all } \lvert k\rvert > R\;.
\]
It is convenient to group together $k$ and $-k$ modes, as follows.
Define  
\begin{equation}\label{eq:north}\north \subset \Zbb^3\end{equation}
as the set of all $k \in \Zbb^3 \cap B_R(0)$ with $k_3 >0$ and additionally half of the $k$-vectors with $k_3 =0$, such that for every $k \in \north$ we have $-k \not\in \north$. We then rewrite $Q_N^{\textnormal{B}}$ as
\begin{equation}
\label{eq:QN0-1}Q_N^{\textnormal{B}} = \frac{1}{2N} \sum_{k \in \north} \hat{V}(k) \left( 2 \tilde b^*_k \tilde b_k + \tilde b^*_k \tilde b^*_{-k} + \tilde b_{-k} \tilde b_k + 2 \tilde b^*_{-k} \tilde b_{-k} + \tilde b^*_{-k} \tilde b^*_{k} + \tilde b_{k} \tilde b_{-k}\right).
\end{equation}
It turns out that the operators $\tilde b_k$ behave as approximate bosonic operators, whenever acting on vectors of $\mathcal{F}$ with only a few particles; the Pauli principle is relaxed by summing over a large number of momenta of which typically only few are occupied.

The main problem, however, is that the term $\di\Gamma(uhu-\cc{v}\cc{h}v)$ in \eqref{eq:cLN} \emph{cannot} be represented as a quadratic operator in terms of $\tilde b_k$ and $\tilde b_k^{*}$. To circumvent this issue we shall split the operators $\tilde{b}_{k}$, $\tilde{b}^{*}_{k}$ into \emph{partially localized} particle--hole operators $\tilde{b}_{\alpha,k}$, $\tilde{b}^*_{\alpha,k}$ involving only modes of one patch of a decomposition (indexed by $\alpha$) of the Fermi surface. This allows us to linearize the kinetic energy around the centers of patches, so that states of the form
\[\tilde{b}^*_{\alpha_1,k_1} \tilde{b}^*_{\alpha_2,k_2} \cdots \tilde{b}^*_{\alpha_m,k_m} \Omega\]
become approximate eigenvectors of $\di\Gamma(uhu-\cc{v}\cc{h}v)$.

The non-trivial question is whether we can localize \eqref{eq:bosonicpair} sufficiently to control the linearization of the kinetic energy, while at the same time keeping it sufficiently delocalized so that $\tilde{b}^*_{\alpha,k}$ involves many fermionic modes, thus relaxing the Pauli principle---complete localization would of course destroy the bosonic behavior since $(a^*_p a^*_h)^2 =0$. We are going to find that this can be achieved by decomposing the Fermi sphere into $M = M(N)$ diameter-bounded equal-area patches if $N^{1/3} \ll M \ll N^{2/3}$.

\paragraph{Patch Decomposition of the Fermi Sphere.} We construct a partition of the Fermi sphere $k_F \Sbb^2$ into $M$ diameter-bounded equal-area patches following \cite{Leo06}, see Figure~\ref{fig:blub}. Let
\[M = M(N) := N^{1/3 + \epsilon} \quad\text{for an } 0 < \epsilon < 1/3\,,\]
 or more precisely, this number rounded to the nearest even integer. Our goal is to first decompose the unit sphere $\mathbb{S}_{2}$ as
\[ \mathbb{S}_2 = \left( \bigcup_{\alpha=1}^{M} p_\alpha \right) \cup p_\textnormal{corri}\,,\]
where $p_{\alpha}$ are suitable pairwise disjoint sets, to be defined below, and $p_\textnormal{corri}$ has small surface measure, $\sigma(p_\textnormal{corri}) = \mathcal{O}(M^{1/2} N^{-1/3}) \to 0$ as $N \to \infty$. The error $p_\textnormal{corri}$ is due to the introduction of a positive distance (``corridors'') separating neighboring patches. The important properties to be ensured in the construction are that all patches $p_\alpha$ have area of order $1/M$ and that they do not degenerate into very long, thin shapes as $M$ becomes large.
\begin{figure}\centering
\begin{tikzpicture}[scale=0.8]
\def\RadiusSphere{4} 
\def\angEl{20} 
\def\angAz{-20} 

\filldraw[ball color = white] (0,0) circle (\RadiusSphere);

\DrawLatitudeCircle[\RadiusSphere]{75+2}
\foreach \t in {0,-50,...,-250} {
  \DrawLatitudeArc{75}{(\t+50-4)*sin(62)}{\t*sin(62)}
 \DrawLongitudeArc{\t*sin(62)}{50+2}{75}
 \DrawLongitudeArc{(\t-4)*sin(62)}{50+2}{75}
  \DrawLatitudeArc{50+2}{(\t+50-4)*sin(62)}{\t*sin(62)}
 }
 \foreach \t in {0,-50,...,-300} {
   \DrawLatitudeArc{50}{(\t+50-4)*sin(37)}{\t*sin(37)}
 \DrawLongitudeArc{\t*sin(37)}{25+2}{50}
  \DrawLongitudeArc{(\t-4)*sin(37)}{25+2}{50}
   \DrawLatitudeArc{25+2}{(\t+50-4)*sin(37)}{\t*sin(37)}
 }
 \DrawLatitudeArc{50}{(-300-4)*sin(37)}{-330*sin(37)}
 \foreach \t in {0,-50,...,-450} {
    \DrawLatitudeArc{25}{(\t+50-4)*sin(23)}{\t*sin(23)}
 \DrawLongitudeArc{\t*sin(23)}{00+2}{25}
 \DrawLongitudeArc{(\t-4)*sin(23)}{00+2}{25}
 \DrawLatitudeArc{00+2}{(\t+50-4)*sin(23)}{\t*sin(23)}
 }
     \DrawLatitudeArc{25}{(-450-4)*sin(23)}{-500*sin(23)}

\fill[black] (0,3.75) circle (.075cm);

\fill[black] (1.72,3.08) circle (.075cm);
\fill[black] (.76,2.73) circle (.075cm);
\fill[black] (-.66,2.73) circle (.075cm);
\fill[black] (-1.73,3.04) circle (.075cm);

\fill[black] (2.25,1.5) circle (.075cm);
\fill[black] (.8,1.2) circle (.075cm);
\fill[black] (-.85,1.22) circle (.075cm);
\fill[black] (-2.27,1.5) circle (.075cm);
\fill[black] (-3.09,1.97) circle (.075cm);
\fill[black] (3.09,1.97) circle (.075cm);

\fill[black] (2.57,-.15) circle (.075cm);
\fill[black] (1.43,-.37) circle (.075cm);
\fill[black] (.155,-.48) circle (.075cm);
\fill[black] (-1.17,-.41) circle (.075cm);
\fill[black] (-2.35,-.2) circle (.075cm);
\fill[black] (-3.26,0.1) circle (.075cm);
\fill[black] (-3.79,.55) circle (.075cm);
\fill[black] (3.37,.18) circle (.075cm);
\fill[black] (3.85,.57) circle (.075cm);

\end{tikzpicture}
 \caption{Diameter-bounded partition of the northern half sphere following \cite{Leo06}: a spherical cap is placed at the pole; then collars along the latitudes are introduced and split into patches, separated by corridors. The vectors $\hat{\omega}_\alpha$ are picked as centers of the patches, marked in black. The patches will be reflected by the origin to cover also the southern half sphere.}\label{fig:blub}
\end{figure}
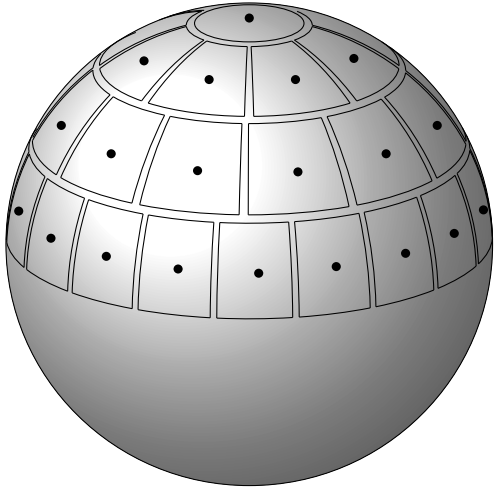
\medskip

We use standard spherical coordinates: for $\hat\omega \in \Sbb^2$, denote by $\theta$ the inclination angle (measured between $\hat\omega$ and $e_3 = (0,0,1)$) and by $\varphi$ the azimuth angle (measured between $e_1 = (1 , 0 , 0)$ and the projection of $\hat{\omega}$ onto the plane orthogonal to $e_3$). We write $\hat{\omega}(\theta,\varphi)$ to specify a vector on the unit sphere in terms of its inclination and azimuth angles.

The construction starts by placing a spherical cap centered at $e_3$, with opening angle $\Delta\theta_0 := D /\sqrt{M}$, with $D \in \Rbb$ chosen so that the area of the spherical cap equals $4\pi / M$. Next, we decompose  the remaining part of the half sphere, i.\,e., the set of all $\hat{\omega} (\theta, \varphi)$ with $D/{\sqrt{M}} \leq \theta \leq \pi/2$, into $\sqrt{M}/2$ (rounded to the next integer) collars; the $i$-th collar consists of all $\hat{\omega}(\theta,\varphi)$ with $\theta \in [\theta_i - \Delta\theta_i,\theta_i+\Delta\theta_i)$ and arbitrary azimuth $\varphi$. The inclination of every collar will extend over a range $\Delta\theta_i \sim 1/{\sqrt{M}}$; the proportionality constant is adjusted so that the number of collars on the half sphere is an integer.

Observe that the circle $\left\{ \hat\omega(\theta_i,\varphi): \varphi \in [0,2\pi) \right\}$ 
has circumference proportional to $\sin(\theta_i)$; therefore we split the $i$-th collar into $\sqrt{M}\sin(\theta_i)$ (rounded to the next integer) patches. This implies that the $j$-th patch in the $i$-th collar covers an azimuth angles $\varphi \in [\varphi_{i,j}-\Delta\varphi_{i,j},\varphi_{i,j} + \Delta\varphi_{i,j})$, where \[\Delta\varphi_{i,j} \sim \frac{1}{{\sin(\theta_i)\sqrt{M}}}\,.\]
We fix the proportionality constants by demanding that all patches have area $4\pi/M$ (this is not necessary though, it would be sufficient that all patches have area of order $1/M$).

The last step is to define $\Delta\widetilde{\theta}_i := \Delta\theta_i - \tilde{D} R N^{-1/3}$ and  $\Delta\widetilde{\varphi}_{i,j} := \Delta\varphi_{i,j} - \tilde{D} R N^{-1/3}/\sin(\theta_i)$, with $\tilde{D} >0$ to be fixed below. We then define $p_1$ as the spherical cap centered at $e_3$ with opening angle $\Delta\widetilde{\theta}_0$ and the other $M/2 - 1$ patches as
\[
p_{i,j} := \big\{ \hat{\omega}(\theta,\varphi): \theta \in [\theta_i-\Delta\widetilde{\theta}_i,\theta_i+\Delta\widetilde{\theta}_i) \text{ and } \varphi \in [\varphi_{i,j}-\Delta\widetilde{\varphi}_{i,j},\varphi_{i,j} + \Delta\widetilde{\varphi}_{i,j})\big\}\,.
\]
The constant $\tilde{D}$ is chosen such that, when patches are scaled up to the Fermi sphere there are corridors of width at least $2R$ between adjacent patches (i.\,e., $\tilde{D}$ has to be slightly larger than $\kappa_0^{-1}$). Having concluded the construction on the northern half sphere, we define the patches on the southern half sphere through reflection by the origin, $k \mapsto -k$. Finally we switch from enumeration by $i$ and $j$ to enumeration with a single index $\alpha \in \{1, \dots , M\}$. From the construction it is clear that the patches $p_\alpha$ have the following three properties:
\begin{enumerate}
\item The area of every patch is
\[
\sigma(p_\alpha) = \frac{4\pi}{M} + \Ocal\big(N^{-1/3} M^{-1/2}\big)\;.
\]
\item The family of decompositions is diameter bounded, i.\,e., there exists a constant $C_0$ independent of $N$ and $M$ such that, for the decomposition into $M$ patches, the diameter\footnote{The linear dimension of a patch measured by the Euclidean norm of $\Rbb^3$ or measured by the geodesic distance on $\Sbb^2$ are of the same order, so we do not need to worry about this distinction.} of every patch is bounded by $C_0/\sqrt{M}$.
\item Point reflection at the origin maps $p_\alpha$ to $-p_\alpha = p_{\alpha+\frac{M}{2}}$ for all $\alpha = 1,\ldots ,\frac{M}{2}$.
\end{enumerate}

Next, we scale the patches from the unit sphere up to the Fermi surface $k_\textnormal{F}  \Sbb^2$ by setting  
\[P_\alpha := k_\textnormal{F} p_\alpha\]
for all $\alpha = 1, \dots , M$. The patches $P_\alpha$ then have the following properties.
\begin{enumerate}
\item The area of every patch is $\sigma(P_{\alpha}) = \frac{4\pi}{M} k_\textnormal{F}^2+ \Ocal\left(N^{1/3} M^{-1/2}\right)$.
\item There exists a constant $C_1$ independent of $N$ and $M$ such that, for the decomposition in $M$ patches, we have $\diam(P_\alpha) \leq C_1 N^{1/3}/\sqrt{M}$.
\end{enumerate}

Finally, we shall introduce a ``fattening'' of the patch decomposition, which will be used to decompose the operators $b_{k}$ as sums of operators corresponding to particle--hole excitations around the patches. This is motivated by the fact that the only modes affected by the interaction are those in a shell around the Fermi sphere, where the thickness of the shell is given by the radius of the support of $\hat{V}$. Recalling again that $R > 0$ is chosen such that $\hat{V} (k) = 0$ for $\lvert k\rvert > R$, we define the fattened Fermi surface as
\[\partial B_\textnormal{F}^R := \left\{ q \in \Zbb^3 :  k_\textnormal{F} - R \leq \lvert q \rvert \leq k_\textnormal{F} +R  \right\}\,.\]
We lift the partition of the unit sphere to a partition of $\partial B_\textnormal{F}^R$,
\[\partial B_\textnormal{F}^R = \left( \bigcup_{\alpha=1}^M B_\alpha \right) \cup B_\textnormal{corri}\,,\]
by introducing the cones $\mathcal{C}_\alpha := \bigcup_{r \in (0,\infty)} r p_\alpha$ and defining
\[B_\alpha := \partial B_\textnormal{F}^R \cap \mathcal{C}_\alpha\,.\]
(The set $B_\textnormal{corri}$ consist of all the remaining modes in the similarly fattened corridors.)
To every patch $B_\alpha$ we assign a vector $\omega_\alpha \in B_\alpha$ as the center of $P_\alpha$ on the Fermi surface; in particular $\lvert \omega_\alpha\rvert =k_\textnormal{F}$. The vectors $\omega_{\alpha}$ inherit the reflection symmetry of the patches, $\omega_{\alpha + M/2} = -\omega_\alpha$ for all $\alpha = 1, \dots , M/2$. 

\paragraph{Localization on the Fermi Surface.}
We recall \eqref{eq:bosonicpair} in momentum representation,
\begin{equation}\label{eq:kspace}
\tilde b^*_k = \sum_{\substack{p\in B_\textnormal{F}^c\\h \in B_\textnormal{F}}} a^*_p a^*_h \delta_{p-h,k}\,.
\end{equation}
Since $\hat{V} (k) = 0$ if $\lvert k\rvert > R$, we are only interested in the case $\lvert k\rvert \leq R$; hence, the sum in \eqref{eq:kspace} effectively runs only over $p$ and $h$ at most at distance $R$ from the Fermi sphere $k_\textnormal{F} \Sbb^2$. In other words, 
\begin{equation}\label{eq:wtbk}
\tilde b^*_k = \sum_{\substack{p\in B_\textnormal{F}^c \cap \partial B_\textnormal{F}^R\\h \in B_\textnormal{F} \cap \partial B_\textnormal{F}^R}} a^*_p a^*_h \delta_{p-h,k}\,.
\end{equation}
Next, we decompose the sum on the r.\,h.\,s.\ of \eqref{eq:wtbk} into contributions associated with different patches. If $k \cdot \omega_\alpha < 0$, there will be few or no particle--hole pairs $(p,h)$ in the patch $B_\alpha$ satisfying $p-h = k$; geometrically, $k$ is approximately pointing from outside to inside of the Fermi ball, which is incompatible with the requirements $p \in B_\textnormal{F}^c$ and $h \in B_\textnormal{F}$. Also if $k \cdot \omega_\alpha$ is positive but small, there are only few particle--hole pairs $(p,h)$ with $p-h = k$. For this reason, for any $k \in \mathbb{Z}^3$, we define the index set\footnote{We use the notation $\hat{k} := k/ \lvert k\rvert$ for the unit vector in direction of $k$.} 
\[ \Ikp := \big\{\alpha =1,\ldots, M: \hat\omega_\alpha \cdot \hat{k} \geq N^{-\delta}\big\} \]
for a parameter $\delta > 0$ to be chosen later. We then write 
\begin{equation}\label{eq:patchdecomp}\tilde b^*_k = \sum_{\alpha \in \Ikp} \tilde b^*_{\alpha,k} + \rfrak^*_k\;,\end{equation}
where
\[
\tilde{b}^*_{\alpha,k} := \sum_{\substack{p\in B_\textnormal{F}^c\cap B_\alpha\\h \in B_\textnormal{F}\cap B_\alpha}} a^*_p a^*_h \delta_{p-h,k}\;.
\]
The operator $\rfrak^*_k$ contains all particle--hole pairs that are not included in $\sum_{\alpha\in\Ikp} \tilde{b}^*_{\alpha,k}$. This can happen for two reasons: because an index $\alpha$ is not included in $\Ikp$, or because one or both momenta of a pair $(p,h)$ belong to a corridor between patches. As we shall see, this operator can be understood as a small error, due to the fact that the number of pairs $(p, h)$ not included in the first sum is small.

\paragraph{Normalization of Particle-Hole Pair Operators.} We still have to normalize the pair operators so that they can be seen as an approximation of bosonic operators. The normalized operators are defined by
\begin{equation}
 \label{eq:bka}
b^*_{\alpha,k} := \frac{1}{n_{\alpha,k}} \tilde b^*_{\alpha,k}, \qquad n_{\alpha,k} := \norm{\tilde b^*_{\alpha,k} \Omega}.
\end{equation}
We call these operators the \emph{pair creation operators}; their adjoints are called \emph{pair annihilation operators}. The normalization constant can be calculated as follows:
\[
\begin{split}
\norm{\tilde b^*_{\alpha,k}\Omega}^2 & = \langle \Omega, \Bigg[\sum_{\substack{p_1 \in B_\textnormal{F}^c \cap B_\alpha\\h_1 \in B_\textnormal{F} \cap B_\alpha}} a^*_{p_1} a^*_{h_1} \delta_{p_1 - h_1,k}\Bigg]^* \Bigg[\sum_{\substack{p_2 \in B_\textnormal{F}^c \cap B_\alpha\\h_2 \in B_\textnormal{F} \cap B_\alpha}} a^*_{p_2} a^*_{h_2} \delta_{p_2 - h_2,k}\Bigg] \Omega\rangle= \sum_{\substack{p \in B_\textnormal{F}^c \cap B_\alpha\\h \in B_\textnormal{F} \cap B_\alpha}}\! \delta_{p - h,k}\;. 
\end{split}
\]
This shows that $n_{\alpha,k}^2$ is the number of particle--hole pairs with momentum $k = p-h$ that lie in the patch $B_\alpha$. Due to the symmetry of the partition under point reflection at the origin we have $n_{\alpha,k} = n_{\alpha+M/2,-k}$. We define $v_\alpha(k) \geq 0$ by setting
\begin{equation}\label{eq:vkdef}n_{\alpha,k}^2 =: k_\textnormal{F}^2 \lvert k\rvert v_\alpha(k)^2\;.\end{equation}
In the next proposition, whose proof is deferred to Section~\ref{sec:continuumapprox}, we estimate the normalization constants.
\begin{prp}
\label{prp:counting}
Let $k \in \Zbb^3 \backslash \{0 \}$, $M = N^{1/3 + \epsilon}$ for an $0< \epsilon < 1/3$. Then, for $0 <\delta < 1/6- \epsilon/2$ and for all $\alpha \in \Ikp$, we have 
\[v_{\alpha} (k)^2 =  \sigma(p_\alpha) \, | \hat{k} \cdot \hat\omega_\alpha |  \, \left( 1 + \mathcal{O}\left(\sqrt{M}N^{-\frac{1}{3}+\delta}\right) \right)\,,\]
where $\sigma(p_\alpha) = \frac{4\pi}{M} + \Ocal(N^{-1/3}M^{-1/2})$ is the surface area of the patch $p_\alpha$ on the unit sphere.
\end{prp}
Due to the cutoff $\hat{\omega}_\alpha\cdot \hat{k} \geq N^{-\delta}$ imposed through the index set $\Ikp$, it immediately follows that there exists a constant\footnote{We use the symbol $C$ for positive constants, the value of which may change from line to line. All constants $C$ are independent of patch indices ($\alpha$, $\beta$ etc.), of momenta ($k$, $p$, $h$ etc.) and most importantly of $N$ (and also of $\hbar$, $M$, and $\nfrak$). They may however depend on $R$ and $\sup_k \hat{V}(k)$.} $C$ such that
\begin{equation}\label{eq:nNM} n_{\alpha, k} \geq C \nfrak, \quad \text{where } \nfrak (N,M) := \frac{N^{1/3-\delta/2}}{\sqrt{M}} \;. \end{equation}

\section{Construction of the Trial State}
In this section we shall introduce the trial state that will produce the upper bound in our main result, Theorem \ref{thm:main}. To begin, let us show that the particle--hole operators $b_{\alpha, k}$ defined in \eqref{eq:bka} behave as almost-bosonic operators when acting on Fock space vectors containing only few fermions.

\subsection{Particle-Hole Creation via Almost-Bosonic Operators}
Recall the definition of $\north$ given after \eqref{eq:north}. For $k \in \north$, let
\[ \Ikm := \Ical_{-k}^+ = \big\{ \alpha = 1, \ldots, M : \hat{\omega}_\alpha \cdot \hat{k} \leq - N^{-\delta} \big\}\;. \]
We shall also set $\mathcal{I}_{k} = \mathcal{I}_{k}^{+} \cup \mathcal{I}_{k}^{-}$. To unify notation, we define
\begin{equation}\label{eq:rigorouscoperators}c^*_\alpha(k) := \left\{ \begin{array}{lr} b^*_{\alpha,k} & \text{ for } \alpha \in \Ikp \\ b^*_{\alpha,-k} & \text{ for } \alpha \in 
\Ikm \end{array}\right.\;.\end{equation}

\begin{lem}[Approximate CCR]\label{lem:ccr}
Let $k,l \in\north$. Let $\alpha \in \Ik$ and $\beta \in \Il$. Then
\begin{equation}\begin{split}
[c_\alpha(k),c_\beta(l)] & = 0 = [c^*_\alpha(k),c^*_\beta(l)]\,,\\
[c_\alpha(k),c^*_\beta(l)] & = \delta_{\alpha,\beta}\left( \delta_{k,l} + \Ecal_\alpha(k,l) \right)\,. 
\end{split}\label{eq:commcc}\end{equation}
The operator $\Ecal_\alpha(k,l)$ commutes with $\Ncal$, and satisfies the bound
\begin{equation}
\label{eq:ccrerror}\norm{ \mathcal{E}_\alpha(k,l) \psi} \leq \frac{2}{n_{\alpha,k} n_{\alpha,l}} \norm{\Ncal\psi}\;, \qquad \forall \psi \in \fock\,.
\end{equation}
The same estimate holds for $\Ecal_\alpha^*(k,l) = \Ecal_\alpha(l,k)$.
\end{lem}
\begin{proof} The two identities  on the first line of \eqref{eq:commcc} are obvious. We prove the second line.
\paragraph{First case: $\alpha \in \Ikp$ and $\beta \in \Ical_l^{+}$.} We have
 \begin{equation}
 \label{eq:commut}[c_\alpha(k),c^*_\beta(l)] = [b_{\alpha,k},b^*_{\beta,l}]\,.
 \end{equation} 
 From the definition it is clear that $b$ and $b^*$ operators belonging to different patches commute, explaining the $\delta_{\alpha,\beta}$-factor. Thus, from now on $\alpha =\beta$. By the CAR,
\begin{equation}
\label{eq:hphp}
[a_{h_1} a_{p_1},a^*_{p_2} a^*_{h_2}] = \delta_{h_1,h_2}\delta_{p_1,p_2} - a^*_{p_2} a_{p_1} \delta_{h_1,h_2} - a^*_{h_2} a_{h_1} \delta_{p_1,p_2}\,.
\end{equation}
The first term in \eqref{eq:hphp} gives the following contribution to the commutator \eqref{eq:commut}:
 \[
 \begin{split}& n_{\alpha,k}^{-1} n_{\alpha,l}^{-1} \sum_{\substack{p_1 \in B_\textnormal{F}^c \cap B_\alpha \\ h_1 \in B_\textnormal{F} \cap B_\alpha}} \sum_{\substack{p_2 \in B_\textnormal{F}^c \cap B_\alpha \\ h_2 \in B_\textnormal{F} \cap B_\alpha}} \delta_{h_1,h_2}\delta_{p_1,p_2} \delta_{p_1-h_1,k}\delta_{p_2-h_2,l}\\
 & = n_{\alpha,k}^{-1} n_{\alpha,l}^{-1} \sum_{\substack{p_1 \in B_\textnormal{F}^c \cap B_\alpha \\ h_1 \in B_\textnormal{F} \cap B_\alpha}} \delta_{p_1-h_1,k} \delta_{p_1-h_1,l} = n_{\alpha,k}^{-2} \delta_{k,l} \sum_{\substack{p \in B_\textnormal{F}^c \cap B_\alpha \\ h \in B_\textnormal{F} \cap B_\alpha}} \delta_{p-h,k} = \delta_{k,l}\,.
 \end{split}
 \]
The two remaining terms in \eqref{eq:hphp} produce the error term
\begin{equation}\label{eq:errorerror}
 \begin{split} &- \sum_{\substack{h_1, h_2 \in B_\textnormal{F} \cap B_\alpha\\p \in B_\textnormal{F}^c \cap B_\alpha}} \frac{\delta_{p - h_1,k} \delta_{p-h_2,l}}{n_{\alpha,k} n_{\alpha,l}} a^*_{h_2} a_{h_1} - \sum_{\substack{p_1, p_2 \in B_\textnormal{F}^c \cap B_\alpha\\h\in B_\textnormal{F}\cap B_\alpha}} \frac{\delta_{p_1-h,k}\delta_{p_2-h,l}}{n_{\alpha,k} n_{\alpha,l}}a^*_{p_2} a_{p_1}\\
 & \hspace{5cm} =: \Ecal_{1}(\alpha,k,l) + \Ecal_{2}(\alpha,k,l) =: \Ecal(\alpha,k,l)\,.\end{split}
 \end{equation}
 In the present case, the error term in the lemma is $\Ecal_\alpha(k,l) := \Ecal(\alpha,k,l)$.
Let us only consider the second term in the left-hand side; the first can be controlled in the same way. Setting $\omega^{(\alpha)} := \sum_{h \in B_\textnormal{F}\cap B_\alpha} \lvert f_h\rangle \langle f_h\rvert$ and $u^{(\alpha)} := \sum_{p \in B_\textnormal{F}^c \cap B_\alpha} \lvert f_p\rangle \langle f_p \rvert$, we have
 \begin{align*}
 \di\Gamma\Big( u^{(\alpha)} e^{ilx} \omega^{(\alpha)} e^{-ikx} u^{(\alpha)} \Big) & = \sum_{p_1,p_2 \in B_\textnormal{F}^c \cap B_\alpha} a^*_{p_1} a_{p_2} \langle f_{p_2}, e^{ilx} \Big[\sum_{h\in B_\textnormal{F}\cap B_\alpha} \lvert f_{h}\rangle \langle f_{h}\rvert \Big] e^{-ikx} \lvert f_{p_1}\rangle \\
 & = \sum_{\substack{p_1,p_2\in B_\textnormal{F}^c\cap B_\alpha\\h \in B_\textnormal{F} \cap B_\alpha}} a^*_{p_1} a_{p_2} \delta_{p_2-h,l} \delta_{p_1-h,k}\,.
 \end{align*}
 Recall also, for the second quantization of any bounded one-particle operator, the standard bound $\norm{\di\Gamma(A)\psi} \leq \norm{A}\OP\norm{\Ncal \psi}$ for all $\psi \in \fock$, with $\norm{A}\OP$ the operator norm.
 Consequently
\[
\norm{ \Ecal_2(\alpha,k,l) \psi } = \Big\| \frac{1}{n_{\alpha,k} n_{\alpha,l}}\di\Gamma\Big( u^{(\alpha)} e^{ilx} \omega^{(\alpha)} e^{-ikx} u^{(\alpha)} \Big) \Big\| \leq \frac{1}{n_{\alpha,k} n_{\alpha,l}} \norm{ \Ncal\psi }
\]
 since $\norm{u^{(\alpha)} e^{ilx} \omega^{(\alpha)} e^{-ikx} u^{(\alpha)}}\OP \leq 1$. 

 \paragraph{Second case: $\alpha \in \Ikm$, $\beta \in \Ical_l^{-}$.} This case is treated like the first case, recalling that
  \[[c_\alpha(k),c^*_\beta(l)] = [b_{\alpha,-k},b^*_{\beta,-l}]\,.\]
 In this case $\Ecal_\alpha(k,l) := \Ecal(\alpha,-k,-l)$, with the same bound as before.
 
\paragraph{Third case: $\alpha \in \Ikp$ and $\beta \in \Ical_l^{-}$, and vice versa.} For $\alpha \neq \beta$ the commutator vanishes, just like in the previous cases. So consider
$\alpha \in \Ikp$ and $\beta = \alpha \in \Ical^{-}_l = \Ical_{-l}^+$. We find
\begin{equation}\label{eq:bibi}[c_\alpha(k),c^*_\alpha(l)] = [b_{\alpha,k},b^*_{\alpha,-l}] = \delta_{k,-l} + \mathcal{E}(\alpha,k,-l) \,.\end{equation}
Since $\Ikp \cap \Ical_k^{-} = \emptyset$, $\alpha = \beta$ is possible only for $k \neq l$. Also $k = -l$ is excluded since $k,l \in \north$. Consequently $\delta_{k,l} = 0 = \delta_{k,-l}$, so \eqref{eq:bibi} agrees with the statement of the Lemma (if we set $\Ecal_\alpha(k,l) := \Ecal(\alpha,k,-l)$). The estimate of the error term remains the same.

\medskip
 
 It is obvious that $\Ecal(\alpha,k,l)$ commutes with $\Ncal$. This completes the proof of the lemma.
\end{proof}

The next lemma provides bounds for the $c_{\alpha}(k)$, $c_{\alpha}^{*}(k)$ operators that are similar to the usual bounds valid for bosonic creation and annihilation operators.

\begin{lem}[Bounds for Pair Operators]\label{lem:bosopbound}
Let $k \in \north$ and $\alpha \in \Ik$. Then,
\begin{equation}\label{eq:cbd}
\norm{c_{\alpha}(k)\psi} \leq \norm{\Ncal(B_\textnormal{F}\cap B_\alpha)^{1/2} \psi}\qquad \forall \psi\in \mathcal{F}\;,
\end{equation}
where $\Ncal(B) := \sum_{i \in B} a^*_i a_i$ for any set of momenta $B\subset \Zbb^{3}$. Furthermore, for $f \in \ell^2(\Ik)$ and $\psi \in \fock$, we have
 \begin{equation}\label{eq:cstarbound}
 \begin{split} 
 \norm{\sum_{\alpha \in \Ik} f(\alpha) c_\alpha(k) \psi} &\leq \Big( \sum_{\alpha \in \Ik} \lvert f(\alpha)\rvert^2 \Big)^{1/2} \norm{\Ncal^{1/2} \psi}
\\ 
 \norm{\sum_{\alpha \in \Ik} f(\alpha) c^*_\alpha(k) \psi} &\leq \Big( \sum_{\alpha \in \Ik} \lvert f(\alpha)\rvert^2 \Big)^{1/2} \norm{(\Ncal+1)^{1/2} \psi}\,.\end{split} \end{equation}
\end{lem}
\begin{proof}
Using $\norm{ a_{q} }\OP = 1$ we have
\begin{align*}
\norm{b_{\alpha,k} \psi} & \leq \frac{1}{n_{\alpha,k}} \sum_\patch{p}{h}{\alpha}\delta_{p-h,k} \norm{a_p a_h \psi} \leq \frac{1}{n_{\alpha,k}} \sum_\patch{p}{h}{\alpha}\delta_{p-h,k} \norm{a_h \psi} \\
& \leq \frac{1}{n_{\alpha,k}} \Big[ \sum_\patch{p}{h}{\alpha} \delta_{p-h,k} \Big]^{1/2} \Big[ \sum_\patch{p}{h}{\alpha} \delta_{p-h,k} \norm{a_h \psi}^2 \Big]^{1/2} = \langle \psi, \Ncal(B_\textnormal{F}\cap B_\alpha) \psi \rangle^{1/2}\,,
\end{align*}
recalling that by definition $\sum_\patch{p}{h}{\alpha} \delta_{p-h,k} = n_{\alpha,k}^2$. This proves \eqref{eq:cbd}. To prove the first inequality from \eqref{eq:cstarbound}, we use \eqref{eq:cbd} together with Cauchy-Schwarz,
\begin{align*}
 \Big\| \sum_{\alpha\in\Ik} \cc{f(\alpha)} c_\alpha(k) \psi \Big\|^2
& \leq  \sum_{\alpha \in \Ik} \lvert f(\alpha) \rvert^2  \sum_{\alpha'\in\Ik} \norm{c_{\alpha'}(k) \psi}^2\\
& \leq  \sum_{\alpha \in \Ik} \lvert f(\alpha) \rvert^2  \sum_{\alpha'\in\Ik} \norm{\Ncal(B_\textnormal{F} \cap B_{\alpha'})^{1/2} \psi}^2 \leq \sum_{\alpha \in \Ik} \lvert f(\alpha) \rvert^2 \langle \psi, \Ncal \psi\rangle\,.
\end{align*}

We now prove the second inequality from \eqref{eq:cstarbound}. By Lemma \ref{lem:ccr}, we have
\begin{align*}
& \Big\|\sum_{\alpha \in \Ik} f(\alpha) c^*_\alpha(k) \psi\Big\|^2 \\
& = \sum_{\alpha,\beta\in\Ik} \cc{f(\alpha)}f(\beta) \langle \psi, c^*_\beta(k) c_\alpha(k) \psi \rangle + \sum_{\alpha,\beta\in\Ik} \cc{f(\alpha)}f(\beta) \langle \psi, [ c_\alpha(k), c^*_\beta(k)]\psi \rangle \\
& = \Big\| \sum_{\alpha\in\Ik} \cc{f(\alpha)} c_\alpha(k) \psi \Big\|^2 + \sum_{\alpha,\beta\in\Ik} \cc{f(\alpha)}f(\beta) \langle \psi, \delta_{\alpha,\beta}\left(1 +\Ecal_\alpha(k,k) \right)\psi \rangle \\
& \leq  \sum_{\alpha \in \Ik} \lvert f(\alpha) \rvert^2  \sum_{\alpha'\in\Ik} \norm{c_{\alpha'}(k) \psi}^2 + \sum_{\alpha \in \Ik} \lvert f(\alpha)\rvert^2 \norm{\psi}^2 + \sum_{\alpha \in \Ik} \lvert f(\alpha)\rvert^2 \langle \psi, \Ecal_\alpha(k,k) \psi\rangle\,. \tagg{newnew}
\end{align*}
Consider the last term on the r.\,h.\,s. Recall from \eqref{eq:errorerror} that for $\alpha \in \Ikp$ we have
\[
\Ecal_\alpha(k,k) =  - \frac{1}{n_{\alpha,k}^2} \sum_{\substack{p\in B_\textnormal{F}^c \cap B_\alpha \\ h \in B_\textnormal{F} \cap B_\alpha}} \delta_{p-h,k} a^*_h a_h - \frac{1}{n_{\alpha,k}^2} \sum_{\substack{p\in B_\textnormal{F}^c \cap B_\alpha \\ h \in B_\textnormal{F} \cap B_\alpha}} \delta_{p-h,k} a^*_p a_p\,.
\]
Obviously $\langle \psi, \Ecal_\alpha(k,k) \psi \rangle \leq 0$. For $\alpha \in \Ikm$ we have $-k$ replacing $k$ on the r.\,h.\,s., again producing a negative semidefinite operator. Hence in \eqref{eq:newnew} we can drop the last summand for the purpose of an upper bound. Together with the first bound from \eqref{eq:cstarbound} this implies
\begin{align*}
\Big\| \sum_{\alpha \in \Ik} f(\alpha) c^*_\alpha(k) \psi \Big\|^2 
& \leq \sum_{\alpha \in \Ik} \lvert f(\alpha) \rvert^2 \langle \psi, \Ncal \psi\rangle + \sum_{\alpha \in \Ik} \lvert f(\alpha) \rvert^2 \langle \psi,\psi\rangle\,. \qedhere
\end{align*}
\end{proof}
\subsection{The Trial State}
In order to motivate the definition of the trial state, let us formally rewrite the correlation Hamiltonian $\mathcal{H}_{\textnormal{corr}}$ in terms of the almost-bosonic pair operators $c^{*}_{\alpha}(k)$ and $c_{\alpha}(k)$.
\paragraph{Bosonization of the Correlation Hamiltonian.} Inserting the decomposition \eqref{eq:patchdecomp} into \eqref{eq:QN0-1} we find
\[ \begin{split} Q_N^{\textnormal{B}} =  \frac{1}{2N} \sum_{k\in \Zbbn} \hat{V}(k) \Big[ &  2 \sum_{\alpha \in \Ikp} \sum_{\beta \in \Ikp} n_{\alpha,k} n_{\beta,k} b_{\alpha,k}^* b_{\beta,k}  + \sum_{\alpha \in \Ikp } \sum_{\beta \in \Ical^{+}_{-k}} n_{\alpha,k} n_{\beta, -k} b_{\alpha, k}^* b_{\beta, -k}^* \\ &   + \sum_{\alpha \in \Ical^{+}_{-k}} \sum_{\beta \in \Ikp} n_{\alpha,-k} n_{\beta, k} b_{\alpha, -k} b_{\beta, k} \Big] + \text{error terms,} \end{split} \]
where the error terms contain at least one $\rfrak_k$--operator (see the discussion following \eqref{eq:rk-errors} for the rigorous proof of their smallness). Recalling the definition \eqref{eq:vkdef} of  $v_{\alpha}(k)$ and the definition of the $c$ and $c^*$ operators \eqref{eq:rigorouscoperators}, we get 
\begin{equation}
\label{eq:Qccstar} \begin{split} Q_N^{\textnormal{B}} = \hbar \kappa^2 \sum_{k \in \north} & \lvert k\rvert \, \hat{V}(k)  \Big[ \sum_{\alpha \in \Ikp} \sum_{\beta \in \Ikp} v_{\alpha} (k) v_\beta (k) c_\alpha^* (k)  c_\beta (k)   \\ &  + \sum_{\alpha \in \Ikm} \sum_{\beta \in \Ikm} v_\alpha (-k) v_\beta (-k) c_\alpha^* (k) c_\beta (k) \\ & + \sum_{\alpha \in \Ikp} \sum_{\beta \in \Ikm} \left( v_\alpha (k) v_\beta (-k) c_\alpha^*(k) c_\beta^*(k) + \text{h.c.} \right)  \Big]  + \text{error terms}\end{split}
\end{equation}
where $\kappa = (3/4\pi)^{1/3} + \Ocal(N^{-1/3})$ is defined as in \eqref{eq:gauss}. 

\medskip 

Let us now consider the operator $\di\Gamma (uhu-\cc{v}\cc{h}v)$ appearing in the definition of the correlation Hamiltonian \eqref{eq:Hcorr}. To express $\di\Gamma (uhu - \bar{v} \bar{h} v)$ in terms of $c_\alpha (k), c_\alpha^* (k)$, we observe that, for $\alpha \in \Ikp$ and neglecting the contribution of the constant direct term and of the exchange operator $X$ on the r.\,h.\,s.\ of \eqref{eq:h-def} (they will be proven to be small)
\begin{align*}
 \di\Gamma (uhu-\cc{v}\cc{h}v) c^*_{\alpha} (k) \Omega  & \simeq \frac{1}{n_{\alpha,k}}\sum_{\substack{p\in B_\textnormal{F}^c \cap B_\alpha\\h \in B_\textnormal{F} \cap B_\alpha}} \frac{\hbar^2 (\lvert p\rvert^2 - \lvert h\rvert^2)}{2} a^*_p a^*_h \delta_{p-h,k} \Omega\\
 & = \frac{1}{n_{\alpha,k}}\sum_{\substack{p\in B_\textnormal{F}^c \cap B_\alpha\\h \in B_\textnormal{F} \cap B_\alpha}} \frac{\hbar^2 (p-h)\cdot(p+h)}{2} a^*_p a^*_h \delta_{p-h,k} \Omega\\
 & \simeq \frac{1}{n_{\alpha,k}}\sum_{\substack{p\in B_\textnormal{F}^c \cap B_\alpha\\h \in B_\textnormal{F} \cap B_\alpha}} \frac{\hbar^2 k\cdot 2\omega_\alpha}{2} a^*_p a^*_h \delta_{p-h,k} \Omega =  \hbar^2 k\cdot \omega_\alpha c^*_{\alpha} (k)\Omega\,,
\end{align*}
where we used the fact that, for $p,h \in B_\alpha$, $p \simeq \omega_\alpha \simeq h$. A similar computation for $\alpha \in \Ikm$ shows that
\begin{equation}
 \label{eq:diG-cc}
\di\Gamma(uhu-\cc{v}\cc{h}v) c^*_{\alpha} (k) \Omega \simeq {\hbar^2} |k \cdot \omega_\alpha| c^*_\alpha (k) \Omega \end{equation}
for all $\alpha \in \Ical_k = \Ikp \cup \Ikm$.  

If the operators $c^*_\alpha (k), c_\alpha (k)$ were bosonic creation and annihilation operators, satisfying canonical commutation relations, and if $\di\Gamma  (uhu-\cc{v}\cc{h}v)$ were quadratic in these operators, \eqref{eq:diG-cc} would lead us to 
\[
\di\Gamma  (uhu-\cc{v}\cc{h}v) \simeq \hbar^2 \sum_{k \in \north} \sum_{\alpha \in \Ical_k} |k \cdot \omega_\alpha| c_\alpha^* (k) c_\alpha (k)\,.
\]
Thus, Equations \eqref{eq:Qccstar} and \eqref{eq:diG-cc} suggest that, if restricted to states with few particles, the correlation Hamiltonian should be approximated by the Sawada-type effective Hamiltonian
\begin{equation}\label{eq:eff-H} \Hcal_\textnormal{eff} = \hbar \kappa \sum_{k \in \north} \lvert k\rvert h_\textnormal{eff} (k) \end{equation}
with
\begin{equation}\label{eq:heff} \begin{split}h_\textnormal{eff}(k) = \; &\sum_{\alpha \in\Ical_k} u_\alpha^2 (k) \,  c^*_{\alpha} (k) c_{\alpha} (k) + g (k) \bigg[ \Big( \sum_{\alpha \in \Ikp}\sum_{\beta \in \Ikm} v_\alpha(k) v_\beta(-k) c^*_{\alpha} (k) c^*_{\beta} (k) + \hc \Big) \\ 
&+\sum_{\alpha\in \Ikp}\sum_{\beta \in \Ikp} v_\alpha(k) v_\beta(k) c^*_{\alpha} (k) c_{\beta} (k)  + \sum_{\alpha \in \Ikm}\sum_{\beta \in\Ikm} v_\alpha(-k) v_\beta(-k) c^*_{\alpha} (k) c_{\beta} (k) \bigg]\;.
\end{split}\end{equation}
We defined
\begin{equation}\label{eq:def-u}
u_\alpha (k) := \lvert \hat{k} \cdot \hat{\omega}_\alpha\rvert^{1/2}, \qquad g(k) := \kappa  \hat{V} (k)\,.
\end{equation}  
(The main difference to Sawada's original Hamiltonian is that he treated pairs $a^*_p a^*_h$ as bosonic; our pair operators instead are delocalized over large patches, thus relaxing the Pauli principle and allowing a controlled bosonic approximation.)
If the operators $c_{\alpha}(k)$, $c_{\alpha}^{*}(k)$ were exactly bosonic, the effective Hamiltonian $\Hcal_{\textnormal{eff}}$ could be diagonalized via a bosonic Bogoliubov transformation. We provide the details of this computation in Appendix \ref{sec:bosonicapproximation}. The ground state of \eqref{eq:eff-H} would be given by  
\begin{equation}
\label{eq:bos-gs}
\xi = \exp \Big[ \frac{1}{2} \sum_{k \in \north} \sum_{\alpha , \beta \in \Ical_k} K_{\alpha, \beta} (k) c_\alpha^* (k) c_\beta^* (k) - \text{h.c.} \Big] \Omega\;,
\end{equation}
where, for every $k \in \north$, $K(k)$ is the $2I_k \times 2I_k$ matrix (with $I_k := \lvert \Ikp\rvert = \lvert\Ikm\rvert$) defined by
\begin{equation}\label{eq:Kk}K(k) := \log \lvert S_1^T (k)\rvert\,,\end{equation}
(the superscript $T$ denoting the transpose of the matrix) with
\begin{equation}\label{eq:S1}S_1 (k) := \left( \D (k) +\W (k) -\Wt (k) \right)^{1/2} E(k)^{-1/2}\,,\end{equation}
and
 \begin{equation*} E (k) := \left( (\D (k) +\W (k) -\Wt (k) )^{1/2} (\D (k) +\W (k) +\Wt (k)) (\D (k)+\W (k)- \Wt (k) )^{1/2} \right)^{1/2} \end{equation*}
and, recalling the definition \eqref{eq:vkdef} of $v_\alpha(k)$, 
\begin{equation}
\label{eq:blocks2}\begin{split}
D (k) & = \diag(u_\alpha^2 (k) : \alpha \in \Ik)\,, \\
\W_{\alpha,\beta} (k) & = \left\{ \begin{array}{cl} g (k) v_\alpha(k) v_\beta(k) & \text{for } \alpha,\beta \in \Ikp \\ g (k) v_\alpha(-k) v_\beta(-k) & \text{for } \alpha,\beta \in \Ikm \\ 0 & \text{for } \alpha \in \Ikp, \beta \in \Ikm \text{ or } \alpha \in \Ikm, \beta \in \Ikp,\end{array} \right. \\
\Wt_{\alpha,\beta} (k) & = \left\{ \begin{array}{cl} g (k) v_\alpha(k) v_\beta(-k) & \text{for } \alpha \in \Ikp, \beta \in \Ikm \\ g (k) v_\alpha(-k) v_\beta(k) & \text{for } \alpha \in \Ikm,\beta \in \Ikp \\ 0 & \text{for } \alpha,\beta \in \Ikp \text{ or } \alpha,\beta \in \Ikm \,.\end{array} \right.
\end{split}
\end{equation}
However, the particle--hole pair operators $c_\alpha^* (k), c_\alpha (k)$ are not exactly bosonic, and thus the ground state vector of \eqref{eq:eff-H} is \emph{not} given by \eqref{eq:bos-gs}. Nevertheless, by Lemma \ref{lem:ccr}, it is reasonable to expect that the true ground state of $\mathcal{H}_{\textnormal{corr}}$ will be energetically close to $\xi$, provided that the number of fermions in $\xi$ is small. This last fact is proven in Section \ref{sec:AB}. 

\bigskip

Motivated by the above heuristic discussion, we define as trial state for the full many-body problem the fermionic Fock space vector
\begin{equation}\label{eq:trial}
\psi_\textnormal{trial} := R_{\omega_{\textnormal{pw}}} T \Omega\;, \qquad T := e^{B}\;,\qquad B := \frac{1}{2}\sum_{k\in \north}\sum_{\alpha,\beta\in\Ik} K(k)_{\alpha,\beta} c^*_\alpha(k) c^*_\beta(k) - \hc
\end{equation}
Notice that $B^* = -B$, so $T$ is unitary and hence $\norm{\psi_\textnormal{trial}} = 1$. We have to check that $\psi_\textnormal{trial}$ is an $N$-particle state.  In fact, writing $\xi := R^*_{\omega_\textnormal{pw}} \psi_\textnormal{trial}$, we have
\[\Ncal \psi_\textnormal{trial} = R_{\omega_\textnormal{pw}} \Big[\sum_{p\in B_\textnormal{F}^c} a^*_p a_p + \sum_{h \in B_\textnormal{F}} a_h a^*_h \Big] \xi = R_{\omega_\textnormal{pw}}(\Ncal_\textnormal{p}-\Ncal_\textnormal{h}) \xi + N\psi_\textnormal{trial}\,,\] which shows that $\psi_\textnormal{trial}$ is an eigenvector of $\Ncal$ with eigenvalue $N$ if and only if $\xi$ is an eigenvector of $\Ncal_\textnormal{p}-\Ncal_\textnormal{h}$ with eigenvalue $0$. This is the content of the next lemma.
\begin{lem}[Particle-Hole Symmetry]\label{lem:phsymmetry}
For $\xi$ as in \eqref{eq:bos-gs} we have $(\Ncal_\textnormal{p} - \Ncal_\textnormal{h}) \xi = 0$.
\end{lem}
\begin{proof}
Let $\xi_\lambda = T_\lambda \Omega$, with $
T_{\lambda} = e^{\lambda B}$ for $\lambda \in [0,1]$.
Then $\xi_1 = \xi$, $\xi_0 = \Omega$, and thus  
\[
\norm{ (\Ncal_\textnormal{p} - \Ncal_\textnormal{h}) \xi}^2 = \int_0^1 \di\lambda \, \frac{\di}{\di\lambda} \langle \xi_\lambda, (
\Ncal_\textnormal{p} - \Ncal_\textnormal{h})^2 \xi_\lambda \rangle =   \int_0^1 \di\lambda \, \langle \xi_\lambda, \left[ (
\Ncal_\textnormal{p} - \Ncal_\textnormal{h})^2 , B \right] \xi_\lambda \rangle = 0
\]
because $[ \Ncal_\textnormal{p} - \Ncal_\textnormal{h} , c_\alpha^* (k)] = 0 = [ \Ncal_\textnormal{p} - \Ncal_\textnormal{h} , c_\alpha (k) ]$ implies $[\Ncal_\textnormal{p} - \Ncal_\textnormal{h} , B] = 0$.
\end{proof}

\subsection{Approximate Bosonic Bogoliubov Transformations}\label{sec:AB}

Our next task is to evaluate the energy of the fermionic many-body trial state $\psi_\textnormal{trial} = R_{\omega_{\textnormal{pw}}} \xi = R_{\omega_{\textnormal{pw}}} T \Omega$, which by \eqref{eq:cLN} and \eqref{eq:Hcorr} reduces to calculating $\langle \xi, \mathcal{H}_{\textnormal{corr}} \xi \rangle$. To do so, we will need some properties of the operator $T$, which are going to be proven in this section. More generally, we shall consider the one-parameter family of unitaries $T_{\lambda} = e^{\lambda B}$, with $B$ defined in \eqref{eq:trial}.

The next proposition establishes that the action of $T_\lambda$ approximates a bosonic Bogoliubov transformation.  
\begin{prp}[Approximate Bogoliubov Transformation]\label{lem:bog1}
Let $\lambda\in[0,1]$. Let $l \in\north$ and $\gamma \in \Il = \Ical_l^+ \cup \Ical_l^-$. Then
\begin{align*}
T^*_\lambda c_\gamma(l) T_\lambda &= \sum_{\alpha\in\Il}\cosh(\lambda K(l))_{\alpha,\gamma} c_\alpha(l) + \sum_{\alpha\in\Il} \sinh(\lambda K(l))_{\alpha,\gamma} c^*_\alpha(l) + \mathfrak{E}_\gamma(\lambda,l)\,,
\end{align*}
where the error operator $\mathfrak{E}_\gamma(\lambda,l)$ satisfies, for all $\psi \in \mathcal{F}$, the bound
\begin{equation}\label{eq:bogerrorbound}
\Big[ \sum_{\gamma \in \Il} \norm{\Efrak_\gamma(\lambda,l)\psi}^2 \Big]^{1/2} \leq \frac{C}{\nfrak^2} \sup_{\tau\in[0,\lambda]}\norm{ (\Ncal+2)^{3/2} T_\tau\psi}  \, e^{\lambda\norm{K(l)}\HS} \sum_{k\in\north} \norm{K(k)}\HS.
\end{equation}
Here $\nfrak = N^{1/3-\delta/2} M^{-1/2}$ as defined in \eqref{eq:nNM}, and $\norm{K(k)}\HS$ denotes the Hilbert-Schmidt norm of the matrix $K(k)$. 
The same estimate holds for $\mathfrak{E}^*_\gamma(\lambda,l)$.
\end{prp}
\begin{proof}
We start from the Duhamel formula
\[
T^*_\lambda c_\gamma(l) T_\lambda = c_\gamma(l) + \int_0^\lambda \di\tau \, T^*_\tau [c_\gamma(l),B]T_\tau\,.
\]
From Lemma \ref{lem:ccr}, the commutator is given by 
\[
[c_\gamma(l),B]  = \sum_{k\in\north}\frac{1}{2} \sum_{\alpha,\beta\in\Ik} K(k)_{\alpha,\beta} [c_\gamma(l),c^*_\alpha(k) c^*_\beta(k)]  = \sum_{\alpha\in\Il} K(l)_{\gamma,\alpha} c^*_\alpha(l) + \mathfrak{e}_\gamma(l)
\]
where the error term is
\begin{equation}
\label{eq:errc}
\mathfrak{e}_\gamma(l) := \sum_{k \in \north}\frac{\chi_{\Ik}(\gamma)}{2} \sum_{\alpha\in\Ik} K(k)_{\gamma,\alpha} \left( \Ecal_\gamma(k,l) c^*_\alpha(k) + c^*_\alpha(k) \Ecal_\gamma(k,l)\right)\,,
\end{equation}
with $\chi_{\Ik}$ the indicator function of the set $\Ik = \Ikp \cup \Ikm$ and $\Ecal_\gamma (k,l)$ bounded as in \eqref{eq:ccrerror}. Thus
\[
\begin{split}
T^*_\lambda c_\gamma(l) T_\lambda & = c_\gamma(l) + \sum_{\alpha\in\Il} K(l)_{\gamma,\alpha} \int_0^\lambda \di\tau T^*_\tau c^*_\alpha(l) T_\tau + \int_0^\lambda \di\tau T^*_\tau \mathfrak{e}_\gamma(l) T_\tau\,, \\
T^*_\lambda c^*_\gamma(l) T_\lambda & = c^*_\gamma(l) + \sum_{\alpha\in\Il} K(l)_{\gamma,\alpha} \int_0^\lambda \di\tau T^*_\tau c_\alpha(l) T_\tau + \int_0^\lambda \di\tau T^*_\tau \mathfrak{e}^*_\gamma(l) T_\tau\,.
\end{split}
\]
We iterate $n_0$ times by plugging the second equation into the second summand on the r.\,h.\,s.\ of the first equation and so forth. The simplex integrals produce factors $1/n!$, so we obtain\footnote{The range of summation used in the matrix multiplication is clear from the momentum dependence of $K$, e.\,g., $\left(K(l)^2\right)_{\gamma,\alpha} = \sum_{\beta\in\Il} K(l)_{\gamma,\beta} K(l)_{\beta,\alpha}$, for $\gamma, \alpha \in \Il$.} 
\begin{align*}
T^*_\lambda c_\gamma(l) T_\lambda & = c_\gamma(l) \\
		  & \quad + \sum_{\alpha\in\Il} \lambda K(l)_{\gamma,\alpha} c^*_\alpha(l) + \int_0^\lambda \di \tau_1 T^*_{\tau_1} \mathfrak{e}_\gamma(l) T_{\tau_1} \\
		  & \quad + \frac{1}{2!}\sum_{\alpha\in\Il} \left(\lambda^2 K(l)^2 \right)_{\gamma,\alpha} c_\alpha(l) + \sum_{\alpha \in \Il} K(l)_{\gamma,\alpha} \int_0^\lambda \di \tau_1\int_0^{\tau_1} \di \tau_2 T^*_{\tau_2} \mathfrak{e}^*_\alpha(l) T_{\tau_2} \\
		  & \quad + \frac{1}{3!}\!\sum_{\alpha \in \Il} \left(\lambda^3 K(l)^3\right)_{\gamma,\alpha} c^*_\alpha(l) + \sum_{\alpha\in\Il} \left( K(l)^2\right)_{\gamma,\alpha} \int_0^\lambda\di\tau_1 
		 \int_0^{\tau_1} \di \tau_2 \int_0^{\tau_2} \di\tau_3 T^*_{\tau_3} \mathfrak{e}_{\alpha}(l) T_{\tau_3}\\
		  & \quad+ \ldots \\
 		  & \quad + \sum_{\alpha\in\Il} \left( K(l)^{n_0} \right)_{\gamma,\alpha} \int_0^\lambda \di \tau_1 \int_0^{\tau_1} \di \tau_2 \ldots \int_0^{\tau_{n_0-1}} \di \tau_{n_0} T^*_{\tau_{n_0}} c^\natural_\alpha(l) T_{\tau_{n_0}} \,.
\end{align*}
Here we introduced the notation $c^\natural_\alpha(l)$, which in this formula means $c^*_\alpha(l)$ for $n_0$ odd, and $c_\alpha(l)$ for $n_0$ even.
The left term on every line is the leading term, the right term on every line is an error term which will be controlled later. The very last line is the `head' of the iteration after $n_0$ steps; we are going to control the expansion as $n_0 \to \infty$, showing that the head vanishes.

Notice that leading terms are of the form of an exponential series $\lambda^n K(l)^n/n!$ but intermittently with $c$ and $c^*$. Separating creation and annihilation operators, we reconstruct $\cosh(\lambda K(l))$ and $\sinh(\lambda K(l))$. We find 
\begin{align*}
T^*_\lambda c_\gamma(l) T_\lambda &= \sum_{\alpha\in\Il}\cosh(\lambda K(l))_{\gamma,\alpha} c_\alpha(l) + \sum_{\alpha\in\Il} \sinh(\lambda K(l))_{\gamma,\alpha} c^*_\alpha(l) + \mathfrak{E}_\gamma(\lambda,l)
\end{align*}
where, for an arbitrary $n_0 \in \mathbb{N}$, 
\begin{align*}
\mathfrak{E}_\gamma(\lambda,l) & := \sum_{\alpha\in\Il} \sum_{n=0}^{n_0-1} \left(K(l)^n\right)_{\gamma,\alpha} \int_0^\lambda \di \tau_1\cdots \int_0^{\tau_{n}} \di \tau_{n+1} T^*_{\tau_{n+1}} \mathfrak{e}^{\natural}_{\alpha}(l) T_{\tau_{n+1}} \\
& \quad + \sum_{\alpha\in\Il} \left( K(l)^{n_0} \right)_{\gamma,\alpha} \int_0^\lambda \di \tau_1 \int_0^{\tau_1} \di \tau_2 \ldots \int_0^{\tau_{n_0-1}} \di \tau_{n_0} T^*_{\tau_{n_0}} c^\natural_\alpha(l) T_{\tau_{n_0}} \\ 
&\quad - \sum_{\alpha \in \Il} \sum_{n = n_0}^\infty \frac{\lambda^n (K(l)^n)_{\gamma,\alpha}}{n!} c^\natural_{\alpha} (l)  
\,.
\end{align*}
(In every summand, $\mathfrak{e}_\alpha(l)$ and $c_\alpha (l)$ appear for even $n$ or $n_0$, $\mathfrak{e}_\alpha^*(l)$ and $c_\alpha^* (l)$ for odd $n$ or $n_0$.) Notice that for any function $f: \Rbb \to \Rbb$ the simplex integration simplifies to
\[
\int_0^\lambda \di \tau_1 \int_0^{\tau_1} \di \tau_2 \cdots \int_0^{\tau_{n}} \di \tau_{n+1} f(\tau_{n+1}) = \int_0^\lambda \frac{(\lambda-\tau)^n}{n!} f(\tau) \di \tau\,.
\] 
Therefore, for all $\psi \in \fock$ we have
\begin{align*}
& \norm{\Efrak_\gamma(\lambda,l)\psi} \\
& \leq \Big\| \sum_{\alpha \in \Il} \sum_{n=0}^{n_0-1} \left( K(l)^n \right)_{\gamma,\alpha} \int_0^\lambda \di \tau \frac{(\lambda-\tau)^n}{n!} T^*_\tau\efrak_\alpha^\natural(l) T_\tau \psi \Big\|\\
& \quad + \Big\| \sum_{\alpha\in\Il} \left( K(l)^{n_0} \right)_{\gamma,\alpha} \int_0^\lambda \di \tau \frac{(\lambda-\tau)^{n_0-1}}{(n_0-1)!} T^*_{\tau} c^\natural_\alpha(l) T_{\tau}\Big\| + \Big\| \sum_{\alpha \in \Il} \sum_{n=n_0}^{\infty} \frac{\lambda^n (K(k)^n)_{\gamma,\alpha}}{n!} c^\natural_\alpha (l) \psi\Big\|\;; \\
\intertext{using the explicit expression \eqref{eq:errc} for $\efrak_\alpha^\natural(l)$ we have
}
& \leq \Big\|  \sum_{\alpha\in\Il} \sum_{n=0}^{n_0-1} \left( K(l)^n \right)_{\gamma,\alpha} \int_0^\lambda \di \tau \frac{(\lambda-\tau)^n}{n!} \sum_{k\in \north} \frac{\chi_{\Ik}(\alpha)}{2} \sum_{\delta \in \Ik} K(k)_{\alpha,\delta}  T^*_\tau (\Ecal_\alpha(k,l) c^*_\delta(k) )^{\natural} T_\tau \psi  \Big\| \\
& \quad + \Big\|  \sum_{\alpha\in\Il} \sum_{n=0}^{n_0-1} \left( K(l)^n \right)_{\gamma,\alpha} \int_0^\lambda \di \tau \frac{(\lambda-\tau)^n}{n!} \sum_{k\in \north} \frac{\chi_{\Ik}(\alpha)}{2} \sum_{\delta \in \Ik} K(k)_{\alpha,\delta} T^*_\tau c^*_\delta(k) \Ecal_\alpha(k,l)^{\natural} T_\tau \psi  \Big\| \\
& \quad + \Big\| \sum_{\alpha\in\Il} \left( K(l)^{n_0} \right)_{\gamma,\alpha} \int_0^\lambda \di \tau \frac{(\lambda-\tau)^{n_0-1}}{(n_0-1)!} T^*_{\tau} c^\natural_\alpha(l) T_{\tau}\Big\| \\ &\quad + \Big\| \sum_{\alpha \in \Il} \sum_{n=n_0}^{\infty} \frac{\lambda^n (K(k)^n)_{\gamma,\alpha}}{n!} c^\natural_\alpha (l) \psi \Big\|\\ 
& =: \text{A}_\gamma  + \text{B}_\gamma + \text{C}_\gamma + \text{D}_\gamma \,.
\end{align*}
Let us start by estimating $\text{B}_\gamma$. We shall neglect the symbol $\natural$; the bounds are the same whether for the operator or its adjoint. Using also \eqref{eq:cstarbound}, we get
\begin{align*}
\text{B}_\gamma & \leq \sum_{\alpha\in\Il}\sum_{n=0}^\infty \lvert \left( K(l)^n\right)_{\gamma,\alpha} \rvert \int_0^\lambda \di \tau \frac{(\lambda-\tau)^n}{n!} \sum_{k\in \north} \frac{\chi_{\Ik}(\alpha)}{2} \Big\| \sum_{\delta \in \Ik} K(k)_{\alpha,\delta} c^*_\delta(k) \Ecal_\alpha(k,l)T_\tau\psi \Big\| \\
& \leq \sum_{\alpha\in\Il}\sum_{n=0}^\infty \lvert \left( K(l)^n\right)_{\gamma,\alpha} \rvert \int_0^\lambda \di \tau \frac{(\lambda-\tau)^n}{n!} \sum_{k\in \north} \frac{\chi_{\Ik}(\alpha)}{2} \Big[ \sum_{\delta \in \Ik} \lvert K(k)_{\alpha,\delta} \rvert^2 \Big]^{1/2}\\
& \hspace{10cm}\times\norm{ (\Ncal+1)^{1/2} \Ecal_\alpha(k,l)T_\tau\psi}\,;
\intertext{pulling $\Ecal_\alpha(k,l)$ through $(\Ncal+1)^{1/2}$ to the front and then using \eqref{eq:ccrerror}, we get}
\text{B}_\gamma & \leq \sum_{\alpha\in\Il}\sum_{n=0}^\infty \lvert \left( K(l)^n\right)_{\gamma,\alpha} \rvert \int_0^\lambda\! \di \tau \frac{(\lambda-\tau)^n}{n!}\!\! \sum_{k\in \north} \chi_{\Ik}(\alpha) \Big[ \sum_{\delta \in \Ik} \lvert K(k)_{\alpha,\delta} \rvert^2 \Big]^{1/2} \frac{\norm{ \Ncal(\Ncal+1)^{1/2} T_\tau\psi}}{n_{\alpha,k} n_{\alpha,l}}\\
& \leq \frac{\sup_{\tau \in [0,\lambda]} \norm{ (\Ncal+1)^{3/2} T_\tau\psi}}{\nfrak^2} \sum_{n=0}^\infty \frac{\lambda^{n+1}}{(n+1)!} \sum_{k\in \north} \sum_{\alpha \in \Il \cap \Ik}\lvert \left( K(l)^n\right)_{\gamma,\alpha} \rvert   \Big[ \sum_{\delta \in \Ik} \lvert K(k)_{\alpha,\delta} \rvert^2 \Big]^{1/2} \\
& \leq \frac{\sup_{\tau \in [0,\lambda]} \norm{ (\Ncal+1)^{3/2} T_\tau\psi}}{\nfrak^2}
\\ &\hspace{.3cm} \times \Bigg[ \lambda \sum_{k \in \north} \Big[ \sum_{\delta \in \Ik} \lvert K(k)_{\gamma,\delta} \rvert^2 \Big]^{1/2}
+  \sum_{n= 1}^\infty \frac{\lambda^{n+1}}{(n+1)!} \Big[\! \sum_{\alpha \in \Il \cap \Ik}\!\! \lvert \left( K(l)^n\right)_{\gamma,\alpha} \rvert^2 \Big]^{1/2} \sum_{k\in \north} \| K (k) \|\HS \Bigg] 
\end{align*}
where we used $n_{\alpha,k} \geq \mathfrak{n}$ as established in \eqref{eq:nNM}, we separated the term with $n=0$ and, for $n \geq 1$, we applied Cauchy-Schwarz to the sum over $\alpha$. Again by Cauchy-Schwarz, we obtain 
\[
\Big[ \sum_{\gamma \in \Il} \text{B}_\gamma^2 \Big]^{1/2} \leq 
C \, \frac{\sup_{\tau \in [0,\lambda]} \norm{ (\Ncal+1)^{3/2} T_\tau\psi}}{\nfrak^2} \, e^{\lambda \| K(k) \|\HS} \sum_{k\in \north} \| K (k) \|\HS \,. 
\]
The error term $\text{A}_\gamma$ can be treated similarly (applying first \eqref{eq:ccrerror} and then \eqref{eq:cstarbound}). We find 
\[
\Big[ \sum_{\gamma \in \Il} \text{A}_\gamma^2 \Big]^{1/2} \leq C \, \frac{\sup_{\tau \in [0,\lambda]} \norm{ (\Ncal+1)^{3/2} T_\tau\psi}}{\nfrak^2} \, e^{\lambda \| K(k) \|\HS} \sum_{k\in \north} 
\| K (k) \|\HS \,.
\]
As for the term $\text{C}_\gamma$, it is controlled with \eqref{eq:cstarbound} by
\begin{align*}
\text{C}_\gamma & \leq \Big\| \sum_{\alpha\in\Il} \left( K(l)^{n_0} \right)_{\gamma,\alpha} \int_0^\lambda \di \tau \frac{(\lambda-\tau)^{n_0-1}}{(n_0-1)!} T^*_{\tau} c^\natural_\alpha(l) T_{\tau}\Big \| \\
& \leq \int_0^\lambda \di \tau \frac{(\lambda-\tau)^{n_0-1}}{(n_0-1)!} \Big\| \sum_{\alpha\in\Il} \left( K(l)^{n_0} \right)_{\gamma,\alpha} c^\natural_\alpha(l) T_\tau \psi\Big\| \\
& \leq \frac{\lambda^{n_0}}{n_0!} \Big( \sum_{\alpha\in\Il} \Big| \Big( K(l)^{n_0} \Big)_{\gamma,\alpha} \Big|^2 \Big)^{1/2} \sup_{\tau \in [0,1]} \norm{ (\Ncal+1)^{1/2} T_\tau \psi}\,.
\end{align*}
This implies that 
\[
\Big[ \sum_{\gamma \in \Il} \text{C}_\gamma^2 \Big]^{1/2} \leq C \, \frac{\lambda^{n_0} \| K(l) \|\HS^{n_0}}{n_0!} \, \sup_{\tau \in [0,\lambda]} \norm{ (\Ncal+1)^{1/2} T_\tau\psi}\,.
\]
Finally, the term $\text{D}_\gamma$ can be bounded by 
\[
\begin{split} \text{D}_\gamma &\leq \sum_{n \geq n_0} \frac{\lambda^n}{n!} \Big\| \sum_{\alpha \in \Il} (K(k)^n)_{\gamma,\alpha} c_\alpha^\natural (l) \psi\Big\| \leq \sum_{n \geq n_0} \frac{\lambda^n}{n!} \Big[ \sum_{\alpha \in \Il} | (K(k)^n)_{\gamma,\alpha}|^2 \Big]^{1/2} \| (\Ncal +1)^{1/2} \psi \| 
\end{split} 
\]
which leads us to 
\[
\Big[ \sum_{\gamma \in \Il} \text{D}_\gamma^2 \Big]^{1/2} \leq C \sum_{n \geq n_0} \frac{\lambda^n \| K(k) \|\HS^n}{n!} \, \| (\Ncal + 1)^{1/2} \psi \| \,.
\]
 Since all these bounds hold for any $n_0 \in \mathbb{N}$, we obtain
 \[
 \left[ \sum_{\gamma \in \Il} \norm{\Efrak_\gamma(\lambda,l)\psi}^2 \right]^{1/2} \leq C \, \frac{\sup_{\tau \in [0,\lambda]} \norm{ (\Ncal+1)^{3/2} T_\tau\psi}}{\nfrak^2} \, e^{\lambda \| K(k) \|\HS} \sum_{k\in \north} \| K (k) \|\HS \,. \qedhere 
 \]
\end{proof}
The next lemma provides the required bounds for the matrix $K(k)$, defined in \eqref{eq:Kk}. For later estimates it is important that this bound implies $K(k) = 0$ outside the support of $\hat{V}$. (Actually the constant $C$ here may be chosen independent of $V$.)
\begin{lem}[Bound on the Bogoliubov Kernel]\label{lem:inverses}
Let $k \in \north$. Then the matrices $E(k)$, $\D(k)+\W(k)-\Wt(k)$, and $\D(k)+\W(k)+\Wt(k)$, all defined in \eqref{eq:blocks2}, are strictly positive. Let $K(k)$ be defined by \eqref{eq:Kk}.
 Then we have
\begin{equation}\label{eq:kernelbound} \norm{ K (k)}\HS \leq \norm{ K (k)}\TR \leq C \hat{V} (k)\,, \end{equation} 
where $\norm{ K (k)}\TR$ denotes the trace norm of the matrix $K(k)$.
\end{lem}
\begin{proof}
Recall that $I_k = |\Ikp| = |\Ikm|$ and that the matrix $K(k)$ is symmetric and has size $2I_k \times 2I_k$. All quantities in this proof depend on the same $k$, so we simplify notation by dropping this dependence where there is no risk of confusion, writing e.\,g., $\ik$ for $I_k$.

\medskip

To exhibit the block structure of the matrices, we
 map the indices $\Ikp$ to $\{1,\ldots, \ik\}$, and the indices $\Ikm$ to $\{\ik+1,\ldots 2\ik\}$. There are many such mappings, but due to the reflection symmetry of the patches ($B_{\alpha+M/2} = - B_\alpha$ and $\omega_{\alpha+M/2} = - \omega_\alpha$ in the original numbering), we can choose one such that $v_\alpha(k) = v_{\alpha+\ik}(-k)$. This implies that 
\begin{equation}
\label{eq:Wdef}\W = \begin{pmatrix} b & 0 \\ 0 & b \end{pmatrix}, \quad \Wt = \begin{pmatrix} 0& b\\ b & 0\end{pmatrix},
\end{equation}
where $b_{\alpha,\beta} = g v_\alpha(k) v_\beta(k)$ defines an $\ik\times \ik$-matrix. We drop the $k$-dependence from the notation and just write $v_\alpha = v_\alpha(k)$. In Dirac notation, where $\lvert v\rangle\langle v\rvert$ is the orthogonal projection onto $v = (v_1, \cdots , v_{\ik})$, we have $b= g \lvert v\rangle \langle v\rvert \in \Cbb^{\ik\times \ik}$.

Also $D_{\alpha,\alpha} = \lvert \hat{k}\cdot\hat{\omega}_\alpha\rvert$ is invariant under reflection at the origin and so $\D$ simplifies to
\[
\D = \begin{pmatrix} d & 0 \\ 0 & d \end{pmatrix}, \quad d = \diag(u_\alpha^2, \alpha = 1,\ldots,\ik)\,.
\]
Recalling the definition of the index set $\Ikp$ we notice that $u_{\alpha}^2 \geq N^{-\delta}$ for all $\alpha \in \{1,\ldots,\ik\}$, and thus $d$ is invertible. Since $b \geq 0$ (because $g \geq 0$), we find $d + 2b \geq d > 0$; hence also  $d+2b$ is invertible.

\medskip

To simplify the computation further, let
\[U = \frac{1}{\sqrt{2}}\begin{pmatrix} \id & \id \\ \id & -\id \end{pmatrix}\,,
\]
where $\id$ is the $\ik\times \ik$-identity matrix. Obviously $U^T = U = U^{-1}$, and it simultaneously blockdiagonalizes
\begin{equation}
\label{eq:blockdiagonalization}
U^T (\D+\W+\Wt) U = \begin{pmatrix} d+2b & 0 \\ 0 & d \end{pmatrix}, \quad U^T (\D+\W-\Wt) U = \begin{pmatrix} d & 0\\ 0 & d+2b \end{pmatrix}.\end{equation}
This shows that $\D+\W+\Wt$ and $\D+\W-\Wt$ are strictly positive, thus invertible, and have a positive square root. We also find 
\[
U^T E U = \begin{pmatrix} \left[ d^{1/2}(d+2b) d^{1/2} \right]^{1/2} & 0 \\ 0 & \left[ (d+2b)^{1/2} d (d+2b)^{1/2}\right]^{1/2} \end{pmatrix}.
\]
Both blocks are strictly positive; $E$ is therefore invertible and has a strictly positive operator square root.

\medskip

Now consider
\[
\lvert S_1\rvert^2 = S_1^T S_1 = E^{-1/2} (\D+\W-\Wt) E^{-1/2}.
\]
We find
\begin{equation}
\label{eq:S12} U^T \lvert S_1 \rvert^2 U = \begin{pmatrix} A_1 & 0 \\ 0 & A_2 \end{pmatrix}\end{equation}
with
\begin{equation}\label{eq:A1A2}
\begin{split} A_1 & := \left[ d^{1/2} (d+2b) d^{1/2} \right]^{-1/4} d \left[ d^{1/2} (d+2b) d^{1/2} \right]^{-1/4}, \\
A_2 & := \left[ (d+2b)^{1/2} d (d+2b)^{1/2} \right]^{-1/4} (d+2b) \left[ (d+2b)^{1/2} d (d+2b)^{1/2} \right]^{-1/4}.
\end{split} 
\end{equation}
Since $b$ is a positive operator, using operator monotonicity of the inverse and the square root, we find
\[
d^{1/2} \left[d^{1/2}(d+2b)d^{1/2}\right]^{-1/2} d^{1/2} \leq \id\,.
\]
Using the equality of the spectra $\sigma(AB) = \sigma(BA)$ for positive operators $A$ and $B$,
we conclude that 
\begin{align*}
\sigma(A_1)& =\sigma \left( d \left[d^{1/2}(d+2b)d^{1/2}\right]^{-1/2}\right) = \sigma \left( d^{1/2} \left[d^{1/2}(d+2b)d^{1/2}\right]^{-1/2} d^{1/2} \right)
\end{align*}
and therefore that \begin{equation}\label{eq:A1-bd}
A_1 \leq \id \, . \end{equation} 
Arguing similarly, we find that
\begin{equation}\label{eq:A2-bd} \id \leq  A_2
\, . \end{equation} 
We introduce the polar decomposition $S_1 = O \lvert S_1 \rvert$; a priori $O$ is a partial isometry, but since $S_1$ is invertible, $O$ is actually an orthogonal matrix. Then $\lvert S_1^T\rvert^2 = S_1 S_1^T = O \lvert S_1\rvert \lvert S_1\rvert^T O^T = O\lvert S_1\rvert^2 O^T$ because $\lvert S_1\rvert^T = \lvert S_1\rvert$. This implies
\[
\norm{ K }\TR = \norm{ \log \lvert S_1^T\rvert }\TR = \frac{1}{2} \norm{ \log \lvert S_1^T\rvert^2 }\TR = \frac{1}{2} \norm{ \log O \lvert S_1\rvert^2 O^T}\TR =  \norm{\log \lvert S_1\rvert^2}\TR\,.
\]
Using furthermore the blockdiagonalization \eqref{eq:S12}, we find 
\[
\norm{ K }\TR = \norm{U^T\log  \lvert S_1\rvert^2 U}\TR = \frac{1}{2} \norm{ \log  \begin{pmatrix} A_1 & 0 \\ 0 & A_2 \end{pmatrix} }\TR = \frac{1}{2} \norm{ \log A_1 }\TR + \frac{1}{2} \norm{ \log A_2 }\TR \,.
\] 
Eq.~\eqref{eq:A1-bd} and Eq.~\eqref{eq:A2-bd} imply that $\log A_1 \leq 0$ and $\log A_2 \geq 0$. Hence
\[
\norm{ K }\TR = \frac{1}{2} \left( - \tr \log A_1 + \tr \log A_2 \right) = \frac{1}{2} \left( - \log \det A_1 + \log \det A_2 \right) \,.\]
From the definition \eqref{eq:A1A2}, we arrive at
\[
\begin{split}  \norm{ K }\TR &=  \log \det (d+2b) - \log \det d  = \log \det \left(\id + 2 d^{-1/2} b d^{-1/2}\right) \\ &\leq 2 \tr d^{-1/2} b d^{-1/2} = 2 g \langle v , d^{-1} v \rangle = 2 g \sum_{\alpha = 1}^{\ik} \frac{v_\alpha^2}{u_\alpha^2} \leq C g = C\kappa \hat{V}(k) \end{split}
\]
where we used Proposition~\ref{prp:counting}, which implies $v_\alpha^2 \leq C M^{-1} u_\alpha^2$. (Recall also $\ik \leq M/2$.)
\end{proof}
We are now ready to estimate the expectation of $\mathcal{N}^{n}$ in the state $\xi$, defined in \eqref{eq:bos-gs}. We follow a strategy similar to the one developed in the dynamical setting in \cite{BPS14b} for the control of the growth of many-body fluctuations around Hartree--Fock dynamics.
\begin{prp}[Bound on the Number of Fermions]\label{prp:particlenumber}
For all $n \in \Nbb$ and for all $\psi \in \fock$ we have (for a constant $C$ that does not depend on $n$)
\[\sup_{\lambda \in [0,1]} \langle T_\lambda \psi, (\Ncal+1)^n  T_\lambda \psi \rangle \leq e^{C n} \langle \psi, (\Ncal+5)^n \psi \rangle\,.\]
\end{prp}
\begin{proof}
From the CAR \eqref{eq:car} we get
\[
[\Ncal,c^*_{\alpha}(k)] = 2c^*_{\alpha}(k) \quad \text{and} \quad c^*_{\alpha}(k) c^*_{\beta}(l) (\Ncal+4) = \Ncal c^*_{\alpha}(k) c^*_{\beta}(l)\,.
\]
We calculate the derivative w.r.t. $\lambda$ of the expectation value of $(\Ncal+5)^n$:
 \begin{align*}
& \left| \frac{\di}{\di \lambda} \langle T_\lambda \psi, (\Ncal+5)^n T_\lambda\psi \rangle \right|  = \Big| \langle T_\lambda\psi, \sum_{j=0}^{n-1} (\Ncal+5)^j [\Ncal,B](\Ncal+5)^{n-j-1} T_\lambda\psi \rangle \Big| \\
& = \Big| 4\Re \sum_{k\in \north} \sum_{\alpha,\beta\in\Ik} K(k)_{\alpha,\beta} \sum_{j=0}^{n-1} \langle T_\lambda\psi, (\Ncal+5)^j c^*_\alpha(k) c^*_\beta(k) (\Ncal+5)^{n-j-1} T_\lambda\psi\rangle\Big|\;.
\intertext{To distribute the powers of the number operator equally to both arguments of the inner product, we insert $\id = (\Ncal+1)^{\frac{n}{2}-1-j} (\Ncal+1)^{j+1-\frac{n}{2}}$ between $(\Ncal+5)^j$ and $c^*_\alpha(k)$ and then pull $(\Ncal+1)^{j+1-\frac{n}{2}}$  through $c^*_\alpha(k)c^*_\beta(k)$ to the right. Thus}
& \left| \frac{\di}{\di \lambda} \langle T_\lambda \psi, (\Ncal+5)^n T_\lambda\psi \rangle \right| = \Big| 4\Re\sum_{k \in \north} \sum_{\alpha,\beta\in\Ik} K(k)_{\alpha,\beta} \sum_{j=0}^{n-1} \langle \xi_j, c^*_\alpha(k) c^*_\beta(k) \tilde \xi \rangle \Big|
\intertext{where we have introduced $\xi_j := (\Ncal+1)^{\frac{n}{2}-1-j}(\Ncal+5)^j T_\lambda\psi$ and $\tilde{\xi} := (\Ncal+5)^{\frac{n}{2}} T_\lambda\psi$. By Cauchy-Schwarz}
& \left| \frac{\di}{\di \lambda} \langle T_\lambda \psi, (\Ncal+5)^n T_\lambda\psi \rangle \right| \leq 4\sum_{k \in \north} \sum_{j=0}^{n-1} \sum_{\alpha,\beta\in\Ik} \left\lvert  K(k)_{\alpha,\beta}\right\rvert  \norm{c_\beta(k) c_\alpha(k) \xi_j} \norm{\tilde \xi} \\
& \leq 4\sum_{k \in \north} \sum_{j=0}^{n-1} \Big( \sum_{\alpha,\beta\in\Ik} \Big|  K(k)_{\alpha,\beta} \Big|^2 \Big)^{1/2} \Big( \sum_{\alpha,\beta\in\Ik}  \norm{c_\beta(k) c_\alpha(k) \xi_j}^2 \Big)^{1/2} \norm{\tilde \xi}
\intertext{using the first bound from Lemma \ref{lem:bosopbound}} 
& \leq 4 \sum_{k\in \north} \sum_{j=0}^{n-1} \norm{K(k)}\HS \Big( \sum_{\alpha,\beta\in\Ik}  \norm{\Ncal(B_\textnormal{F} \cap B_\beta)^{1/2} c_\alpha(k) \xi_j}^2 \Big)^{1/2} \norm{\tilde \xi}\\
& = 4 \sum_{k\in \north} \sum_{j=0}^{n-1} \norm{K(k)}\HS \Big( \sum_{\alpha\in\Ik}  \langle c_\alpha(k) \xi_j, \sum_{\beta\in\Ik} \Ncal(B_\textnormal{F}\cap B_\beta) c_\alpha(k) \xi_j\rangle \Big)^{1/2} \norm{\tilde \xi}
\intertext{with the trivial estimate $\sum_{\beta\in\Ik} \Ncal(B_\textnormal{F}\cap B_\beta) \leq \Ncal$ then}
& \leq 4 \sum_{k\in \north} \sum_{j=0}^{n-1} \norm{K(k)}\HS \Big( \sum_{\alpha\in\Ik}  \langle c_\alpha(k) \xi_j, \Ncal c_\alpha(k) \xi_j\rangle \Big)^{1/2} \norm{\tilde \xi}\\
& \leq 4 \sum_{k\in \north} \sum_{j=0}^{n-1} \norm{K(k)}\HS \Big( \sum_{\alpha\in\Ik}  \norm{ (\Ncal+2)^{1/2} c_\alpha(k) \xi_j}^2 \Big)^{1/2} \norm{\tilde \xi} \\
& = 4 \sum_{k\in \north} \sum_{j=0}^{n-1} \norm{K(k)}\HS \Big( \sum_{\alpha\in\Ik}  \norm{ c_\alpha(k) \Ncal^{1/2} \xi_j}^2 \Big)^{1/2} \norm{\tilde \xi}
\intertext{and, estimating $c_\alpha(k)$ by the first bound from Lemma \ref{lem:bosopbound},}
& \leq 4 \sum_{k\in \north} \sum_{j=0}^{n-1} \norm{K(k)}\HS \norm{ \Ncal \xi_j}  \norm{\tilde \xi} \leq 4\sum_{k \in \north} n \norm{K(k)}\HS \langle T_\lambda\psi, (\Ncal+5)^n T_\lambda\psi\rangle\,. \tagg{lastline}
\end{align*}
From the differential inequality \eqref{eq:lastline}, using Gr\"onwall's Lemma, we conclude that
\[\langle T_\lambda\psi, (\Ncal+5)^n T_\lambda\psi \rangle \leq \exp\Big( 4 n \lambda \sum_{k \in \north} \norm{ K (k) }\HS \Big) \langle \psi, (\Ncal+5)^n \psi \rangle \leq e^{C n \lambda} \langle \psi, (\Ncal+5)^n \psi \rangle \]
where in the last inequality we used \eqref{eq:kernelbound} and the assumptions on $V$.
\end{proof}

\section{Evaluating the Energy of the Trial State}
In this section we calculate the expectation value $\langle \xi, \mathcal{H}_{\textnormal{corr}} \xi \rangle$, for the trial state $\xi$ defined in \eqref{eq:bos-gs} and $\Hcal_\textnormal{corr}$ defined in \eqref{eq:Hcorr}. We start with some simple estimates for the non-bosonizable terms.  Afterwards we linearize the kinetic energy and calculate its contribution to the expectation value, before we eventually turn to the main part of the interaction.

\subsection{Getting Rid of Non-Bosonizable Terms}\label{sec:nonboson}
In the next lemma, we show that the contribution of the terms in \eqref{eq:cE1} to the expectation $\langle \xi, \mathcal{H}_{\textnormal{corr}} \xi \rangle$ is negligible for $N \to \infty$. 
\begin{lem}[Non-Bosonizable Interaction Terms]
Let $\mathcal{E}_1 (x,y)$ be defined as in \eqref{eq:cE1}. Let $\xi$ be the trial state defined as in \eqref{eq:bos-gs}. Then we have
\[ \Big| \Big\langle \xi, \frac{1}{2N}\int_{\Tbb^3\times \Tbb^3} \di x\di y\, V(x-y) \,\mathcal{E}_1(x,y)\, \xi \Big\rangle \Big| \leq C N^{-1} \,. \]
\end{lem}
\begin{proof}
We are going to show that for all $\psi \in \fock$ we have
 \begin{equation}\label{eq:la}\Big| \Big\langle \psi, \frac{1}{2N}\int_{\Tbb^3\times \Tbb^3} \di x\di y\, V(x-y) \,\mathcal{E}_1(x,y)\, \psi \Big\rangle\Big| \leq \frac{2}{N} \sum_{k \in \Zbb^3} \lvert \hat{V}(k)\rvert\, \langle \psi, (\Ncal+1)^2\psi \rangle\,.\end{equation}
 The final claim then follows using Proposition~\ref{prp:particlenumber}. To prove \eqref{eq:la}, let us rewrite the first term on the r.\,h.\,s.\ of \eqref{eq:cE1} by using the CAR and $\langle u_x,u_y\rangle = u(x,y)$, yielding
 \begin{equation}\label{eq:E1proof}
\begin{split} & \frac{1}{2N} \int_{\Tbb^3\times\Tbb^3} \di x\di y\, V(x-y) a^*(u_x) a^*(u_y) a(u_y) a(u_x)\\
& = \frac{1}{2N} \int_{\Tbb^3\times\Tbb^3} \di x\di y\, V(x-y) \Big( a^*(u_x) a(u_x) a^*(u_y) a(u_y) - a^*(u_x) \langle u_x,u_y\rangle a(u_y) \Big)\\
& = \frac{1}{2N} \sum_{k\in \Zbb^3} \hat{V}(k) \Big( \di\Gamma(u e^{ikx}u)\di\Gamma(u e^{-ikx} u) - \di\Gamma(u e^{ikx} u e^{-ikx} u) \Big)\;.\end{split}\end{equation}
Recall the two bounds $\norm{\di\Gamma(A) \psi} \leq \norm{A}\OP \norm{\Ncal \psi}$ and $\lvert \langle \psi, \di\Gamma(A) \psi\rangle\rvert \leq \norm{A}\OP \langle \psi, \Ncal \psi\rangle$ for any bounded one-particle operator $A$ and any $\psi \in \fock$. Thus, using that $\| u \|_{\text{op}} \leq 1$,
\[\Big| \Big\langle \xi, \frac{1}{2N} \sum_{k\in\Zbb^3} \hat{V}(k) \di\Gamma(u e^{ikx}u)\di\Gamma(u e^{-ikx} u) \xi \Big\rangle\Big| \leq \frac{1}{2N} \sum_{k\in\Zbb^3} \lvert \hat{V}(k)\rvert \norm{\Ncal \xi}^2\,.\]
The second summand in \eqref{eq:E1proof} can be estimated in the same way. The same holds true for the other two terms in \eqref{eq:cE1}.
\end{proof}

Let us now consider the error term $\mathcal{E}_{2}$, defined in \eqref{eq:cE2}. We prove that this term vanishes in our trial state $\xi$. 
\begin{lem}[Interaction Terms of Wrong Parity] Let $\mathcal{E}_{2}(x,y)$ be defined as in \eqref{eq:cE2}. Let $\xi$ be the trial state defined in \eqref{eq:bos-gs}. Then we have
\[
\Big\langle \xi, \frac{1}{2N}\int_{\Tbb^3\times \Tbb^3} \di x\di y\, V(x-y) \big( \mathcal{E}_2(x,y) + \hc\big) \xi \Big\rangle =0\,.
\]
\end{lem}
\begin{proof}
Since terms in $\mathcal{E}_2(x,y)$ create exactly two fermions, we have 
\[i^\Ncal \mathcal{E}_2(x,y) = \mathcal{E}_2(x,y) i^{\Ncal+2} = -\mathcal{E}_2(x,y) i^\Ncal\,.\]
Recall that $\xi = T \Omega$, with $T=  \exp (B)$ and $B$ as in \eqref{eq:trial}. We have $[i^\Ncal,B] =0$, since $B$ creates or annihilates particles four at a time. This implies $T i^\Ncal = i^\Ncal T$. Using $i^\Ncal \Omega = \Omega$, we get
\begin{align*}
\langle T\Omega , \Ecal_2 T \Omega \rangle & = \langle T \Omega, \Ecal_2 T i^{\Ncal} \Omega\rangle = -\langle T \Omega, i^\Ncal \Ecal_2 T\Omega \rangle = - \langle (-i)^\Ncal T\Omega, \Ecal_2 T\Omega \rangle \\
& = - \langle  T (-i)^\Ncal\Omega, \Ecal_2 T\Omega \rangle = - \langle  T\Omega, \Ecal_2 T\Omega \rangle\;,
\end{align*}
which thus vanishes.
\end{proof}

\subsection{Estimating Direct and Exchange Operators}
\label{sec:direx}

In this section we estimate the contribution of the direct and exchange terms to $\di\Gamma(uhu-\cc{v}\cc{h}v)$. Recall that
\[h = -\frac{\hbar^2 \Delta}{2} + (2\pi)^3 \hat{V} (0) + X \]
where $X$ has the integral kernel $X(x,y) = -N^{-1} V(x-y) \omega_{\textnormal{pw}}(x,y)$.
The contribution of the constant direct term $(2\pi)^3 \hat{V} (0)$ is
\[
(2\pi)^3 \hat{V} (0) \, \di \Gamma (u^2 - \cc{v} v) = (2\pi)^3 \hat{V} (0)  \di\Gamma (\id -2 \omega_\textnormal{pw}) = (2\pi)^3 \hat{V} (0) (\Ncal_\textnormal{p} - \Ncal_\textnormal{h})
\]
and therefore it vanishes on $\xi$ by Lemma \ref{lem:phsymmetry}. The next lemma allows us to control the contribution of the exchange term $X$. 
\begin{lem}[Bound for the Exchange Term] Let $\xi$ be the trial state defined as in \eqref{eq:bos-gs}. Then we have
\[
\lvert \langle \xi , \di \Gamma(uXu-\cc{v}\cc{X}v) \xi \rangle \rvert \leq C N^{-1}\,. 
\]
\end{lem}
\begin{proof}
Notice that
\[\omega_\textnormal{pw}(x,y) = \frac{1}{(2\pi)^3}\sum_{h \in B_\textnormal{F}} e^{ih\cdot(x-y)} =: f (x-y)\,.\]
Thus $X$ is translation invariant, and hence
\[
\norm{X}_\text{op}  = N^{-1} \norm{ \widehat{Vf}}_{L^\infty} \leq N^{-1} \norm{ \hat{f} }_{L^\infty} \sum_{k \in \Zbb^3}  \lvert \hat{V}(k)\rvert\, \leq C N^{-1}.
\]
Using that $\norm{u}\OP = 1= \norm{v}\OP$, we get, by Proposition~\ref{prp:particlenumber}:
\[
\lvert \langle \xi, \di\Gamma(uXu-\cc{v}\cc{X}v) \xi \rangle \rvert \leq \norm{uXu-\cc{v}\cc{X}v}\OP \langle \xi ,\Ncal \xi \rangle \leq C N^{-1}\,. \qedhere
\]
\end{proof}

\subsection{Expectation Value of the Kinetic Energy}
\label{sec:lin}
In this section we evaluate the contribution of the Laplacian to the expectation value of the correlation Hamiltonian in the trial state $\xi$ defined as in \eqref{eq:bos-gs}. We start by linearizing in Fourier space,
\begin{align*}
-\frac{\hbar^2}{2}  \langle \xi, \di\Gamma\left(u \Delta u - \cc{v} \Delta v\right) \xi\rangle 
  & = \frac{\hbar^2}{2}\langle \xi, \Big[ \sum_{p \in B_\textnormal{F}^c} p^2 a^*_p a_p - \sum_{h \in B_\textnormal{F}} h^2  a^*_h a_h \Big] \xi\rangle \\
 & = \frac{\hbar^2}{2}\langle \xi, \sum_{\alpha=1}^M \Big[ \sum_{p \in B_\textnormal{F}^c \cap B_\alpha} \left((p-\omega_\alpha)^2 + 2p\cdot \omega_\alpha - \omega_\alpha^2\right) a^*_p a_p\\
 & \hspace{2.5cm} - \sum_{h \in B_\textnormal{F} \cap B_\alpha} \left((h-\omega_\alpha)^2+ 2h \cdot \omega_\alpha -\omega_\alpha^2\right)  a^*_h a_h \Big] \xi\rangle\,.
\end{align*}
Notice that from the first to the second line, momenta $p$ and $h$ that lie in the corridors or are more than a distance $R$ away from the Fermi surface have disappeared from the sums; this is justified since such modes are never occupied in the trial state, i.\,e., $a_p \xi = 0$ and $a_h\xi =0$.
Furthermore, thanks to Lemma \ref{lem:phsymmetry}, we have
\[
\langle \xi , \sum_{\alpha=1}^M \left[ \sum_{p \in B_\textnormal{F}^c \cap B_\alpha} \omega_\alpha^2 a^*_p a_p- \sum_{h \in B_\textnormal{F} \cap B_\alpha} \omega_\alpha^2 a^*_h a_h\right] \xi \rangle = k_\textnormal{F}^2 \, \langle \xi, \left[ \Ncal_\textnormal{p}- \Ncal_\textnormal{h}\right] \xi \rangle = 0
\]
where we used that $\lvert \omega_\alpha \rvert = k_\textnormal{F}$ for all $\alpha$. To estimate $(p-\omega_\alpha)^2$ and $(h-\omega_\alpha)^2$, we recall that the diameter of the patches is bounded by $C \sqrt{N^{2/3}/M}$ (since the diameter of the patch on the Fermi surface is bounded by $\sqrt{N^{2/3}/M}$ which is large compared to its thickness of order $R$). Therefore 
\[
\begin{split}& -\frac{\hbar^2}{2} \langle \xi,  \di\Gamma\left(u \Delta u - \cc{v} \Delta v\right) \xi \rangle = \langle \xi, \Hbb_\textnormal{kin} \xi \rangle +\mathfrak{E}_\textnormal{lin}
\end{split}
\]
where we introduced 
\[\Hbb_\textnormal{kin} := {\hbar^2} \sum_{\alpha=1}^M \Big[ \sum_{p \in B_\textnormal{F}^c \cap B_\alpha}p\cdot \omega_\alpha\, a^*_p a_p - \sum_{h \in B_\textnormal{F} \cap B_\alpha} h\cdot \omega_\alpha\, a^*_h a_h \Big] \]
and where the error $\mathfrak{E}_\textnormal{lin}$ is bounded by 
\[\begin{split}
\left \lvert \Efrak_\textnormal{lin} \right\rvert & = \Big| \frac{\hbar^2}{2} \langle \xi , \sum_{\alpha=1}^M \Big[ \sum_{p \in B_\textnormal{F}^c \cap B_\alpha} (p - \omega_\alpha)^2 a^*_p a_p- \sum_{h \in B_\textnormal{F} \cap B_\alpha} (h - \omega_\alpha)^2 a^*_h a_h \Big] \xi  \rangle \Big|\\
& \leq C \frac{\hbar^2}{2} \frac{N^{2/3}}{M} \langle \xi, \Ncal \xi \rangle \leq \frac{C}{M}
\end{split}\]
where in the last step we used Proposition~\ref{prp:particlenumber} to bound the expectation value of the number operator and $\hbar = N^{-1/3}$. 

\medskip

To compute the expectation of the linearized kinetic energy operator $\Hbb_\textnormal{kin}$, we will make use of the following lemma. 
\begin{lem}[Kinetic Energy of Particle-Hole Pairs]\label{lem:kingcomm}
For all $k\in \north$ and $\alpha \in \Ik$ we have
\[ [\Hbb_\textnormal{kin}, c^*_\alpha(k)] = {\hbar^2} \lvert k\cdot \omega_\alpha \rvert c^*_\alpha(k)\,.\]
\end{lem}
\begin{proof}We first treat the case $\alpha \in \Ikp$, for which $k\cdot\omega_\alpha >0$. Using the CAR we calculate
\begin{align*}
[ \Hbb_\textnormal{kin}, c^*_{\alpha}(k)] & = [ \Hbb_\textnormal{kin}, b^*_{\alpha,k} ] = [\Hbb_\textnormal{kin}, \frac{1}{n_{\alpha,k}} \sum_{\substack{p \in B_\textnormal{F}^c \cap B_\alpha\\ h \in B_\textnormal{F} \cap B_\alpha}} \delta_{p-h,k} a^*_p a^*_h]\\
& = {\hbar^2} \sum_{\beta=1}^M \frac{1}{n_{\alpha,k}} \sum_{\substack{p \in B_\textnormal{F}^c \cap B_\alpha\\h \in B_\textnormal{F} \cap B_\alpha}} \delta_{p-h,k} \sum_{\tilde p \in B_\textnormal{F}^c \cap B_\beta} \tilde p \cdot \omega_\beta [a^*_{\tilde p} a_{\tilde p}, a^*_p a^*_h] \\
& \quad - {\hbar^2} \sum_{\beta=1}^M \frac{1}{n_{\alpha,k}} \sum_{\substack{p \in B_\textnormal{F}^c \cap B_\alpha\\h \in B_\textnormal{F} \cap B_\alpha}} \delta_{p-h,k} \sum_{\tilde h \in B_\textnormal{F} \cap B_\beta} \tilde h \cdot \omega_\beta [a^*_{\tilde h} a_{\tilde h}, a^*_p a^*_h] \\
& = {\hbar^2} \sum_{\beta=1}^M \frac{1}{n_{\alpha,k}} \sum_{\substack{p \in B_\textnormal{F}^c \cap B_\alpha\\ h \in B_\textnormal{F}\cap B_\alpha}} \delta_{p-h,k} \Big( \sum_{\tilde p \in B_\textnormal{F}^c \cap B_\beta} \tilde p \cdot \omega_\beta \delta_{p, \tilde p}  - \sum_{\tilde h \in B_\textnormal{F} \cap B_\beta} \tilde h\cdot \omega_\beta \delta_{h,\tilde h}\Big) a^*_p a^*_h\,;
\intertext{notice that the Kronecker deltas $\delta_{p,\tilde p}$ and $\delta_{h,\tilde h}$ imply $\beta=\alpha$, so we find}
& = {\hbar^2} \frac{1}{n_{\alpha,k}} \sum_{\substack{p \in B_\textnormal{F}^c \cap B\alpha\\h \in B_\textnormal{F} \cap B_\alpha}} \delta_{p-h,k} (p-h)\cdot\omega_\alpha a^*_p a^*_h  = {\hbar^2} \lvert k \cdot \omega_\alpha\rvert c^*_\alpha(k)\,.
\end{align*}
The absolute value was trivially introduced since the scalar product is anyway non-negative. For $k\cdot\omega_\alpha < 0$, recall that  $c^*_\alpha(k) = b^*_{\alpha,-k}$; the calculation then proceeds the same way, but in the second last line we use $(p-h)\cdot \omega_\alpha = (-k)\cdot\omega_\alpha = \lvert k \cdot \omega_\alpha\rvert$.
\end{proof}

We are now ready to calculate the kinetic energy of our trial state.

\begin{prp}[Kinetic Energy]\label{prp:kin}
Let $\xi$ be defined as in \eqref{eq:bos-gs}. Then
\[\langle \xi, \Hbb_\textnormal{kin} \xi \rangle = {\hbar\kappa}\sum_{k\in\north} \lvert k\rvert  \tr \D(k)\sinh^2(K(k))  + \mathfrak{E}_\textnormal{kin}\,,\]
where $\D (k)$ is defined in \eqref{eq:blocks2} and the error term is such that $\lvert \Efrak_\textnormal{kin}\rvert \leq C \hbar / \nfrak^2$ with $\nfrak = N^{1/3-\delta/2} M^{-1/2}$ as in \eqref{eq:nNM}. 
\end{prp}
\begin{proof}
We write $T_\lambda = \exp (\lambda B)$, with $B$ as in \eqref{eq:trial}, and $\xi = T \Omega$. Hence
\begin{align*}
& \langle \xi , \Hbb_\textnormal{kin} \xi \rangle\\
& = \int_0^1 \di \lambda \langle \Omega, T^*_\lambda [\Hbb_\textnormal{kin}, B] T_\lambda \Omega \rangle \\
& = \int_0^1 \di \lambda \langle \Omega, T^*_\lambda \Big[ \Hbb_\textnormal{kin}, \sum_{k\in\north} \frac{1}{2} \sum_{\alpha,\beta \in\Ik} K(k)_{\alpha,\beta} c^*_\alpha(k) c^*_\beta(k) - \hc \Big] T_\lambda \Omega\rangle \\
& = \Re\int_0^1 \di \lambda \sum_{k\in \north} \sum_{\alpha,\beta \in\Ik} K(k)_{\alpha,\beta} \langle \Omega, T^*_\lambda \left( [\Hbb_\textnormal{kin}, c^*_\alpha(k)]c^*_\beta(k) + c^*_\alpha(k) [\Hbb_\textnormal{kin},c^*_\beta(k)]\right) T_\lambda \Omega\rangle\,.
\end{align*}
From Lemma \ref{lem:kingcomm}
\begin{align*}
\langle \xi , \Hbb_\textnormal{kin} \xi \rangle & = \Re\int_0^1\di\lambda \sum_{k\in\north} {\hbar^2} \sum_{\alpha,\beta \in\Ik} K(k)_{\alpha,\beta} \left( \lvert k\cdot \omega_\alpha \rvert + \lvert k\cdot \omega_\beta \rvert \right) \langle \Omega, T^*_\lambda c^*_\alpha(k) c^*_\beta(k) T_\lambda \Omega \rangle\,.
\intertext{Recall that $\lvert k \cdot \omega_\alpha\rvert = \lvert k\rvert \kappa \hbar^{-1} u_\alpha(k)^2$ with  $u_\alpha (k)$ defined in \eqref{eq:def-u}. Using Proposition \ref{lem:bog1} then}
\langle \xi , \Hbb_\textnormal{kin} \xi \rangle & = \Re\int_0^1\di\lambda \sum_{k\in\north} \lvert k\rvert{\hbar \kappa} \sum_{\alpha,\beta \in\Ik} K(k)_{\alpha,\beta} \left( u_\alpha(k)^2 + u_\beta(k)^2 \right) \\
& \quad \times \langle \Omega,\left(\sum_{\delta\in\Ik} \cosh(\lambda K(k))_{\alpha,\delta} c^*_\delta(k) + \sum_{\delta\in\Ik} \sinh(\lambda K(k))_{\alpha,\delta} c_\delta(k) + \mathfrak{E}^*_\alpha(\lambda,k) \right)\\
& \hspace{1cm}\times \left(\sum_{\gamma\in\Ik} \cosh(\lambda K(k))_{\beta,\gamma} c^*_\gamma(k) + \sum_{\gamma\in\Ik} \sinh(\lambda K(k))_{\beta,\gamma} c_\gamma(k) + \mathfrak{E}^*_\beta(\lambda,k) \right)\Omega\rangle\,.
\intertext{Finally, using $c_\delta(k)\Omega = 0$ and $\langle \Omega, c_\delta(k) c^*_\gamma(k)\Omega\rangle = \delta_{\delta,\gamma}$, we get}
 \langle \xi , \Hbb_\textnormal{kin} \xi \rangle & = \Re \int_0^1 \di \lambda \sum_{k \in\north} \lvert k\rvert\frac{\hbar\kappa}{2} \sum_{\alpha,\beta \in\Ik} K(k)_{\alpha,\beta} \left( u_\alpha(k)^2 + u_\beta(k)^2 \right)\\
& \hspace{2.5cm} \times  \sum_{\delta\in\Ik} \sinh(\lambda K(k))_{\alpha,\delta} \sum_{\gamma\in\Ik} \cosh(\lambda K(k))_{\beta,\gamma} \delta_{\delta,\gamma} + \Efrak_\textnormal{kin} \\
& =  \sum_{k\in\north} \lvert k\rvert {\hbar\kappa} \sum_{\alpha \in\Ik}  u_\alpha(k)^2 \int_0^1 \di \lambda \Big( \sinh\big(2\lambda K(k)\big) K(k) \Big)_{\alpha,\alpha} + \mathfrak{E}_\textnormal{kin} \tagg{step}
\end{align*}
where we defined 
\begin{align*}
\mathfrak{E}_\textnormal{kin} & := \Re\int_0^1 \di \lambda \sum_{k \in\north} \lvert k\rvert\frac{\hbar\kappa}{2} \sum_{\alpha,\beta \in\Ik} K(k)_{\alpha,\beta} \left( u_\alpha(k)^2 + u_\beta(k)^2 \right)\\
& \quad \times \Big( \langle \Omega, \mathfrak{E}^*_\alpha(\lambda,k) \mathfrak{E}^*_\beta(\lambda,k) \Omega\rangle + \sum_{\delta\in\Ik} \sinh(\lambda K(k))_{\alpha,\delta} \langle\Omega,c_\delta(k) \mathfrak{E}^*_\beta(\lambda,k)\Omega\rangle \\
& \hspace{1.2cm}+ \sum_{\gamma\in \Ik} \cosh(\lambda K(k))_{\beta,\gamma} \langle \Omega, \mathfrak{E}_\alpha^*(\lambda,k) c^*_\gamma(k) \Omega\rangle \Big)
\\ & =: \mathfrak{E}^{(1)}_\textnormal{kin} +  \mathfrak{E}^{(2)}_\textnormal{kin}+ \mathfrak{E}^{(3)}_\textnormal{kin}
\end{align*}
We compute the integral in \eqref{eq:step}. We get
\[\langle \xi , \Hbb_\textnormal{kin} \xi \rangle  = {\hbar\kappa}\sum_{k\in \north} \lvert k\rvert \tr \D(k)\sinh^2(K(k))+ \mathfrak{E}_\textnormal{kin}\,.\]
Using that $u_\alpha(k)^2 = \lvert \hat{k}\cdot \hat{\omega}_\alpha \rvert \leq 1$, we bound the first error term by 
\begin{align*}
| \mathfrak{E}^{(1)}_\textnormal{kin}| 
& \leq \Big| \int_0^1 \di \lambda \sum_{k \in\north} \lvert k\rvert {\hbar\kappa} \sum_{\alpha,\beta \in\Ik} K(k)_{\alpha,\beta} \left( u_\alpha(k)^2 + u_\beta(k)^2 \right) \langle \Omega, \mathfrak{E}^*_\alpha(\lambda,k) \mathfrak{E}^*_\beta(\lambda,k) \Omega\rangle \Big|\\
& \leq  2\int_0^1\di\lambda \sum_{k \in \north} \lvert k\rvert {\hbar\kappa}\sum_{\alpha,\beta \in \Ik} \lvert K(k)_{\alpha,\beta}\rvert \norm{\Efrak_\alpha(\lambda,k)\Omega} \norm{\Efrak^*_\beta(\lambda,k)\Omega} \\
& \leq  2\int_0^1\di\lambda \sum_{k \in \north} \lvert k\rvert  {\hbar\kappa} \Big[ \sum_{\alpha,\beta \in \Ik} \lvert K(k)_{\alpha,\beta}\rvert^2 \Big]^{1/2} \Big[ \sum_{\alpha\in\Ik} \norm{\Efrak_\alpha(\lambda,k)\Omega}^2 \sum_{\beta\in \Ik} \norm{\Efrak^*_\beta(\lambda,k)\Omega}^2 \Big]^{1/2}\,;
\intertext{and finally using \eqref{eq:bogerrorbound}}
| \mathfrak{E}^{(1)}_\textnormal{kin}|& \leq  \frac{C\hbar }{\nfrak^4} \sum_{k \in \north} \lvert k\rvert \norm{K(k)}\HS e^{2\norm{K(k)}\HS}  \sup_{\lambda \in [0,1]} \langle T_\lambda \Omega, (\Ncal+2)^{3} T_\lambda \Omega \rangle \Big( \sum_{l\in\north} \norm{K(l)}\HS \Big)^{2} . 
\end{align*}
From Proposition~\ref{prp:particlenumber} and Lemma \ref{lem:inverses}, we conclude that $|\mathfrak{E}^{(1)}_\textnormal{kin}| \leq C \hbar / \nfrak^4$. The third error term $\mathfrak{E}^{(3)}_\textnormal{kin}$ can be controlled similarly, using Lemma \ref{lem:bosopbound}:
\begin{align*}
&|\mathfrak{E}^{(3)}_\textnormal{kin}| \\ 
& \leq C\hbar\int_0^1 \di \lambda  \norm{(\Ncal+1)^{1/2} \Omega}  \sum_{k\in\north} \lvert k\rvert  \sum_{\alpha,\beta \in \Ik} \lvert K(k)_{\alpha,\beta} \rvert \norm{\Efrak_\alpha(\lambda,k)\Omega} \Big[ \sum_{\gamma\in\Ik} \lvert \cosh(\lambda K(k))_{\beta,\gamma} \rvert^2 \Big]^{1/2} \\
& \leq \frac{C\hbar}{\nfrak^2} \norm{(\Ncal+1)^{1/2} \Omega} \sup_{\lambda \in [0,1]} \norm{(\Ncal+2)^{3/2}T_\lambda \Omega}   \sum_{k\in \north} \lvert k\rvert \norm{K(k)}\HS \, e^{\norm{K(k)}\HS} \\ &\hspace{6cm} \times \int_0^1 \di \lambda \, \norm{\cosh(\lambda K(k))}\HS  \sum_{l\in \north} \norm{K(l)}\HS \\
& \leq \frac{C\hbar}{\nfrak^2} \norm{(\Ncal+1)^{1/2} \Omega} \sup_{\lambda \in [0,1]} \norm{(\Ncal+2)^{3/2}T_\lambda \Omega}\sum_{k\in \north} \lvert k\rvert \norm{K(k)}\HS \, e^{2\norm{K(k)}\HS}  \sum_{l\in \north} \norm{K(l)}\HS\\
& \leq \frac{\hbar}{\nfrak^2} C\,.
\end{align*}
The second error term $\Efrak_\textnormal{kin}^{(2)}$ can be controlled in the same way.
\end{proof}

\subsection{Expectation Value of the Interaction Energy}

We now evaluate the main contribution \eqref{eq:QN0-def} to the interaction energy. This is the content of the next proposition.

\begin{prp}[Interaction Energy]\label{prp:int}
Let $\xi$ be the trial state defined as in \eqref{eq:bos-gs}, and let $Q_{N}^{\textnormal{B}}$ be given by \eqref{eq:QN0-def}. Then
\[\begin{split}\langle \xi , Q_N^{\textnormal{B}} \xi \rangle & = {\hbar\kappa} \sum_{k\in \north} \lvert k\rvert \tr \left( \W(k) \sinh^2 (K(k)) + \Wt (k) \sinh(K(k)) \cosh(K(k)) \right)\\
& \quad + \mathfrak{E}_\textnormal{int} + \Ocal(\hbar N^{-\delta/2})\end{split}\]
where $\W (k)$ and $\Wt (k)$ are defined in \eqref{eq:blocks2}. The error term is bounded by $\lvert \mathfrak{E}_\textnormal{int}\rvert \leq C\hbar / \nfrak^2$, with $\nfrak = N^{1/3-\delta/2} M^{-1/2}$ as in \eqref{eq:nNM}. 
\end{prp}
\begin{proof}
We start by decomposing the $b_{k}$-operators in the interaction Hamiltonian \eqref{eq:QN0-1} by their patch decomposition \eqref{eq:patchdecomp},
\begin{equation}\label{eq:bdec}
\tilde{b}_k = \sum_{\alpha \in \Ikp} \tilde{b}_{\alpha,k} + \rfrak_k\,.\end{equation}
We recall that the error terms $\rfrak_k$ collect modes in the corridors and close to the equator:
\begin{equation}\label{eq:splitr}\rfrak_k = \tilde{\rfrak}_k + \sum_{\alpha \not\in \Ik} \tilde b_{\alpha,k}\,,\end{equation}
where $\tilde{\rfrak}_k$ is a linear combination of products $a_h a_p$ such that at least one of the two momenta is in the corridors $B_\textnormal{corri}$ (see Figure \ref{fig:nontouching}), and the second term collects the contributions coming from the patches close to the equator. We are going to show that $\rfrak_k$ gives a negligible contribution to $\langle \xi, Q_{N}^{\textnormal{B}} \xi \rangle$.

\paragraph{Contribution of Corridors.}
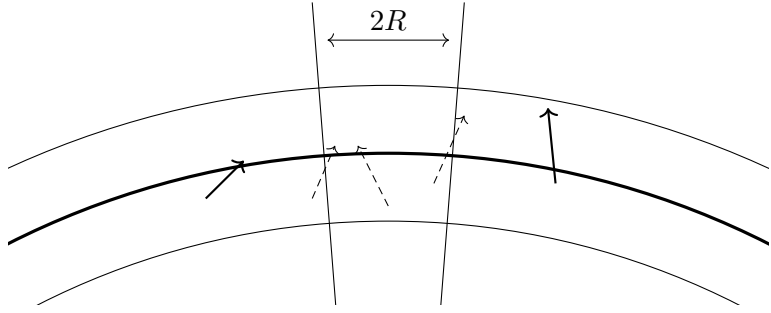
\begin{figure}
\centering
\begin{tikzpicture}
\clip (-5,3) rectangle (5,7);

\draw(0,-6) circle (11.9cm);
\draw[very thick] (0,-6) circle (11cm);
\draw (0,-6) circle (10.1cm);

\draw [<->] (-.8,6.5)--(.8,6.5);
\node [above] at (0,6.5) {$2R$};

\draw[thick,->] (-2.4,4.4)--(-1.9,4.9);
\draw[thick,->] (2.2,4.6)--(2.1,5.6);

\draw[densely dashed,->] (-1,4.4)--(-.7,5.1);
\draw[densely dashed,->] (.6,4.6)--(1,5.5);
\draw[densely dashed,->] (0,4.3)--(-.4,5.1);

 \draw (0,-6)--(-1,7);
\draw (0,-6)--(1,7);
 \end{tikzpicture}
 \caption{Fermi surface in bold; two patches separated by a corridor of width $2R$. Bold arrows represent particle--hole pairs $(p,h)$ that contribute to the expectation value of the interaction Hamiltonian. Dashed arrows represent particle--hole pairs of which mode $p$ or $h$ (or both) is not occupied in the trial state $\xi$. Since $\lvert k\rvert \leq R$, pairs connecting the patches across the corridor do not exist in $Q_N^{\textnormal{B}}$.}\label{fig:nontouching}
\end{figure}
We claim that the error operators $\tilde{\rfrak}_k$ do not contribute to $\langle \xi, Q_{N}^{\textnormal{B}} \xi \rangle$. To see this, recall that $T = e^{B}$, with $B$ not containing any mode $q\in B_\textnormal{corri}$, see \eqref{eq:trial}. Since at least one of the two momenta $p$ and $h$ appearing in $\tilde{\rfrak}_k$ is in the corridor $B_\textnormal{corri}$, we have $\tilde{\rfrak}_k \xi = 0$. Plugging the decomposition \eqref{eq:bdec} into $Q^{\textnormal{B}}_{N}$ from \eqref{eq:QN0-1} and taking the expectation value on $\xi$, we realize that all terms containing at least one error operator $\tilde{\rfrak}^\natural_k$ are zero, due to the fact that there is at least one error operator $\tilde{\rfrak}_k$ directly acting on $T \Omega$. 

\paragraph{Contribution of Patches near the Equator.} We claim that the contribution to $\langle\xi, Q^{\textnormal{B}}_{N} \xi\rangle$ coming from  patches $\beta \not\in \Ik$ is subleading as $N\to \infty$. These are the patches $\beta$ in the collar where $\lvert \hat{k}\cdot \hat{\omega}_\beta \rvert < N^{-\delta}$. The width of this collar is bounded above by $C k_\textnormal{F} N^{-\delta}$, and its length---approximately equal to the circumference of the equator---is bounded above by $C k_\textnormal{F}$; we conclude that the surface area of the collar is of order $k_F^2 N^{-\delta}$.

Recall that $n^2_{\beta,k}$ is the number of particle--hole pairs with relative momentum $k$ in patch $\beta$; thus, adding in the corridors for an upper bound,
 $\sum_{\beta \not\in \Ik} n^2_{\beta,k}$ is bounded by the number of particle--hole pairs with relative momentum $k$ in the collar. This number is bounded above by the number of hole momenta $h \in B_\textnormal{F}$ that are at most a distance $R$ from the collar (since $\lvert k\rvert \leq R$). The number of such points of the lattice $\Zbb^3$ can be counted by Gauss' classical argument: assign to each lattice point $k$ the cube $[k_1,k_1+1]\times[k_2,k_2+1]\times[k_3,k_3+1]$. Then the number of cubes belonging to lattice points near the collar is bounded by the Lebesgue measure of the collar ``fattened'' to a thickness $R$; i.\,e.,
\begin{equation}\label{eq:number-in-collar}
 \sum_{\beta \not\in \Ik} n^2_{\beta,k} \leq C k_\textnormal{F}^2 N^{-\delta} \times R = \Ocal(N^{2/3-\delta})\,.
\end{equation}

We are now ready to estimate the contribution to $\langle \xi, Q_{N}^{\textnormal{B}} \xi \rangle$ coming from the modes close to the equator. Consider, e.\,g., the term $\frac{1}{2N} \sum_{k \in \north} \hat{V}(k) 2\tilde b^*_k \tilde b_k$ (all the other terms can be dealt with similarly). We get three contributions from \eqref{eq:splitr}, namely
\begin{equation}\label{eq:rk-errors}\begin{split} & \frac{1}{N} \sum_{k \in \north}\hat{V}(k) \sum_{\beta \not \in \Ik} \tilde{b}^*_{\beta,k} \sum_{\alpha \in \Ikp} \tilde{b}_{\alpha,k}\,,\\
& \frac{1}{N} \sum_{k \in \north}\hat{V}(k) \sum_{\beta \in \Ikp} \tilde{b}^*_{\beta,k} \sum_{\alpha \not \in \Ik} \tilde{b}_{\alpha,k}\,,\\
& \frac{1}{N} \sum_{k \in \north}\hat{V}(k) \sum_{\beta \not \in \Ik} \tilde{b}^*_{\beta,k} \sum_{\alpha \not \in \Ik} \tilde{b}_{\alpha,k}\,.\end{split}\end{equation}
We give the detailed estimate for the first term in the list (the other two terms can be controlled similarly)
\begin{align*}
& \frac{1}{N} \sum_{k \in \north}\hat{V}(k) \Big\langle \xi, \sum_{\beta \not \in \Ik} \tilde{b}^*_{\beta,k} \sum_{\alpha \in \Ikp} \tilde{b}_{\alpha,k} \xi \Big\rangle  \\ & = 
\frac{1}{N} \sum_{k \in \north}\hat{V}(k) \Big\langle \xi, \sum_{\beta \not \in \Ik} n_{\beta,k} b^*_{\beta,k} \sum_{\alpha \in \Ikp} n_{\alpha,k} {b}_{\alpha,k} \xi \Big\rangle \\
& \leq \frac{1}{N} \sum_{k\in\north}\hat{V}(k) \Big( \sum_{\beta \not \in \Ik} n_{\beta,k}^2 \Big)^{1/2} \Big( \sum_{\alpha \in \Ikp} n_{\alpha,k}^2\Big)^{1/2} \norm{\Ncal^{1/2} \xi}^2 \\
& \leq \frac{1}{N} \sum_{k\in\north}\hat{V}(k) \Big( C N^{2/3-\delta} \Big)^{1/2} \Big( M \frac{C k_\textnormal{F}^2}{M}\Big)^{1/2} \langle \xi,\Ncal \xi\rangle 
= \Ocal(\hbar N^{-\delta/2})\,,
\end{align*}
where we used \eqref{eq:number-in-collar} to control the sum over $\beta \not\in \Ik$, $n_{\alpha,k}^2 \leq C k_\textnormal{F}^2/M$ due to Proposition \ref{prp:counting} for the sum over $\alpha \in \Ikp$, and the bound on the number of fermions $\Ncal$ from Proposition~\ref{prp:particlenumber}.

\paragraph{Approximate Bogoliubov Diagonalization of the Effective Interaction.} By the discussion of the previous paragraph
\[
\begin{split}& \langle \xi, Q_N^{\textnormal{B}} \xi \rangle \\ &  = \frac{1}{N} \langle \xi, \sum_{k \in \north} \hat{V}(k) \Big( \sum_{\alpha,\beta \in\Ikp}\tilde b^*_{\alpha,k} \tilde{b}_{\beta,k} + \sum_{\alpha,\beta \in \Ikm}\tilde b^*_{\alpha,-k} \tilde b_{\beta,-k} + \Big[ \sum_{\alpha \in \Ikp,\, \beta \in \Ikm} \tilde b^*_{\alpha,k} \tilde b^*_{\beta,-k} + \hc\Big]\Big) \xi \rangle \\
& \quad + \Ocal(\hbar N^{-\delta/2})\,.\end{split}
\]
Introducing the normalization factors $n_{\alpha,k} = k_\textnormal{F} \sqrt{\lvert k\rvert} v_\alpha(k)$ and combining the $b^*_{\alpha,k}$ and $b^*_{\alpha,-k}$ operators to $c^*_\alpha(k)$ operators as in \eqref{eq:rigorouscoperators}, we get
\begin{equation}\label{eq:QN0-deco}\langle \xi, Q_N^{\textnormal{B}} \xi \rangle = \langle \xi, \Hbb_\textnormal{int} \xi\rangle + \Ocal(\hbar N^{-\delta/2})\,, \qquad \Hbb_\textnormal{int} := \Hbb_\textnormal{int}^{(1)} + \Hbb_\textnormal{int}^{(2)} + \Hbb_\textnormal{int}^{(3)}\,,\end{equation}
where, recalling that $g(k) = \kappa\hat{V}(k)$,
\begin{align*}\Hbb_\textnormal{int}^{(1)} & := {\hbar\kappa} \sum_{k\in \north} \lvert k\rvert g(k)\sum_{\alpha,\beta \in \Ikp} v_\alpha(k) v_\beta(k) c^*_\alpha(k) c_\beta(k)\,,\\
\Hbb_\textnormal{int}^{(2)} & := \hbar\kappa \sum_{k\in \north} \lvert k\rvert g(k)\sum_{\alpha,\beta \in \Ikm} v_\alpha(k) v_\beta(k) c^*_\alpha(k) c_\beta(k)\,,\\
\Hbb_\textnormal{int}^{(3)} & := \hbar\kappa\sum_{k\in\north} \lvert k\rvert g(k) \sum_{\alpha\in\Ikp} \sum_{\beta \in \Ikm} v_\alpha(k) v_\beta(k) c^*_\alpha(k) c^*_\beta(k) + \hc\end{align*}
We shall evaluate $\langle \xi, \mathbb{H}^{(i)}_{\text{int}} \xi \rangle$, $i=1,2,3$, with $\xi = T\Omega$, using the fact that the $T$ operator behaves as an approximate bosonic Bogoliubov transformation, recall Proposition \ref{lem:bog1}. Using also $\langle \Omega, c_\delta(k) c^*_\gamma(k)\Omega\rangle = \delta_{\delta,\gamma}$, we have
\begin{equation}\label{eq:int1}
   \langle \xi , \Hbb_\textnormal{int}^{(1)} \xi\rangle = \hbar\kappa\sum_{k\in\north} \lvert k\rvert \tr \W^{++}(k) \sinh^2(K(k)) + \mathfrak{E}_\textnormal{int}^{(1)}\,,
\end{equation}
where
\[\W^{++}(k)_{\alpha,\beta} = \left\{ \begin{array}{cl} g(k) v_\alpha(k) v_\beta(k) & \text{for }\alpha,\beta \in \Ikp \\
0 & \text{otherwise} \end{array}\right.\] and the error term is
\begin{align*}
 \mathfrak{E}_\textnormal{int}^{(1)} & = \hbar\kappa \sum_{k\in\north} \lvert k\rvert g(k) \sum_{\alpha,\beta\in \Ikp} v_\alpha(k) v_\beta(k) \Bigg[ \sum_{\gamma\in\Ik} \sinh(K(k))_{\alpha,\gamma} \langle \Omega, c_\gamma(k) \mathfrak{E}_\beta (1,k) \Omega\rangle \\
 & \quad +\sum_{\gamma\in\Ik} \sinh(K(k))_{\beta,\gamma} \langle \Omega, \mathfrak{E}^*_\alpha(1,k) c^*_\gamma(k) \Omega\rangle +\langle \Omega, \mathfrak{E}^*_\alpha (1,k) \mathfrak{E}_\beta (1,k) \Omega\rangle\Bigg].
\end{align*}
For the second part of the interaction we find
\begin{equation}\label{eq:int2}
   \langle \xi, \Hbb_\textnormal{int}^{(2)} \xi \rangle = \hbar\kappa\sum_{k\in\north} \lvert k\rvert \tr \W^{--}(k) \sinh^2(K(k)) + \mathfrak{E}_\textnormal{int}^{(2)}\,,
\end{equation}
where
\[\W^{--}(k)_{\alpha,\beta} = \left\{ \begin{array}{cl} g(k) v_\alpha(k) v_\beta(k) & \text{for }\alpha,\beta \in \Ikm \\
0 & \text{otherwise} \end{array}\right.\] 
and 
\begin{align*}
 \mathfrak{E}_\textnormal{int}^{(2)} & = \hbar\kappa \sum_{k\in\north} \lvert k\rvert g(k) \sum_{\alpha,\beta\in\Ikm} v_\alpha(k) v_\beta(k) \Bigg[ \sum_{\gamma\in\Ik} \sinh(K(k))_{\alpha,\gamma} \langle \Omega, c_\gamma(k) \mathfrak{E}_\beta (1,k) \Omega\rangle \\
 & \hspace{3.4cm} +\sum_{\gamma\in\Ik} \sinh(K(k))_{\beta,\gamma} \langle \Omega, \mathfrak{E}  ^*_\alpha(1,k) c^*_\gamma(k) \Omega\rangle +\langle \Omega, \mathfrak{E}^*_\alpha(1,k) \mathfrak{E}_\beta(1,k) \Omega\rangle\Bigg].
\end{align*}
Finally, for the third interaction term we find
\begin{equation} \label{eq:int3} 
\begin{split} 
\langle \xi, \Hbb_\textnormal{int}^{(3)} \xi \rangle & = 2 \hbar\kappa \Re\sum_{k\in\north} \lvert k\rvert \tr \W^{+-}(k) \sinh(K(k)) \cosh(K(k)) + \mathfrak{E}_\textnormal{int}^{(3)}  \\
& = \hbar\kappa \sum_{k\in\north} \lvert k\rvert \tr \Wt(k) \sinh(K(k)) \cosh(K(k))  + \mathfrak{E}_\textnormal{int}^{(3)}\,,  \end{split}\end{equation}
where
\[\W^{+-}(k)_{\alpha,\beta} = \left\{ \begin{array}{cl} g(k) v_\alpha(k) v_\beta(k) & \text{for }\alpha \in \Ikp \text{ and } \beta \in \Ikm \\
0 & \text{otherwise} \end{array}\right.\;;\] 
we used the fact that all terms are real to write the more symmetric expression in terms of $\Wt(k) = \W^{+-}(k) + \W^{-+}(k)$ (the latter is defined by exchanging the role of $\Ikp$ and $\Ikm$ in the former). The error term $\mathfrak{E}_\textnormal{int}^{(3)}$ is given by 
\begin{align*}
\mathfrak{E}_\textnormal{int}^{(3)} & = 2\hbar\kappa\Re\sum_{k\in\north} \lvert k\rvert g(k) \sum_{\alpha\in\Ikp} \sum_{\beta\in\Ikm} v_\alpha(k) v_\beta(k) \Bigg[ \sum_{\gamma\in\Ik} \sinh(K(k))_{\gamma,\alpha} \langle \Omega, c_\gamma(k) \mathfrak{E}^*_\beta(1,k) \Omega\rangle \\
& \quad + \sum_{\gamma\in\Ik} \cosh(K(k))_{\gamma,\beta} \langle \Omega,\mathfrak{E}^*_\alpha(1,k) c^*_\gamma(k) \Omega\rangle + \langle \Omega, \mathfrak{E}^*_\alpha(1,k) \mathfrak{E}^*_\beta(1,k) \Omega\rangle\Bigg]
\end{align*}
Combining \eqref{eq:int1}, \eqref{eq:int2}, \eqref{eq:int3} and \eqref{eq:QN0-deco}, we conclude that
\[
\langle \xi, Q_N^{\textnormal{B}} \xi \rangle = \hbar\kappa \sum_{k\in \north} \lvert k\rvert \tr \left( \W(k) \sinh^2 (K(k)) + \Wt (k) \sinh(K(k)) \cosh(K(k)) \right) + \mathfrak{E}_\textnormal{int}
\]
with $ \mathfrak{E}_\textnormal{int} =  \mathfrak{E}^{(1)}_\textnormal{int} +  \mathfrak{E}^{(2)}_\textnormal{int} +  \mathfrak{E}^{(3)}_\textnormal{int}$. To control the error term $\mathfrak{E}^{(1)}_\textnormal{int}$, we decompose it as 
\begin{align*}
 \mathfrak{E}_\textnormal{int}^{(1)} & = \hbar\kappa \sum_{k\in\north} \lvert k\rvert g(k) \sum_{\alpha,\beta\in \Ikp} v_\alpha(k) v_\beta(k) \Bigg[ \sum_{\gamma\in\Ik} \sinh(K(k))_{\alpha,\gamma} \langle \Omega, c_\gamma(k) \mathfrak{E}_\beta (1,k) \Omega\rangle \\
 & \quad +\sum_{\gamma\in\Ik} \sinh(K(k))_{\beta,\gamma} \langle \Omega, \mathfrak{E}^*_\alpha(1,k) c^*_\gamma(k) \Omega\rangle +\langle \Omega, \mathfrak{E}^*_\alpha (1,k) \mathfrak{E}_\beta (1,k) \Omega\rangle\Bigg] \\ & =: \mathfrak{E}_\textnormal{int}^{(1,1)} +  \mathfrak{E}_\textnormal{int}^{(1,2)} +  \mathfrak{E}_\textnormal{int}^{(1,3)}  .
\end{align*}
Recall that $u_\alpha(k)^2 = \lvert \hat{k}\cdot\hat{\omega}_\alpha \rvert \leq 1$ and hence, by Proposition~\ref{prp:counting}, 
$v_\alpha(k) \leq \sqrt{\frac{C}{M}} u_\alpha(k) \leq \sqrt{\frac{C}{M}}$. Thus, using Proposition \ref{lem:bog1} and Cauchy-Schwarz, we find
\begin{align*} 
| \mathfrak{E}_\textnormal{int}^{(1,3)} | &\leq 
\Big| \hbar\kappa \sum_{k\in\north} \lvert k\rvert g(k) \sum_{\alpha,\beta\in \Ikp} v_\alpha(k) v_\beta(k) \langle \Omega, \mathfrak{E}^*_\alpha(1,k) \mathfrak{E}_\beta(1,k) \Omega\rangle \Big| \\
& \leq \hbar\kappa\sum_{k\in\north} \lvert k\rvert g(k) \sum_{\alpha \in \Ikp} \sqrt{\frac{C}{M}} \norm{\mathfrak{E}_\alpha(1,k) \Omega} \sum_{\beta \in \Ikp} \sqrt{\frac{C}{M}} \norm{\mathfrak{E}_\beta(1,k) \Omega}\\
& \leq C {\hbar} \sum_{k\in\north} \lvert k\rvert \hat{V}(k) \sum_{\alpha\in\Ikp} \norm{\Efrak_\alpha(1,k)\Omega}^2\,.
\intertext{(Recall that $|\Ikp| = I_k \leq M/2$.) With \eqref{eq:bogerrorbound}, we get}
| \mathfrak{E}_\textnormal{int}^{(1,3)} | & \leq \hbar \frac{C}{\nfrak^4} \sup_{\lambda \in [0,1]} \langle T_\lambda \Omega, (\Ncal+2)^{3} T_\lambda \Omega\rangle  \sum_{k\in\north} \lvert k\rvert \hat{V}(k)  e^{2\norm{K(k)}\HS} \Big[ \sum_{l\in \north} \norm{K(l)}\HS \Big]^2. 
\end{align*}
Lemma \ref{lem:inverses} and Proposition \ref{prp:particlenumber} imply that $|\mathfrak{E}_\textnormal{int}^{(1,3)}| \leq C \hbar / \nfrak^4$. The term $\mathfrak{E}_\textnormal{int}^{(1,1)}$ can be controlled similarly:
\begin{align*}
|\mathfrak{E}_\textnormal{int}^{(1,1)}|& \leq \Big\lvert \hbar\kappa \sum_{k \in\north} \lvert k\rvert g(k) \sum_{\alpha,\beta \in \Ikp} v_\alpha(k) v_\beta(k) \sum_{\gamma \in \Ik} \sinh(K(k))_{\alpha,\gamma} \langle \Omega, c_\gamma(k) \Efrak_\beta(k,1) \Omega \rangle \Big\rvert \\
& \leq \frac{C \hbar}{M} \sum_{k\in \north} \lvert k\rvert \hat{V}(k) \sum_{\alpha,\beta \in \Ikp} \Big\| \sum_{\gamma \in \Ik} \sinh(K(k))_{\alpha,\gamma} c^*_\gamma(k) \Omega\Big\| \norm{\Efrak_\beta(k,1)\Omega} \\
& \leq \frac{C \hbar}{M} \sum_{k\in \north} \lvert k\rvert \hat{V}(k) \sum_{\alpha,\beta \in \Ikp} \Big( \sum_{\gamma \in \Ik} \lvert \sinh(K(k))_{\alpha,\gamma}\rvert^2 \Big)^{1/2} \norm{ (\Ncal+1)^{1/2} \Omega} \norm{\Efrak_\beta(k,1)\Omega}\;;
\intertext{applying Cauchy-Schwarz in $\alpha$ and in $\beta$, using $|\Ikp| = I_k \leq M/2$ and \eqref{eq:bogerrorbound} we arrive at}
|\mathfrak{E}_\textnormal{int}^{(1,1)}|& \leq \frac{C\hbar}{\nfrak^2} \sup_{\lambda \in[0,\lambda]}\norm{ (\Ncal+2)^{3/2} T_\lambda\xi}  \sum_{k \in\north} \lvert k\rvert \hat{V}(k) e^{2\norm{K(k)}\HS} \sum_{l\in\north} \norm{K(l)}\HS\,.
\end{align*}
Again, Lemma \ref{lem:inverses} and Proposition \ref{prp:particlenumber} show that $|\mathfrak{E}_\textnormal{int}^{(1,1)}| \leq C \hbar / \nfrak^2$. Analogously, we obtain also $|\Efrak_\textnormal{int}^{(1,2)}| \leq C \hbar/ \nfrak^2$. Hence $|\Efrak_\textnormal{int}^{(1)}| \leq C \hbar/ \nfrak^2$. 

The error term $\Efrak_\textnormal{int}^{(2)}$ differs from $\Efrak_\textnormal{int}^{(1)}$ only in the replacement of the index set $\mathcal{I}_{k}^{+}$ by $\mathcal{I}_{k}^{-}$. Therefore, we find $\lvert\Efrak_\textnormal{int}^{(2)}\rvert \leq C \hbar/ \nfrak^2$. As for the error term $\mathfrak{E}_\textnormal{int}^{(3)}$, it also differs from $\Efrak_\textnormal{int}^{(1)}$ in the index set, some hermitian conjugations, and the appearance of a $\cosh$ instead of a $\sinh$. The estimates however remain valid and we also obtain $|\Efrak_\textnormal{int}^{(3)}| \leq C \hbar/ \nfrak^2$. 
\end{proof}

\subsection{Proof of the Main Theorem}
\begin{proof}[Proof of Theorem \ref{thm:main}]
Recall the definition \eqref{eq:Hcorr} of the correlation Hamiltonian and the decomposition \eqref{eq:QN0-def} of the quartic interaction $Q_N$. Combining the results of Section~\ref{sec:nonboson}, Section~\ref{sec:direx}, Section~\ref{sec:lin}, Proposition~\ref{prp:kin} and Proposition~\ref{prp:int}, we conclude that   
\[
\begin{split}& \langle \xi, \mathcal{H}_{\textnormal{corr}} \xi \rangle \\ &= \hbar\kappa \sum_{k\in \north} \lvert k\rvert \tr \left( (\D(k)+\W(k)) \sinh^2 (K(k)) + \Wt(k) \sinh(K(k)) \cosh(K(k)) \right) +  \mathfrak{E} \end{split}
\]
for an error $\mathfrak{E}$ such that 
\[
\lvert \mathfrak{E}\rvert \leq C \Big[ \frac{1}{N} + \frac{1}{M} + \frac{\hbar}{\nfrak^2} + \hbar N^{-\delta/2} \Big]
\]
with $\hbar = N^{-1/3}$ and $\nfrak = N^{1/3-\delta/2} M^{-1/2}$.

\medskip

To evaluate the main part of the expectation value explicitly, notice that by definition \eqref{eq:Kk} of $K(k)$ we have
\[
\sinh(K(k)) = \frac{1}{2} \left( \lvert S_1(k)^T\rvert - \lvert S_1(k)^T\rvert^{-1} \right), \quad \cosh(K(k)) = \frac{1}{2} \left( \lvert S_1(k)^T\rvert + \lvert S_1(k)^T\rvert^{-1} \right)\,.
\]
Notice also that $S_1(k) S_1(k)^T = \lvert S_1(k)^T\rvert^2$ and $ \left( \lvert S_1(k)^T\rvert^{-1} \right)^2 = S_2(k) S_2(k)^T$, where \[S_2(k) = \left( D(k) + W(k) - \tilde{W}(k)\right)^{-1/2} E(k)^{1/2}.\]
Consequently
\[
\begin{split}\sinh(K(k)) \cosh(K(k)) & = \frac{1}{4} \left( \lvert S_1(k)^T\rvert - \lvert S_1(k)^T\rvert^{-1} \right)^T \left( \lvert S_1(k)^T\rvert + \lvert S_1(k)^T\rvert^{-1} \right) \\
& = \frac{1}{4} \left( S_1(k) S_1(k)^T - S_2(k) S_2(k)^T \right).\end{split}
\]
Likewise
\[
\begin{split}\sinh^2(K(k)) & = \frac{1}{4} \left( \lvert S_1(k)^T\rvert - \lvert S_1(k)^T\rvert^{-1} \right)^T \left( \lvert S_1(k)^T\rvert - \lvert S_1(k)^T\rvert^{-1} \right) \\
& = \frac{1}{4} \left( S_1(k) S_1(k)^T + S_2(k) S_2(k)^T - 2 \id\right).\end{split}
\]
Now using the explicit form \eqref{eq:S1} of $S_1(k)$, $E(k)$, and $S_2(k)$, this can be simplified to yield
\begin{align}
\langle \xi , \mathcal{H}_{\textnormal{corr}} \xi \rangle &= \frac{\hbar\kappa}{4} \sum_{k\in \north} \lvert k\rvert \Big( \tr \left( \D(k)+\W(k)+\Wt(k)\right) S_1(k) S_1(k)^T \nonumber\\
&\hspace{2.7cm} + \tr \left( \D(k)+\W(k)-\Wt(k)\right) S_2(k) S_2(k)^T \Big) \nonumber\\
& \quad - \frac{\hbar\kappa}{2} \sum_{k\in \north} \lvert k\rvert \tr \big( \D(k)+\W(k) \big)  + \mathfrak{E}\nonumber\\
&= \hbar\kappa \sum_{k\in \north} \lvert k\rvert \left( \frac{1}{2} \tr E(k) - \frac{1}{2} \tr \big(\D(k)+\W(k)\big)\right)  + \mathfrak{E}\,. \label{eq:final}
\end{align}
We are left with evaluating the traces in \eqref{eq:final}.

\paragraph{Evaluation of the Traces.} For simplicity, we shall drop the $k$-dependence in the notation (we will restore it in \eqref{eq:fermigsenergy}). Recall the block diagonalization \eqref{eq:blockdiagonalization}, by which
\begin{equation}
\label{eq:trE}\begin{split}\frac{1}{2}\tr E & = \frac{1}{2}\tr \left[ \begin{pmatrix}
			d & 0 \\ 0 & d+2b
                     \end{pmatrix}^{1/2}
                     \begin{pmatrix}
                      d+2b & 0 \\ 0 & d
                     \end{pmatrix}
                     \begin{pmatrix}
                      d & 0 \\ 0 & d+2b
                     \end{pmatrix}^{1/2}\right]^{1/2}\\
& = \frac{1}{2}\tr \left[ d^{1/2} (d+2b) d^{1/2}\right]^{1/2} + \frac{1}{2}\tr \left[(d+2b)^{1/2} d (d+2b)^{1/2}\right]^{1/2}\\
& = \tr \left[ d^{1/2} (d+2b) d^{1/2}\right]^{1/2}, \end{split}\end{equation}
since $d^{1/2} (d+2b) d^{1/2}$ and $(d+2b)^{1/2} d (d+2b)^{1/2}$ have the same spectrum. To calculate this trace, notice that
\[
d^{1/2}(d+2b)d^{1/2} = d^2 + 2g \lvert \tilde u\rangle\langle \tilde u \rvert
\]                                                                                            is a rank-one perturbation of a diagonal operator, 
with diagonal part $d^2 = \diag(u_\alpha^4: \alpha=1,\ldots \ik)$ and with $\tilde u = ( v_1 u_1, \ldots , v_{\ik} u_{\ik} ) \in \Rbb^{\ik}$.

The resolvent of a matrix with rank-one perturbation can easily be calculated: For any invertible matrix $A \in \Cbb^{n\times n}$, and $x,y \in \Cbb^n$,
\[
(A+\lvert x\rangle\langle y\rvert)^{-1} = A^{-1} - \frac{A^{-1}\lvert x\rangle\langle y\rvert A^{-1}}{1+\langle y,A^{-1} x\rangle}
\]
whenever the right-hand side is well-defined. So for $\lambda \in [0,\infty)$ we find
\[
\left(d^2 + 2g\lvert \tilde u\rangle\langle \tilde u\rvert+\lambda^2\right)^{-1} = \left(d^2 + \lambda^2 \right)^{-1} - \frac{2g}{1+2g \sum_{\alpha=1}^{I_k} \frac{u_\alpha^2 v_\alpha^2}{u_\alpha^4 + \lambda^2}} \left\lvert w \right\rangle\left\langle  w\right\rvert\;,
\]
with $w \in \Rbb^{\ik}$ defined by $w_\alpha = u_\alpha v_\alpha(u_\alpha^4 + \lambda^2)^{-1}$.                                                                                                                                                                                                                                                                                                                                     By functional calculus, for any non-negative operator $A$ we have the identity
\[
\sqrt{A}= \frac{2}{\pi}\int_0^\infty  \left(\id - \frac{\lambda^2}{A+\lambda^2} \right)\di\lambda\,.
\]
Using the integral identity twice we find
\begin{align*}
 \tr \left[ d^{1/2} (d+2b) d^{1/2}\right]^{1/2}
& = \frac{2}{\pi} \int_0^{\infty}
 \tr \left( \id - \frac{\lambda^2}{d^2 + \lambda^2}\right) \di\lambda+ \frac{2}{\pi} \int_0^{\infty}   \frac{\lambda^2\, 2g}{1+2g\sum_{\alpha=1}^{\ik} \frac{u_\alpha^2 v_\alpha^2}{u_\alpha^4+\lambda^2}}\norm{w}^2\di\lambda\\
 & = \tr d + \frac{2}{\pi} \int_0^\infty   \frac{\lambda^2}{1+2g\sum_{\alpha=1}^{\ik} \frac{u_\alpha^2 v_\alpha^2}{u_\alpha^4+\lambda^2}} 2g \sum_{\alpha=1}^{\ik} \frac{u_\alpha^2 v_\alpha^2}{(u_\alpha^4+\lambda^2)^2}\di \lambda
\,.
\end{align*}
Restoring the $k$-dependence, let 
\[
f_k (\lambda) := 1+2g(k)\sum_{\alpha=1}^{I_k} \frac{u_\alpha(k)^2 v_\alpha(k)^2}{u_\alpha(k)^4+\lambda^2}\, . 
\]
Integrating by parts (noting that the boundary terms vanish since $\log f_k(\lambda) \sim 1/\lambda^2$), we find 
\begin{align*}
\tr \left[ d^{1/2} (d+2b) d^{1/2}\right]^{1/2} = \frac{1}{2}\tr D - \frac{1}{\pi} \int_0^\infty  \lambda \frac{f_k'(\lambda)}{f_k (\lambda)}\di \lambda
& = \frac{1}{2}\tr D + \frac{1}{\pi} \int_0^\infty  \log f_k (\lambda)\di \lambda\,.
\end{align*}
Thus, inserting in \eqref{eq:trE} and then in \eqref{eq:final}, we obtain 
\begin{equation}\label{eq:fermigsenergy}\langle \xi,\mathcal{H}_{\textnormal{corr}} \xi\rangle = \hbar\kappa \sum_{k\in\north} \lvert k\rvert \left( \frac{1}{\pi}\int_0^\infty  \log f_k (\lambda)\di\lambda - g(k) \sum_{\alpha=1}^{I_k}v_\alpha(k)^2 \right) + \mathfrak{E} \end{equation}
where we used that according to \eqref{eq:Wdef} $\tr \W = 2 \tr b =2g\sum_{\alpha=1}^{\ik} v_\alpha^2$. 

\paragraph{Convergence to the Gell-Mann--Brueckner formula.} To conclude the proof of Theorem \ref{thm:main}, we show that \eqref{eq:fermigsenergy} reproduces the Gell-Mann--Brueckner formula as stated in the theorem. Let
\[
\tilde{f}_k (\lambda) := 1+ 4\pi g(k)\left(1-\lambda \arctan\left( \frac{1}{{\lambda}} \right) \right).
\]
We claim that
\begin{equation}\label{eq:gellmann} \begin{split}& \Big\lvert \Big( \frac{1}{\pi}\int_0^\infty  \log f_k (\lambda)\di\lambda  - g(k) \sum_{\alpha=1}^{I_k}v_\alpha(k)^2 \Big) - \Big( \frac{1}{\pi}\int_0^\infty  \log \tilde{f}_k (\lambda)\di\lambda - g(k)\pi \Big) \Big\rvert \\
& \hspace{6.7cm} \leq C \left( M^{1/4 }N^{-\frac{1}{6}+\frac{\delta}{2}} + N^{-\frac{\delta}{2}} + M^{-\frac{1}{4}}N^{\frac{\delta}{2}} \right)\,.\end{split}
\end{equation}
Since $\log \tilde{f}_k (\lambda) = g(k) = 0$ for all $\lvert k\rvert > R$, inserting \eqref{eq:gellmann} into \eqref{eq:fermigsenergy} we obtain 
\[
\langle \xi,\mathcal{H}_{\textnormal{corr}} \xi\rangle =  \hbar\kappa\!\!\sum_{k\in\north}\!\! \lvert k\rvert\! \left( \frac{1}{\pi}\!\int_0^\infty\! \log\! \left[ 1 + 4\pi g(k) \left(1- {\lambda} \arctan \left(\frac{1}{{\lambda}}\right)\right) \right]\di\lambda - g(k)\pi \right) + \tilde{\mathfrak{E}}
\]
with an error 
\[
\begin{split} \lvert \tilde{\mathfrak{E}}\rvert & \leq C \Big[ N^{-1} + M^{-1} + N^{-1+\delta} M\Big] + C\hbar \Big[  M^{1/4 }N^{-\frac{1}{6}+\frac{\delta}{2}} + N^{-\frac{\delta}{2}} + M^{-\frac{1}{4}}N^{\frac{\delta}{2}}  \Big]\,. \end{split}
\]
Recalling that $M = N^{1/3 + \epsilon}$ and optimizing over $0 < \epsilon < 1/3$, $0 < \delta < 1/6 - \epsilon /2$, we find (with $\epsilon = 1/27$ and $\delta = 2/27$), that $\lvert \tilde{\mathfrak{E}} \rvert \leq C \hbar N^{-1/27}$. Replacing the sum over $k \in \north$ by $1/2$ times the sum over $k \in \Zbb^3$, and replacing $\kappa = \kappa_0 + \Ocal(N^{-1/3})$ by $\kappa_0 = (3/4\pi)^{1/3}$ (using also the Lipschitz continuity of the logarithm), we arrive at \eqref{eq:mainresult}.

\medskip

We still have to show \eqref{eq:gellmann}. To this end, recall from Proposition \ref{prp:counting} that, in terms of the surface measure $\sigma$ of the patch $p_\alpha$ on the unit sphere, we have
\[
v_\alpha(k)^2 = \sigma(p_\alpha) u_\alpha(k)^2\Big( 1 + \mathcal{O}\Big(\sqrt{M}N^{-\frac{1}{3}+\delta}\Big) \Big)\,.
\]
Thus
\[
f_k(\lambda) = 1+2 g(k) \sum_{\alpha=1}^{I_k} \frac{u_\alpha(k)^2 v_\alpha(k)^2}{u_\alpha(k)^4+\lambda^2} = 1+2 g(k) \sum_{\alpha=1}^{I_k}\sigma(p_\alpha)\frac{u_\alpha(k)^4}{u_\alpha(k)^4+\lambda^2} + \mathcal{O}\left(\sqrt{M}N^{-\frac{1}{3}+\delta}\right)\,.
\]
We approximate this Riemann sum by the corresponding surface integral over a subset of $\Sbb^2$.
We write $\cos \theta_\alpha = \hat{k}\cdot\hat{\omega}_\alpha = u_\alpha(k)^2$ and $\varphi_\alpha$ for the azimuth of $\omega_\alpha$. We parametrize the surface integrals in the same spherical coordinate system\footnote{This is not the spherical coordinate system used to introduce patches in the first place, where inclination was measured with respect to $e_3$.} (i.\,e., the inclination $\theta$ is measured with respect to $k$, and the azimuth $\varphi$ in the plane perpendicular to $k$). We estimate every summand by
\begin{align*}
& \left\lvert \int_{p_\alpha} \frac{\cos^2 \theta}{\cos^2 \theta + \lambda^2} \di\sigma - \sigma(p_\alpha) \frac{\cos^2 \theta_\alpha}{\cos^2 \theta_\alpha +\lambda^2} \right\rvert\\
& \leq  \int_{p_\alpha} \left\lvert \frac{\cos^2 \theta}{\cos^2 \theta + \lambda^2}  - \frac{\cos^2 \theta_\alpha}{\cos^2 \theta_\alpha +\lambda^2} \right\rvert \di\sigma\\
&\leq  \iint_{\hat{\omega}(\theta,\varphi) \in p_\alpha} \left\lvert \frac{\cos^2 \theta}{\cos^2 \theta + \lambda^2}  - \frac{\cos^2 \theta_\alpha}{\cos^2 \theta_\alpha +\lambda^2} \right\rvert \lvert \sin \theta\rvert \di\theta \di\varphi\;.
\end{align*}
Bounding the difference using the supremum of the derivative
\begin{align*}
\left\lvert \int_{p_\alpha} \frac{\cos^2 \theta}{\cos^2 \theta + \lambda^2} \di\sigma - \sigma(p_\alpha) \frac{\cos^2 \theta_\alpha}{\cos^2 \theta_\alpha +\lambda^2} \right\rvert 
& \leq \sup_{\hat{\omega}(\theta,\varphi) \in p_\alpha} \left\lvert \frac{\di}{\di\theta} \frac{\cos^2 \theta}{\cos^2 \theta + \lambda^2} \right\rvert \frac{C}{\sqrt{M}} \sigma(p_\alpha)\, ,\tagg{supest}
\end{align*}
where we also used that, since the partition is diameter bounded,
$\sup_{(\theta,\varphi) \in p_\alpha} \lvert \theta-\theta_\alpha\rvert \leq C/\sqrt{M}$.
The derivative is bounded by
\[
\left\lvert \frac{\di}{\di\theta} \frac{\cos^2 \theta}{\cos^2 \theta + \lambda^2} \right\rvert \leq 2 \frac{\lambda^2}{\cos^2 \theta + \lambda^2 } \frac{\lvert\cos \theta\rvert \lvert \sin \theta\rvert}{\cos^2 \theta + \lambda^2 } \leq \frac{2}{\lvert \cos \theta\rvert}\,.
\]
Recall that $\alpha \in \{1,2,\ldots,I_k\}$, which by definition of the index set implies $\cos \theta_\alpha > N^{-\delta}$. The bound $\lvert \theta - \theta_\alpha \rvert \leq CM^{-1/2}$ implies that also $\cos \theta > CN^{-\delta}$. 
So \eqref{eq:supest} implies
\[
\left\lvert \int_{p_\alpha} \frac{\cos^2 \theta}{\cos^2 \theta + \lambda^2} \di\sigma - \sigma(p_\alpha) \frac{\cos^2 \theta_\alpha}{\cos^2 \theta_\alpha +\lambda^2} \right\rvert \leq C \frac{N^\delta}{M^{3/2}} \,.
\]
Since the number of patches is at most $M$ we conclude that
\[
\Big| \sum_{\alpha=1}^{I_k}\sigma(p_\alpha) \frac{u_\alpha(k)^4}{u_\alpha(k)^4+\lambda^2} - \int_{\Sbb^2_\textnormal{reduced}} \frac{\cos^2 \theta}{\cos^2 \theta + \lambda^2} \di\sigma \Big| \leq C \frac{N^{\delta}}{\sqrt{M}}\,.
\]
Here we wrote $\Sbb^2_\textnormal{reduced}$ for a unit half-sphere excluding the collar of width $N^{-\delta}$ and the corridors $p_\textnormal{corri}$.
Since $\cos^2 \theta/(\cos^2 \theta + \lambda^2) \leq 1$ we can compare to the integral over the whole unit half-sphere $\Sbb^2_\textnormal{half}$, 
\[
\left\lvert \int_{\Sbb^2_\textnormal{reduced}} \frac{\cos^2 \theta}{\cos^2 \theta + \lambda^2} \di\sigma  - \int_{\Sbb^2_\textnormal{half}} \frac{\cos^2 \theta}{\cos^2 \theta + \lambda^2} \di\sigma \right\rvert \leq C \left[ N^{-\delta} + M^{1/2} N^{-1/3} \right] \,.
\]
The surface integral over the unit half-sphere is easy to compute,
\begin{equation}\label{eq:integralidentity}\begin{split}\int_{\Sbb^2_\textnormal{half}}\!\! \frac{\cos^2 \theta}{\cos^2 \theta + \lambda^2} \di\sigma & = \int_{0}^{\pi/2}\!\! \di \theta \sin(\theta) \frac{\cos(\theta)^2}{\cos(\theta)^2+\lambda^2} \int_0^{2\pi}\!\! \di \varphi   = 2\pi \Big(1-{\lambda} \arctan\Big( \frac{1}{{\lambda}} \Big)\Big)\,.\end{split}\end{equation}
Since $g(k) = \kappa \hat{V}(k)$ is uniformly bounded (by assumption on $\hat{V}$), we conclude that
\[\left\lvert f(\lambda) - \tilde{f}(\lambda) \right\rvert \leq C \Big( \sqrt{M}N^{-\frac{1}{3}+\delta} + N^{-\delta} + \frac{N^\delta}{\sqrt{M}} \Big)\,.\]
Since for $x\geq 0$ the function $x \mapsto \log(1+x)$ has Lipschitz constant $1$ we get 
\[\left\lvert \log f(\lambda) - \log \tilde{f}(\lambda) \right\rvert \leq C \Big( \sqrt{M}N^{-\frac{1}{3}+\delta} + N^{-\delta} + \frac{N^\delta}{\sqrt{M}} \Big)\,.\]
It remains to compare the integrals over $\lambda$. Since $\log(1+x) \leq x$ for all $x \geq 0$, we have
\[\left\lvert \log f(\lambda) \right\rvert \leq 2g(k) \sum_{\alpha=1}^{I_k} \sigma(p_\alpha) \frac{u_\alpha(k)^4}{u_\alpha(k)^4 + \lambda^2} \leq 2g(k) \sum_{\alpha=1}^{I_k} \frac{C}{M} \frac{1}{\lambda^2} \leq \frac{C}{\lambda^2}\,,\]
where we used the two inequalities $0 \leq u_\alpha(k)^4 \leq 1$.
Using the integral identity \eqref{eq:integralidentity} it is easy to see that also
\[\left\lvert \log \tilde{f}(\lambda) \right\rvert \leq 4\pi g(k) \Big| 1-{\lambda}\arctan\left(\frac{1}{{\lambda}}\right) \Big| \leq \frac{C}{\lambda^2}\,.\]
Using the last three estimates, by splitting the integration at some $\Lambda > 0$ to be optimized in the last step, we obtain
\begin{align*}
& \left\lvert \frac{1}{\pi} \int_0^\infty \log f(\lambda)\di\lambda -  \frac{1}{\pi} \int_0^\infty  \log \tilde f(\lambda)\di\lambda \right\rvert \\
& \leq \frac{1}{\pi} \int_0^{\Lambda} \left\lvert \log f(\lambda) - \log \tilde f(\lambda) \right\rvert\di\lambda + \frac{1}{\pi} \int_{\Lambda}^\infty  \frac{8\pi g(k)}{\lambda^2}\di\lambda \\
& \leq C \Lambda\left( \sqrt{M}N^{-\frac{1}{3}+\delta} + N^{-\delta} + \frac{N^\delta}{\sqrt{M}} \right) + C \Lambda^{-1}\\
& \leq C \left( M^{1/4 }N^{-\frac{1}{6}+\frac{\delta}{2}} + N^{-\frac{\delta}{2}} + M^{-\frac{1}{4}}N^{\frac{\delta}{2}} \right)\,.\tagg{err}
\end{align*}

By a similar (simpler) Riemann sum argument we obtain
\begin{align*}-g(k)\sum_{\alpha=1}^{I_k} v_\alpha^2(k) & = -g(k) \sum_{\alpha=1}^{I_k} \sigma(p_\alpha)  u_\alpha^2(k) \left( 1 + \mathcal{O}\left(\sqrt{M}N^{-\frac{1}{3}+\delta}\right) \right) \\
& = - g(k)\pi  + \Ocal\left(\sqrt{M}N^{-\frac{1}{3}+\delta} + N^{-\delta}\right)\end{align*}
where the error is obviously smaller than \eqref{eq:err}. This concludes the proof of \eqref{eq:gellmann}.
 \end{proof}

\section{Counting Particle-Hole Pairs in Patches}\label{sec:continuumapprox}

In this section we prove Proposition \ref{prp:counting}, which is concerned with estimating the number 
\begin{equation}
\label{eq:number}
n_{\alpha,k}^2 = \sum_{\substack{p \in B_\textnormal{F}^c \cap B_\alpha\\h \in B_\textnormal{F} \cap B_\alpha}} \delta_{p-h,k}
\end{equation}
of particle--hole pairs with momentum $p-h = k$ in patch $\alpha$ under the condition that $\hat{\omega}_\alpha\cdot \hat{k} \geq N^{-\delta}$. Recall that $p_\alpha$ is a patch on the unit sphere, and $P_\alpha = k_\textnormal{F} p_\alpha$.

To illustrate the idea of the proof we first consider $k = e_3 = ( 0, 0, 1)$. Consider the lattice lines $L_{n} := \{n+tk: t\in\Rbb\}$, $n \in \Zbb e_1 + \Zbb e_2 \subset \Zbb^3$. For each lattice line $L_{n}$ intersecting $P_\alpha$ there is exactly one contribution to the sum  \eqref{eq:number}---
in fact, a simple geometric consideration shows that since $N^{-\delta} \gg M^{-1/2}$ (which is implied by the assumption $\delta \leq \frac{1}{6} - \frac{\epsilon}{2}$) a line never enters the Fermi ball at such a small angle (measured with respect to the tangent plane of the Fermi surface) that it would cross the surface immediately a second time and leave the Fermi ball without picking up a pair (i.\,e., the situation of Figure \ref{fig:smallangle} is excluded due to $\hat{\omega}_\alpha\cdot \hat{k} \geq N^{-\delta}$).
\begin{figure}
\centering
 \begin{minipage}[t]{0.45\textwidth}
 \centering
\begin{tikzpicture}
\clip (-3.4,3) rectangle (3.4,6.5);

\draw(0,-6) circle (11.9cm);
\draw[very thick] (0,-6) circle (11cm);
\draw (0,-6) circle (10.1cm);

\draw[->] (-2.4,4.85)--(1.5,5.05);
\draw[dashed] (-6.3,4.65)--(5.4,5.25);
	  \draw [fill]  (-2.4,4.85) circle [radius=0.04];
	  	  \draw [fill]  (1.5,5.05) circle [radius=0.04];
\node [above] at (+.9,5) {$k$};

\draw[thick,->] (0,-6)--(0,5);
   \node [right] at (0,4.7) {$\omega_\alpha$};

 \draw (0,-6)--(-3,6.5);
\draw (0,-6)--(3,6.5);
 \end{tikzpicture}
   \caption{Fermi surface in bold. A line (dashed) intersects the patch but no particle--hole pair is picked up because both ends of $k$ would be outside the Fermi ball. This could only happen if $k$ was very long (excluded due to $k \in \supp\hat{V}$) or almost tangent to the Fermi surface (excluded by $\hat{\omega}_\alpha\cdot \hat{k} \geq N^{-\delta}$).}\label{fig:smallangle}
 \end{minipage}\hfill
 \begin{minipage}[t]{0.45\textwidth}
 \centering
\begin{tikzpicture}
\clip (-3.4,3) rectangle (3.4,6.5);

\draw(0,-6) circle (11.9cm);
\draw[very thick] (0,-6) circle (11cm);
\draw (0,-6) circle (10.1cm);

\draw[->] (2,4.2)--(2.9,5.3);
\draw[dashed] (1.1,3.1)--(3.8,6.4);
	  \draw [fill=white]  (2,4.2) circle [radius=0.04];
	  	  \draw [fill]  (2.9,5.3) circle [radius=0.04];
\node [above] at (+2.45,4.8) {$k$};

\draw[thick,->] (0,-6)--(0,5);
   \node [right] at (0,4.7) {$\omega_\alpha$};

 \draw (0,-6)--(-3,6.5);
\draw (0,-6)--(3,6.5);
 \end{tikzpicture} 
 \caption{Fermi surface in bold. A line (dashed) intersects the patch but no particle--hole pair is picked up because $k$ points from a hole momentum $h$ in the patch out into a corridor between patches. This can happen only for hole momenta near the boundary. Since the area of a patch grows faster with $N$ than its boundary length, the number of such lines is an error term of lower order.}
\label{fig:boundaryterm}
 \end{minipage}
\end{figure}
There is only one exception to this argument: A lattice line might cross the surface at a distance less than $R$ from a side of the patch. Depending on the angle it could then leave the patch to the side before picking up a pair, as represented in Figure \ref{fig:boundaryterm}. However, the number of such lines is of the same order as the length of the boundary. We can thus absorb this number in the circumference error from the Gauss argument (see next paragraph). 

So to leading order $n_{\alpha,k}^2$ is the number of lines $L_n$ intersecting $P_\alpha$. The number of such lines is equal to the number of lines intersecting the projection $P_\alpha^k$ of $P_\alpha$ to the plane spanned by $e_1$ and $e_2$; see Figure \ref{fig:projection}.
\begin{figure}
\centering
 \begin{minipage}[t]{0.45\textwidth}
 \centering
\begin{tikzpicture}
\draw [use as bounding box,white] (-1,-.5) rectangle (5,5.5);
  \draw[step=.5cm,white!85!black,thin] (-1,-.5) grid (5,5.5);
\coordinate (center) at (0,0);
\cercle{center}{4cm}{-7.9}{105.3};
\boldcercle{center}{4cm}{50}{30};
  \draw [thick,->] (0,-.5)--(0,5);
  \draw [thick,->] (-1,0)--(4.5,0);  
  \draw[dashed] (.7,4)--(.7,0);
  \draw[dashed] (2.6,3.1)--(2.6,0);
  \draw (1,5.5)--(1,-.5);
  \draw (1.5,5.5)--(1.5,-.5);
  \draw (2,5.5)--(2,-.5);
  \draw (2.5,5.5)--(2.5,-.5);  
  \draw[very thick] (.7,.05)--(2.6,0.05);  
   \draw [thick,->] (4,2.5)--(4,3);   
   \node [right] at (4,2.75) {$k$};
   \node [above] at (2,3) {$P_\alpha$};
 \end{tikzpicture}
   \caption{The number of lattice lines through the patch is the same as the number of lattice lines through the projection of the patch along $k$ onto the plane spanned by $e_1$ and $e_2$.}\label{fig:projection}
 \end{minipage}\hfill
 \begin{minipage}[t]{0.45\textwidth}
 \centering
 \begin{tikzpicture} 
\draw [use as bounding box,white] (-1,-.5) rectangle (5,5.5);
  \draw[step=.5cm,white!85!black,thin] (-1,-.5) grid (5,5.5);
\coordinate (center) at (0,0);
\cercle{center}{4cm}{-7.9}{105.3};
\boldcercle{center}{4cm}{50}{30};
  \draw [thick,->] (0,-.5)--(0,5);
  \draw [thick,->] (-1,0)--(4.5,0);  
  \draw [thick,->] (4,3)--(4.5,4.5);
  \node [right] at (4.2,3.6) {$k$};  
  \coordinate (A) at (-1,-2/3);
  \coordinate (B) at (3,-2); 
  
\node [below] at (.8,0) {$P_\alpha^k$}; 
  
  \clip (-1,0) rectangle (6,6);
  
  \foreach \x in {4,4.5,5}
    \foreach \y in {1.5,2,2.5}
    {
	  \coordinate (C) at (\y,\x);
	  \draw (C)--($(A)!(C)!(B)$);
	  
	  \draw [fill] (C) circle [radius=0.04];
	  \coordinate (Cp) at ($(C)+(-.5,-1.5)$);
	  \draw [fill=white] (Cp) circle [radius=0.04];
    }
    \coordinate (C) at (1,4);
    \draw (C)--($(A)!(C)!(B)$);
    \draw [fill] (C) circle [radius=0.04];
    \coordinate (Cp) at ($(C)+(-.5,-1.5)$);
	  \draw [fill=white]  (Cp) circle [radius=0.04];
    
    \coordinate (C) at (1,4.5);
    \draw (C)--($(A)!(C)!(B)$);
    \draw [fill] (C) circle [radius=0.04];
    \coordinate (Cp) at ($(C)+(-.5,-1.5)$);
	  \draw [fill=white]  (Cp) circle [radius=0.04];
    
     \coordinate (C) at (2.5,3.5);
     \draw (C)--($(A)!(C)!(B)$);
     \draw [fill] (C) circle [radius=0.04];
     \coordinate (Cp) at ($(C)+(-.5,-1.5)$);
	  \draw [fill=white]  (Cp) circle [radius=0.04];
     
    \coordinate (C) at (3,4.5);
    \draw (C)--($(A)!(C)!(B)$);
    \draw [fill] (C) circle [radius=0.04];
    \coordinate (Cp) at ($(C)+(-.5,-1.5)$);
	  \draw [fill=white]  (Cp) circle [radius=0.04];
    
  \coordinate (Ap) at ($(A)!(.7,4)!(B)$);
  \coordinate (Bp) at ($(A)!(2.6,3.1)!(B)$);
  \draw[dashed] (.7,4)--($(A)!(.7,4)!(B)$);
  \draw[dashed] (2.6,3.1)--($(A)!(2.6,3.1)!(B)$);
   \draw[very thick] (-.64,.05)--(1.6,.05);
 \end{tikzpicture}
 \caption{Particles and holes are indicated by black and white dots, respectively; they are paired along lines parallel to $k$.
The number of pairs per line is given by the greatest common divisor $\gcd(k_1,k_2,k_3)$ (here $ = 1$).}
\label{fig:projectpatch}
 \end{minipage}
\end{figure}
To count we use Gauss' classical argument (in two dimensions):
\[\left\lvert \left\{ L_{n} : n \in \Zbb e_1 + \Zbb e_2\right\} \cap P_{\alpha}^{k} \right\rvert = \lebesgue\left( P_{\alpha}^{k} \right) + \mathcal{O}\left(\text{circumference of }P_{\alpha}^{k}\right)\,,\]
where $\lebesgue$ is the two-dimensional Lebesgue measure in the plane. Hence we conclude that to leading order, $n_{\alpha,k}^2 = \lebesgue\left( P_\alpha^k\right)$ if $k = e_3$.

If $k = (0, 0, k_3)$ then for every lattice line there are $k_3$ contributing pairs. As illustrated in Figure \ref{fig:projectpatch}, for the general case we have to take into account that the distance of lattice points along the lines changes, and the density of intersection points in the $e_1$-$e_2$-plane changes.
\begin{proof}[Proof of Proposition \ref{prp:counting}]
We are going to prove that, assuming $\delta \leq \frac{1}{6} - \frac{\epsilon}{2}$ and $\alpha \in \Ikp$, the number of particle--hole pairs with momentum $k$ in patch $B_\alpha$ is
\begin{equation}\label{eq:rescounting}n_{\alpha,k}^2 = u_\alpha(k)^2 k_\textnormal{F}^2 \sigma(p_\alpha) \lvert k\rvert \left( 1 + \mathcal{O}\left(\sqrt{M}N^{-\frac{1}{3}+\delta}\right) \right)\,.\end{equation}
The statement of the proposition then follows immediately.

\bigskip

Let $k = (k_1, k_2, k_3)$, and consider a patch $P_\alpha$. 
Possibly reflecting at coordinate planes, we can assume that $k_1$, $k_2$, and $k_3$ are all non-negative, and without loss of generality we assume $k_3 \neq 0$ (if $k_3=0$ we would project onto another coordinate plane). Let $P_\alpha^k$ the projection of $P_\alpha$ along $k$ onto $\Rbb^2 \times \{0\}$, the plane spanned by $e_1$ and $e_2$.

\medskip

First we calculate $\lebesgue\left(P_\alpha^k\right)$. Consider the lines $\{k_\textnormal{F} \hat{\omega}(\theta,\varphi) + t k: t \in \Rbb\}$; their intersection with $\Rbb^2 \times \{0\}$ is at $t = -k_\textnormal{F}\hat{\omega}(\theta,\varphi)_3/k_3$; so
\[P_\alpha^k := \left\{  (x(\theta,\varphi), y(\theta,\varphi) , 0) = k_\textnormal{F} \hat{\omega}(\theta,\varphi) - \frac{k_\textnormal{F}}{k_3} \hat{\omega}(\theta,\varphi)_3 \, k: \hat{\omega}(\theta,\varphi) \in p_\alpha\right\} \subset \Rbb^2 \times \{0\}\,.\]
Writing $\Phi(\theta,\varphi) = (x(\theta,\varphi), y(\theta,\varphi))$, we find that $P_\alpha^{k}$ has two-dimensional Lebesgue measure
\[
\lebesgue\left(P_\alpha^k\right) = \int_{P_\alpha^k} \di x \di y = \int_{p_\alpha} \lvert \det D\Phi(\theta,\varphi) \rvert \di\theta \di \varphi\,.
\]
Using $\hat{\omega}(\theta,\varphi) = (\sin \theta \cos \varphi , \sin\theta \sin \varphi , \cos \theta )$ it is easy to calculate the Jacobi determinant
\[\begin{split}\lvert \det D\Phi(\theta,\varphi) \rvert & = k_\textnormal{F}^2 \frac{\lvert \sin \theta\rvert}{k_3} \lvert k_1 \sin \theta \cos \varphi + k_2 \sin \theta \sin \varphi + k_3 \cos \theta \rvert  = k_\textnormal{F}^2 \frac{\lvert \sin \theta\rvert}{k_3} \lvert k\cdot \hat{\omega}(\theta,\varphi) \rvert\,.\end{split}\]
Since the patch is diameter bounded we have $\lvert k\cdot \hat{\omega}(\theta,\varphi) \rvert = \lvert k\cdot\hat\omega_\alpha\rvert + \Ocal(M^{-1/2})$; and using $\lvert k\cdot \hat{\omega}_\alpha \rvert \geq N^{-\delta}$ to convert the additive error into a multiplicative error, this implies
\begin{equation}\label{eq:mu}\lebesgue\left( P_\alpha^k \right) = \frac{k_\textnormal{F}^2}{k_3} \int_{p_\alpha} \lvert k\cdot \hat{\omega}(\theta,\varphi) \rvert\, \lvert \sin \theta\rvert \di \theta \di\varphi = \frac{k_\textnormal{F}^2}{k_3} \lvert k\cdot \hat\omega_\alpha \rvert \left( 1+\Ocal\left(M^{-1/2} N^{\delta}\right)\right) \sigma(p_\alpha)\,.\end{equation}

We now determine the distance between neighboring lattice points along every line
\[L_{n} := \left\{ n+tk: t \in \Rbb \right\}\quad \text{where } n \in \Zbb^3\,.\]
Let $p := \gcd(k_1,k_2,k_3)$ be the greatest common divisor of the components of $k$. It is not difficult to see that the distance between neighboring lattice points on each line $L_n$ is $\lvert k\rvert/p$. Given a line $L_n$ intersecting $P_\alpha$, let $h \in L_{n} \cap B_\textnormal{F}$ be the lattice point closest to $P_{\alpha}$. Then on the line segment $\left\{ h+tk: t \in (0,1] \right\}$ there are $p$ lattice points; by shifting along the line, these correspond to $p$ particle--hole pairs contributing to $n_{\alpha,k}^2$. We conclude that $n_{\alpha,k}^2$ is to leading order the number of lattice lines $L_{n}$ intersecting $P_\alpha^{k}$, multiplied with $\gcd(k_1,k_2,k_3)$.

\medskip

We now determine how many lattice lines run through $P_\alpha^{k}$. Intersecting $L := \bigcup_{n \in \Zbb^3} L_{n}$ with $\Rbb^2 \times \{0\}$ we find $t= -n_3/k_3$. So
\[L \cap \left( \Rbb^2\times \{0\} \right) = \left\{ \left( n_1 - n_3 \frac{k_1}{k_3} ,  n_2 - n_3 \frac{k_2}{k_3} , 0 \right) : n \in \Zbb^3 \right\}\,.\]
This can be seen as the two-dimensional square lattice $\Zbb^2$ (the translates of the unit square indexed by $n_1$ and $n_2$) and a point pattern repeated in every lattice translation of the unit square. As soon as $n_3 k_1/k_3$ and $n_3 k_2/k_3$ simultaneously become integer, we start repeating the point pattern in another translate of the unit square. So the number of points in the unit square is the smallest integer $n_3$ such that both $n_3 k_1/k_3$ and $n_3 k_2/k_3$ are integer. We claim that this is $k_3/p$.

To prove this claim, consider the fraction $k_1/k_3$. Obviously $n_3 k_1/k_3$ is integer if and only if $n_3$ is a multiple of $k_3 / \gcd(k_1,k_3)$. Similarly $n_3 k_2/k_3$ is integer if and only if $n_3$ is a multiple of $k_3 / \gcd(k_2,k_3)$. So the number of points in the unit square is given by the least common multiple,
\[\#\text{points in unit square} = \lcm\left(\frac{k_3}{\gcd(k_1,k_3)},\frac{k_3}{\gcd(k_2,k_3)} \right).\]
From the standard identity $\gcd(a,b)\lcm(a,b)=\lvert ab\rvert$ for all $a,b \in \Zbb$ we get
\begin{align*}& \lcm\left(\frac{k_3}{\gcd(k_1,k_3)},\frac{k_3}{\gcd(k_2,k_3)} \right)  = \frac{k_3^2}{ \gcd(k_1,k_3)\gcd(k_2,k_3)\gcd\left(\frac{k_3}{\gcd(k_1,k_3)} , \frac{k_3}{\gcd(k_2,k_3)} \right)}
\intertext{using twice the fact that $m \gcd(a,b) = \gcd(ma,mb)$ for all $m\in \Nbb$;  then the same fact in inverse direction with $m=k_3$; then the fact $\gcd(a,b,c) = \gcd(a,\gcd(b,c))$ and the analogous identity for four integers}
& = \frac{k_3^2}{ \gcd\left({k_3}{\gcd(k_2,k_3)} , {k_3}{\gcd(k_1,k_3)} \right)}
 = \frac{k_3}{ \gcd\left({\gcd(k_2,k_3)} , {\gcd(k_1,k_3)} \right)} = \frac{k_3}{\gcd(k_1,k_2,k_3)}\,.
\end{align*}

In extension of Gauss' argument, the number of lines intersecting $P_\alpha^{k}$ is equal to the Lebesgue measure of $P_\alpha^{k}$ times the number of intersection points per unit square.
We thus conclude that
\begin{equation}\label{eq:nalpha}n_{\alpha,k}^2 = \lebesgue\left(P_\alpha^{k}\right) k_3 + \mathfrak{e}_{\alpha,k}\,.\end{equation}
The error term $\mathfrak{e}_{\alpha,k}$ is proportional to the circumference of $P_\alpha^{k}$, times the number of lines per unit square. Consider a patch that is not a spherical cap (the estimate for the two spherical caps works analogously); its circumference consists of four pieces. The first piece is parametrized by $\gamma(\varphi) := \Phi(\theta_\alpha+\Delta\theta_\alpha,\varphi)$, and has length
\[\int_{\varphi_\alpha-\Delta\varphi_\alpha}^{\varphi_\alpha+\Delta\varphi_\alpha} \lvert \dot\gamma(\varphi)\rvert \di \varphi = 2\Delta\varphi_\alpha k_\textnormal{F} \lvert \sin\left(\theta_\alpha+\Delta\theta_\alpha \right) \rvert = \Ocal\left(\frac{k_\textnormal{F}}{\sqrt{M}}\right)\,.\]
The second piece, parametrized by $\varphi \mapsto \Phi(\theta_\alpha-\Delta\theta_\alpha,\varphi)$, is of the same order.
The third piece is parametrized by $\tilde{\gamma}(\theta) := \Phi(\theta,\varphi_\alpha+\Delta\varphi_\alpha)$. By straightforward estimates
\[\begin{split}& \lvert \dot{\tilde{\gamma}}(\theta)\rvert^2 \\
& = \frac{k_\textnormal{F}^2}{k_3^2} \Big\lvert k_3^2 \cos^2 \theta + (k_1^2 + k_2^2)\sin^2 \theta + 2 k_3 \cos \theta \sin \theta \big( k_1 \sin(\varphi_\alpha + \Delta\varphi_\alpha) + k_2 \cos(\varphi_\alpha + \Delta\varphi_\alpha) \big) \Big\rvert \\ & \leq 2 \frac{k_\textnormal{F}^2}{k_3^2} \lvert k \rvert^2\,.\end{split}\]
Integrating and recalling that $\Delta\theta_\alpha = \Ocal(M^{-1/2})$, the length of this piece is at most of order $k_\textnormal{F}/\sqrt{M}$. The fourth piece has length of the same order as the third piece. We conclude that $\lvert \mathfrak{e}_{\alpha,k}\rvert = \Ocal(k_F M^{-1/2})$. Combining \eqref{eq:nalpha} with \eqref{eq:mu}, and using $u_\alpha(k)^2 \geq N^{-\delta}$ to convert the additive error into a multiplicative error (the new contribution is the dominating error), we obtain \eqref{eq:rescounting}.
\end{proof}

\appendix

\section{The Bosonic Effective Theory}\label{sec:bosonicapproximation}
In this section, we start with the Sawada-type effective Hamiltonian given by \eqref{eq:eff-H} and \eqref{eq:heff}, now assuming the exact canonical commutation relations 
\[
[c_{\alpha} (k), c_{\beta} (l)] = 0 = [c^*_{\alpha} (k), c^*_{\beta} (l)], \qquad [c_{\alpha} (k), c^*_{\beta} (l)] = \delta_{\alpha,\beta} \delta_{k,l}\,.
\]
We show how to diagonalize $h_\textnormal{eff}(k)$ and therefore how to compute the ground state of $\Hcal_\textnormal{eff}$, which inspired the choice of the trial state \eqref{eq:bos-gs}. 

\subsection{Diagonalization of the Effective Hamiltonian}
We follow \cite{GS13}. Dropping the $k$-dependence where no confusion arises, we write the effective Hamiltonian in standard form,
\[h_\textnormal{eff}(k) = \mathbb{H} - \frac{1}{2} \tr (\D+\W)\,,
\]
with
\[\mathbb{H} = \frac{1}{2} \begin{pmatrix} (c^*)^T & c^T \end{pmatrix} \begin{pmatrix} \D+\W & \Wt \\ \Wt & \D+\W\end{pmatrix} \begin{pmatrix}c\\c^* \end{pmatrix},\quad  c = \begin{pmatrix} \vdots \\ c_\alpha \\ \vdots\end{pmatrix},\quad c^* = \begin{pmatrix} \vdots \\ c^*_\alpha \\ \vdots \end {pmatrix},\]
where $c^T = \begin{pmatrix} \cdots & c_\alpha & \cdots \end{pmatrix}$.
The $2I_k\times 2I_k$-matrices $\D$, $\W$, and $\Wt$ are defined in \eqref{eq:blocks2}; they are real and symmetric.  

The Segal field operators $\phi = \begin{pmatrix} \cdots & \phi_\alpha & \cdots \end{pmatrix}^T$ and $\pi = \begin{pmatrix} \cdots & \pi_\alpha & \cdots \end{pmatrix}^T$ are defined by 
\begin{equation}\label{eq:defsegalfields}\begin{pmatrix}
   c \\ c^*
  \end{pmatrix} := \Theta \begin{pmatrix} \phi \\ \pi\end{pmatrix}, \quad \Theta := \frac{1}{\sqrt{2}} \begin{pmatrix} 1 & i\\ 1& -i \end{pmatrix}.
\end{equation}
Notice that $\phi = \frac{1}{\sqrt{2}}(c+c^*) = \phi^*$ and $\pi = \frac{i}{\sqrt{2}}(c^*-c) = \pi^*$.
In terms of the Segal field operators we have
\[\begin{split}\mathbb{H} & = \begin{pmatrix}\phi^T & \pi^T \end{pmatrix} \mathfrak{M} \begin{pmatrix} \phi \\ \pi\end{pmatrix}, \\ \mathfrak{M} & = \frac{1}{2} \Theta^* \begin{pmatrix} \D+\W & \Wt \\ \Wt & \D+\W\end{pmatrix} \Theta = \frac{1}{2} \begin{pmatrix}                                                                                                                                                                                                    \D+\W+\Wt & 0\\ 0 & \D+\W-\Wt                                                                                                                                                                            \end{pmatrix} \in \Cbb^{4I_k\times 4I_k}.\end{split}
\]
The commutator relations of the Segal field operators are invariant under symplectic transformations (which correspond to Bogoliubov transformations of the bosonic creation and annihilation operators). We introduce
\[E := \left( (\D+\W-\Wt)^{1/2} (\D+\W+\Wt) (\D+\W-\Wt)^{1/2} \right)^{1/2} \in \Cbb^{2I_k\times 2I_k}\]
and the symplectic matrix\footnote{$S$ is symplectic means that $S^T J S = J$, with $J = \begin{pmatrix} 0 & \id\\-\id&0\end{pmatrix}$.} 
\[
S  := \begin{pmatrix} S_1 & 0 \\ 0 & S_2 \end{pmatrix}, \quad 
S_1 := (\D+\W-\Wt)^{1/2} E^{-1/2}, \quad  S_2  := (\D+\W-\Wt)^{-1/2} E^{1/2}.\]
The square roots are well-defined thanks to Lemma \ref{lem:inverses}.
 Using $S$ we can symplectically blockdiagonalize $\mathfrak{M}$, i.\,e., 
\[S^T \mathfrak{M} S = \frac{1}{2}\begin{pmatrix} E & 0\\ 0 & E\end{pmatrix}.\]
We define transformed field operators $\tilde\phi$ and $\tilde\pi$ by
\[\begin{pmatrix} \phi \\ \pi\end{pmatrix} = S \begin{pmatrix} \tilde\phi \\ \tilde\pi\end{pmatrix}\,.\]
After a change of basis that diagonalizes $E$ into $\diag(e_\gamma: \gamma \in \Ik)$ we call them $\dbtilde\phi$ and $\dbtilde\pi$. Then
\[\begin{split}\mathbb{H} & = \begin{pmatrix}\tilde\phi^T & \tilde\pi^T \end{pmatrix} \frac{1}{2}\begin{pmatrix} E & 0\\ 0 & E\end{pmatrix} \begin{pmatrix} \tilde\phi \\ \tilde\pi\end{pmatrix} = \begin{pmatrix}\dbtilde\phi^T & \dbtilde\pi^T \end{pmatrix} \frac{1}{2}\begin{pmatrix} \diag(e_\gamma) & 0\\ 0 & \diag(e_\gamma)\end{pmatrix} \begin{pmatrix} \dbtilde\phi \\ \dbtilde\pi\end{pmatrix} \\
& = \sum_{\gamma\in\Ik} \frac{e_\gamma}{2} \left(\dbtilde{\phi}_\gamma^2 + \dbtilde{\pi}_\gamma^2 \right) \geq \sum_{\gamma\in \Ik} \frac{e_\gamma}{2} = \frac{1}{2}\tr E.\end{split}\]
We conclude that the ground state energy of the effective theory at momentum $k$ is
\[\begin{split}\inf \sigma(h_\textnormal{eff}(k)) &= \frac{1}{2} \tr \left( E - (\D+\W) \right) \\ & = \frac{1}{2} \tr E - \frac{1}{2}\sum_{\alpha\in\Ik} u_\alpha^2 - \frac{g}{2}\sum_{\alpha \in \Ikp} v_\alpha(k)^2 - \frac{g}{2}\sum_{\alpha\in\Ikm} v_\alpha(-k)^2\,.\end{split}\]
The minimum is attained by the bosonic Fock space vector $\xi_\textnormal{gs}(k)$ satisfying $\dbtilde{c}_\gamma(k) \xi_\textnormal{gs}(k) = 0$ for all $\gamma \in \Ik$. Since the operators $\tilde{c}(k)$ and $\dbtilde{c}(k)$ are related by a change of one-particle basis, this state is actually the same as the state annihilated by the operators $\tilde{c}_\gamma(k)$ for all $\gamma \in \Ik$.

\subsection{Construction of the Bosonic Ground State}
The construction of the bosonic ground state $\xi_\textnormal{gs}(k)$ follows \cite[Section 5.1]{GS13}. The ground state satisfies  $\tilde{c}_\gamma\xi_\textnormal{gs}(k) = 0$ for all $\gamma \in \Ik$, where $\tilde{c}$ is related to the new Segal field operators $\tilde \phi$, $\tilde \pi$ as in  \eqref{eq:defsegalfields}. We express $c$ and $c^*$ through $\tilde{c}$ and $\tilde{c}^*$, so
\[\begin{pmatrix} c \\ c^* \end{pmatrix} = \Theta \begin{pmatrix} \phi \\ \pi \end{pmatrix} = \Theta S \begin{pmatrix} \tilde\phi \\ \tilde\pi \end{pmatrix} = \Theta S \Theta^{-1} \begin{pmatrix} \tilde c \\ \tilde{c}^* \end{pmatrix}.\]
The relation is through the Bogoliubov map
\[\Vcal := \Theta S \Theta^{-1} = \frac{1}{2} \begin{pmatrix} S_1 + S_2 & S_1-S_2 \\ S_1 - S_2 & S_1 + S_2 \end{pmatrix},\]
or more explicitly, the annihilation and creation operators transform as
\begin{equation}\label{eq:trumpsucks}c = \frac{1}{2} (S_1 + S_2)\tilde{c} + \frac{1}{2} (S_1 - S_2)\tilde{c}^*\,, \quad c^* = \frac{1}{2}(S_1 - S_2)\tilde{c} + \frac{1}{2}(S_1+S_2)\tilde{c}^*\,.\end{equation}

\paragraph{Implementation of Bogoliubov Transformations.} Define the unitary operator
\[T_\lambda := e^{\lambda B}\,, \quad \text{with } \lambda \in \Rbb \text{ and }  
 B := \frac{1}{2}\sum_{\alpha,\beta \in \Ik} K_{\alpha,\beta} c^*_\alpha c^*_\beta - \hc\]
 Notice that, since $c^*_\alpha$ and $c^*_\beta$ commute, only the symmetric part of the matrix $K$ contributes. We also assume $K_{\alpha,\beta} \in \Rbb$. For short we write $T := T_1$. The operator $T$ acts as a Bogoliubov transformation, i.\,e., 
\[T^* c_\gamma T = \sum_{\alpha \in \Ik} \left(\cosh K\right)_{\gamma,\alpha} c_\alpha + \sum_{\alpha \in \Ik} \left(\sinh K \right)_{\gamma,\alpha} c^*_\alpha\,.\]
Since $\cosh K$ is a symmetric matrix but $S_1+S_2$ is not symmetric, it is not possible to pick $K$ such that $\cosh K = \frac{1}{2}(S_1+S_2)$. Instead we choose
\[K := \log \lvert S_1^T\rvert\,.\]
This is well-defined because $\lvert S_1^T\rvert$ is symmetric and strictly positive definite, according to Lemma \ref{lem:inverses}. Furthermore $K$ is real and symmetric, so we obtain
\begin{equation}\label{eq:bogob}T^* c_\gamma T = \sum_{\alpha \in \Ik} \frac{1}{2} \left( \lvert S_1^T\rvert + \lvert S_1 ^T\rvert^{-1} \right)_{\gamma,\alpha} c_\alpha + \sum_{\alpha \in \Ik} \frac{1}{2} \left( \lvert S_1^T \rvert - \lvert S_1^T\rvert^{-1} \right)_{\gamma,\alpha} c^*_\alpha.\end{equation}
Let us introduce the polar decomposition $S_1 = O \lvert S_1\rvert$ with some orthogonal matrix $O$. Then
\[ \frac{1}{2}\left(S_1 + S_2\right) = \frac{1}{2} \left( \lvert S_1^T \rvert + \lvert S_1^T\rvert^{-1} \right)O^T, \quad \frac{1}{2}\left(S_1 - S_2\right) = \frac{1}{2} \left( \lvert S_1^T \rvert - \lvert S_1^T\rvert^{-1} \right)O^T.\]
The orthogonal matrix $O^T$ acts as a change of the one-particle basis, so the vacuum transformed by the Bogoliubov transformation in \eqref{eq:bogob} is the same as the vacuum transformed by the Bogoliubov transformation in \eqref{eq:trumpsucks}.

We conclude that the ground state of the total system is given by
\[\xi_\textnormal{gs} = \bigotimes_{k\in \north} \xi_\textnormal{gs}(k)\,, \quad \xi_\textnormal{gs}(k) = T(k)\Omega\,,\]
where $\Omega$ is the vacuum vector in bosonic Fock space, and we restored the $k$-dependence in the notation. Since operators at different $k$ commute, we can take the tensor product into the exponent as a sum, yielding
\[\xi_\textnormal{gs} = \exp\Big( \sum_{k \in \north}\frac{1}{2}\sum_{\alpha,\beta \in \Ik} K_{\alpha,\beta}(k) c^*_\alpha(k) c^*_\beta(k) - \hc \Big) \Omega\,.\]

\section*{Acknowledgements}
We thank Christian Hainzl and Jan Philip Solovej for helpful discussion. N.~B.\ and R.~S.\ have received funding from the European Research Council (ERC) under the European Union’s Horizon 2020 research and innovation programme (grant agreement 694227). R.~S.\ was also supported by Austrian Science Fund (FWF), project Nr.\ P27533-N27. The work of M.~P.\ has been supported by the Swiss National Science Foundation via the grant ``Mathematical Aspects of Many-Body Quantum Systems''. B.~S.\ gratefully acknowledges support from the NCCR SwissMAP and from the 
Swiss National Science Foundation through the grant ``Dynamical 
and energetic properties of Bose-Einstein condensates''. 

\bibliographystyle{abbrv}
\bibliography{completepatches}{}

\end{document}